%% file: stoc-27.tex
\documentclass[11pt,letterpaper,oneside,onecolumn]{article}
\usepackage{preamble}
\allowdisplaybreaks

\makeatletter
\let\original@fnsymbol\@fnsymbol
\renewcommand{\@fnsymbol}[1]{%
  \ifnum#1=1\relax\ensuremath{\diamond}%
  \else\ifnum#1=2\relax\ensuremath{\star}%
  \else\original@fnsymbol{#1}\fi\fi}
\makeatother

\title{Hardness of Counting Orderings:\\ Taking Gadget Amplification to the Limit}
\author{%
  Marcelo Arenas\thanks{Universidad Cat{\'o}lica de Chile and Instituto Milenio Fundamentos de los Datos, Chile. \url{marenas@uc.cl}} \qquad
  Mar\'ia Alejandra Schild\thanks{Universidad Cat{\'o}lica de Chile and Instituto Milenio Fundamentos de los Datos, Chile. \url{aleschildz@uc.cl}} \qquad
  Bernardo Subercaseaux\thanks{Carnegie Mellon University. \url{bersub@cmu.edu}}
}

\date{}

\begin{document}
\hypersetup{pageanchor=false}
\begin{titlepage}
\maketitle
\thispagestyle{empty}
\begin{abstract}
We introduce a new elementary technique for proving $\sharpP$-hardness results under Turing reductions, and showcase its applicability by settling the counting complexity of several fundamental problems about counting orderings in graphs, posets, and geometric configurations. The technique addresses a natural difficulty in counting reductions: when designing gadgets that attempt to enforce a semantic property, the introduction of the gadgets themselves can affect the function to count in complicated ways.
In a nutshell, we introduce an amplification parameter $q$ for the gadgets (e.g., replacing a subset of the vertices of a graph by $q$ copies) and show that the limiting behavior when $q$ tends to infinity, or even to negative values, filters out the unwanted effects of the gadgets, which allows us to recover the desired counts.  These limits can be computed symbolically via interpolation, after showing that the probability of a random ordering satisfying the problem's constraints is a rational function in $q$.  Concretely, with this methodology, we show $\sharpP$-completeness for counting the following objects: (a) successive vertex orderings, even on bipartite graphs, (b) $st$-numberings of graphs, (c) shellings, even for bipartite graphs, (d) perfect elimination orderings, even for split graphs, (e) linear extensions of N-free posets of height $3$, (f) linear extensions of posets of height $2$, and (g) peeling sequences of point sets.  Result (e) settles a conjecture of Felsner and Manneville (2015), and although result (f) was first proved by Dittmer and Pak (2018), we include an alternative proof, using our technique, that does not rely on the result of Brightwell and Winkler (1991) about hardness of counting for general posets. Result (g) answers a question of Eppstein.
\end{abstract}
\end{titlepage}
\pagenumbering{arabic}
\hypersetup{pageanchor=true}

\section{Introduction}
\label{sec-intro}
\input{sections/new_intro.tex}



\section{Preliminaries}
\label{sec-prelim}

\input{sections/preliminaries}

\section{Proof Strategy}
\label{sec-technique}
\label{sec:strategy}

\input{sections/technique-peo}

\section{Linear Extensions of Height-2 Posets}
\label{sec-height-two}

\input{sections/height-two-posets}

\section{Peeling Sequences}
\label{sec-peeling}

\input{sections/peeling}

\section{Linear Extensions of N-Free Posets}
\label{sec-n-free}

\input{sections/n-free}

\section{Successive Vertex Orderings and Shellings}
\label{sec-svos-shellings}

\input{sections/svos-shellings.tex}

\section{Concluding Remarks}
\label{sec-conclusions}

\input{sections/conclusions}

\section*{Acknowledgments}
\input{sections/acknowledgments}

\bibliographystyle{alphaurl}
\bibliography{references}

\appendix

\input{sections/appendix}





\end{document}

%% file: sections/new_intro.tex
Valiant~\cite{V79,Valiant79} laid the foundations of counting complexity by establishing a striking phenomenon: counting the solutions for two combinatorial problems can be equally hard even when their decision versions (or finding a single solution) have very different complexities.  Establishing such an equivalence requires reductions of a different kind than those used to transfer hardness for decision problems, since unless $\mathsf{P} = \mathsf{NP}$, a reduction from, e.g., $\#\mathsf{SAT}$ to $\#\mathsf{PerfectMatchings}$ cannot simply preserve the number of solutions. Instead, polynomial-time Turing reductions allow one to recover the solution count of the source problem through computation on top of the solution count of potentially many instances of the target problem.  The seminal work of Valiant introduced the usage of polynomial interpolation and the Chinese Remainder Theorem as techniques for recovering solution counts, and in his own words, \emph{``the reductions themselves are characterized by being of an algebraic or arithmetic, rather than combinatorial nature''}~\cite{Valiant79}.

As in standard $\mathsf{NP}$-hardness reductions, Turing reductions for showing $\#\mathsf{P}$-hardness routinely leverage ``gadgets'' to enforce semantic properties in the target instances. The main technical challenge, however, is that the gadgets themselves can affect the function to count in complicated ways.  This challenge is addressed by various techniques such as gadget cancellation~\cite{Curticapean18}, making the unwanted contributions of gadgets vanish modulo some prime~\cite{BrightwellWinkler,dittmer2018counting}, relating polynomials algebraically~\cite{JVW1990}, rational interpolation with partition-function zeros~\cite{SinclairS13}, and even symbolic integration~\cite{SchulmanSS15}. In a nutshell, most of these techniques work not by suppressing the unwanted effects of the gadgets, but rather by showing that their aggregate contribution can be separated from the contribution of the original instances.

In this paper, we develop an elementary method for designing $\#\mathsf{P}$-hardness reductions, which we call \emph{taking gadget amplification to the limit}. We construct families of instances in which a parameter $q$ controls the number of copies of simple gadgets. After a suitable normalization, the contributions of different orderings become rational functions of $q$. The central combinatorial step is to identify the orderings encoding source solutions through their decay at infinity or their pole order at a finite point, and to show that each source solution has the same total surviving contribution. A suitable normalized limit then recovers the desired count. Finite poles are particularly useful when the leading behavior at infinity does not separate the desired orderings from the others. To compute these limits exactly, we establish polynomial degree bounds and explicit common denominators, allowing standard polynomial interpolation to recover the normalized count from polynomially many oracle queries on polynomial-size instances. All oracle queries use positive integer values of $q$. Once the rational function has been reconstructed, however, we can compute normalized limits as $q\to\infty$ or $q\to-\frac{1}{3}$, even though the latter value has no interpretation as a number of copies. 

While the general idea is straightforward, we find its applicability surprising: our original goal was to establish $\#\mathsf{P}$-hardness of counting \emph{shellings} in graphs (described below), but since then we have observed that it can be applied to a range of counting problems (we show $\#\mathsf{P}$-hardness for $7$ different problems in this paper), some of which were explicitly conjectured to be $\#\mathsf{P}$-hard in the literature, such as counting the number of linear extensions of N-free posets of height $3$~\cite{FelsnerManneville}. These applications show how the same method can accommodate different ordering constraints while preserving strong structural restrictions on the target instances.

\subsection{Perfect Elimination Orderings: an Introductory Example}\label{subsec:peo}
\input{sections/peo_example.tex}

\subsection{Problems and Results}\label{subsec:problems}

\input{sections/contributions}

\bigskip

\noindent
{\bf Organization of the Paper:} \Cref{sec-prelim} gives the preliminaries. \Cref{sec-technique} outlines our proof strategy. \Cref{sec-height-two} studies linear extensions of height-$2$ posets. \Cref{sec-peeling} treats peeling sequences. \Cref{sec-n-free} studies linear extensions of N-free posets. \Cref{sec-svos-shellings} treats successive vertex orderings and shellings. \Cref{sec-conclusions} discusses extensions and open questions. Proofs omitted from the main text appear in~Appendix~\ref{app:proofs}.

%% file: sections/peo_example.tex
We first illustrate the technique in its simplest form by proving hardness of counting perfect elimination orderings. We then present the remaining problems and results in~\Cref{subsec:problems}.

\begin{samepage}
Given a finite undirected graph $G=(V,E)$ on $n$ vertices, a \defn{perfect elimination ordering} (PEO) is an ordering $(v_1,\ldots,v_n)$ of $V$ such that for each vertex $v_i$, its neighbors that appear later in the ordering, $N(v_i)\cap\{v_{i+1},\ldots,v_n\}$, form a clique. A vertex $v$ is \defn{simplicial} if its neighborhood is a clique. Therefore, we can think of PEOs as orderings of the vertices that are sequentially ``eliminated'' from the graph, where a vertex can only be eliminated when it is simplicial in the remaining graph. We write $\peo(G)$ for the number of these orderings. Perfect elimination orderings characterize chordal graphs (i.e., graphs where all induced cycles have length $3$): a graph admits a PEO if and only if it is chordal, and this can be recognized in linear time~\cite{RoseTL76}. Notice that, in a chordal graph, the simplicial vertices are exactly the ones that are the first vertex of some PEO. A graph is \defn{split} if its vertex set can be partitioned into an independent set and a clique.
\par
\end{samepage}

\begin{restatable}{thm}{peocounting}\label{thm:peo}
    $\peo(\cdot)$ is $\#\mathsf{P}$-complete, even for split graphs.
\end{restatable}

We reduce from counting linear extensions of posets of height $2$. A \defn{linear extension} is an ordering in which each element follows all its predecessors, and height $2$ means that the longest chain has two elements. Dittmer and Pak~\cite{Dittmer_Pak_2020} proved that this counting problem is $\sharpP$-complete; we give an independent proof in~\Cref{sec-height-two}. Given such a poset $P$, we write $\LE(P)$ for its set of linear extensions; our goal is to compute $\#\mathsf{LE}(P) = |\LE(P)|$ using polynomially many queries for the numbers of PEOs of split graphs constructed from $P$.

We begin with the construction of Cardinal, Doignon, and Merckx~\cite{cardinal2017shelling}. The poset $P$ in~\Cref{fig:peo-example-poset} has lower elements $\ell_1,\ldots,\ell_5$ and upper elements $u_1,\ldots,u_5$. We create a graph vertex for each element, retaining its label, and join $\ell_i$ to $u_j$ whenever $\ell_i\prec u_j$. Adding all edges among the upper vertices gives the intermediate split graph $G$ in~\Cref{fig:peo-example-graph}.
Every linear extension of $P$ gives a PEO of $G$, but $G$ also admits PEOs that violate the order of $P$. For example, both $\ell_1$ and $\ell_2$ must precede $u_2$ in a linear extension. In $G$, these two nonadjacent neighbors prevent $u_2$ from being eliminated while both remain. Once every other upper vertex and $\ell_2$ have been eliminated, however, $u_2$ can be eliminated before $\ell_1$. Consider, for example, the ordering $(\ell_2, \ell_3, \ell_4, \ell_5, u_1, u_3, u_4, u_5, u_2, \ell_1)$. Thus, the intermediate graph does not yet enforce all precedence constraints of $P$.

\begin{figure}[t]
    \centering
    \tikzset{peo source vertex/.style={draw, circle, minimum size=6.2mm,
                                     inner sep=1pt, font=\small}}
    \newcommand{\peoSourceLayers}{%
        \foreach \i in {1,...,5} {
            \node[peo source vertex, fill=blue!10] (l\i) at ({1.8*(\i-1)},0)
                {$\ell_\i$};
            \node[peo source vertex, fill=orange!10] (u\i) at ({1.8*(\i-1)},2)
                {$u_\i$};
        }
    }
    \begin{subfigure}[t]{0.49\linewidth}
        \centering
        \begin{tikzpicture}[scale=0.85]
            \peoSourceLayers
            \begin{scope}[on background layer]
                \foreach \i/\j in {1/2,2/2,3/1,3/3,3/4,4/3,4/5,5/4,5/5}
                    \draw[-Latex] (l\i) -- (u\j);
            \end{scope}
            \path[use as bounding box] (-0.4,-0.4) rectangle (7.6,3.8);
        \end{tikzpicture}
        \caption{The input poset $P$ of height $2$.}
        \label{fig:peo-example-poset}
    \end{subfigure}
    \hfill
    \begin{subfigure}[t]{0.49\linewidth}
        \centering
        \begin{tikzpicture}[scale=0.85]
            \peoSourceLayers
            \begin{scope}[on background layer]
                \foreach \i/\j in {1/2,2/2,3/1,3/3,3/4,4/3,4/5,5/4,5/5}
                    \draw (l\i) -- (u\j);
                \foreach \i/\j in {1/2,2/3,3/4,4/5}
                    \draw (u\i) -- (u\j);
                \foreach \i/\j in {1/3,1/4,1/5,2/4,2/5,3/5}
                    \draw (u\i) to[bend left=45] (u\j);
            \end{scope}
            \path[use as bounding box] (-0.4,-0.4) rectangle (7.6,3.8);
        \end{tikzpicture}
        \caption{The intermediate split graph $G$.}
        \label{fig:peo-example-graph}
    \end{subfigure}
    \caption{The first step of the reduction: retain the incidences between the two layers of $P$ and add all edges within the upper layer to obtain $G$.}
    \label{fig:peo-example}
\end{figure}

To enforce these constraints, add a vertex $\star$ adjacent to every upper vertex and to no lower vertex, as in~\Cref{fig:peo-example-star}. If $\star$ is kept until all original vertices have been eliminated, their order must be a linear extension of $P$. The difficulty is that $\star$ can itself be eliminated at any time, so an oracle for PEOs also counts orderings that remove it earlier.
For each positive integer $q$, we therefore replace $\star$ by vertices $s_1,\ldots,s_q$ that form a clique and are adjacent to every upper vertex and to no lower vertex. With these new vertices, which we call \emph{anchors}, we obtain a new graph $G_q$; see~\Cref{fig:peo-example-copies}. With $P$ fixed, increasing $q$ increases the proportion of PEOs that keep at least one anchor until all original vertices have been eliminated. In every such PEO, the original vertices follow a linear extension of $P$: as long as one anchor remains, an upper vertex can be eliminated only after all its lower predecessors. For any finite $q$, however, all anchors can still disappear earlier, allowing orderings that violate these constraints. After division by the total number of vertex permutations, the contribution of PEOs in which all anchors disappear before the original vertices have been eliminated tends to zero as $q\to\infty$, while each linear extension makes the same nonzero contribution to the limit.

We formalize this idea in three steps: count the possible interleavings of the anchors, reconstruct the normalized count as a rational function using polynomially many oracle queries, and recover $\#\mathsf{LE}(P)$ by computing its limit symbolically.

\begin{figure}[t]
    \centering
    \tikzset{peo vertex/.style={draw, circle, minimum size=6.2mm,
                              inner sep=1pt, font=\small}}
    \newcommand{\peoExampleLayers}{%
        \foreach \i in {1,...,5} {
            \node[peo vertex, fill=blue!10] (l\i) at ({1.8*(\i-1)},0)
                {$\ell_\i$};
            \node[peo vertex, fill=orange!10] (u\i) at ({1.8*(\i-1)},2)
                {$u_\i$};
        }
        \begin{scope}[on background layer]
            \foreach \i/\j in {1/2,2/2,3/1,3/3,3/4,4/3,4/5,5/4,5/5}
                \draw (l\i) -- (u\j);
            \foreach \i/\j in {1/2,2/3,3/4,4/5}
                \draw (u\i) -- (u\j);
            \foreach \i/\j in {1/3,1/4,1/5,2/4,2/5,3/5}
                \draw (u\i) to[bend left=45] (u\j);
        \end{scope}
    }
    \begin{subfigure}[t]{0.49\linewidth}
        \centering
        \begin{tikzpicture}[scale=0.75, transform shape]
            \peoExampleLayers
            \node[peo vertex, fill=green!10] (star) at (3.6,4.8) {$\star$};
            \begin{scope}[on background layer]
                \foreach \i in {1,...,5}
                    \draw[gray!70, thin] (star) -- (u\i);
            \end{scope}
            \path[use as bounding box] (-0.4,-0.4) rectangle (7.6,5.2);
        \end{tikzpicture}
        \caption{A single additional vertex $\star$.}
        \label{fig:peo-example-star}
    \end{subfigure}
    \hfill
    \begin{subfigure}[t]{0.49\linewidth}
        \centering
        \begin{tikzpicture}[scale=0.75, transform shape]
            \peoExampleLayers
            \foreach \j/\x/\y in {
                1/3.30/5.06,2/3.60/5.06,3/3.90/5.06,
                4/3.15/4.80,5/3.45/4.80,6/3.75/4.80,7/4.05/4.80,
                8/3.30/4.54,9/3.60/4.54,10/3.90/4.54}
                \node[draw, circle, fill=green!10, minimum size=2.4mm,
                      inner sep=0pt] (s\j) at (\x,\y) {};
            \node[anchor=west, font=\small] at (4.4,4.8) {$K_q$};
            \begin{scope}[on background layer]
                \foreach \j in {1,...,10}
                    \foreach \i in {1,...,5}
                        \draw[gray!70, thin] (s\j) -- (u\i);
            \end{scope}
            \path[use as bounding box] (-0.4,-0.4) rectangle (7.6,5.2);
        \end{tikzpicture}
        \caption{The anchors form a clique $K_q$.}
        \label{fig:peo-example-copies}
    \end{subfigure}
    \caption{Adding anchors to $G$. In (b), every anchor is adjacent to every upper vertex and to every other anchor; edges within the anchor clique $K_q$ are omitted.}
    \label{fig:peo-example-third-layer}
\end{figure}

 Let $(P, \preceq)$ be a poset of height $2$, and let $n = |P|$ with $n\geq 1$. We will now describe how to calculate $\#\mathsf{LE}(P)$ in polynomial time (in terms of $n$) using queries to $\peo(G)$ for different split graphs $G$. 
Let us write $P = L \sqcup U$, where $L$ is the set of lower elements, and $U$ is the set of upper elements. For each integer $q \geq 1$, consider the graph $G_q = (V_q, E_q)$
constructed above. For clarity, we denote the graph vertices
corresponding to $\ell \in L$ and $u \in U$ by $\tilde{\ell}$ and
$\tilde{u}$, respectively, to distinguish them from the elements of $P$.
We call the vertices $\tilde{\ell}$ \emph{lower vertices} and the
vertices $\tilde{u}$ \emph{upper vertices}, and refer to both types
collectively as \emph{poset vertices}.
Notice that $G_q$ is a split graph, with the lower vertices forming
an independent set and the upper vertices together with the anchors
forming a clique.

\begin{lemma}\label{lemma:anchors_block}
 Let $q \geq 1$ and let $V' \subseteq V_q$ be a subset of vertices that contains at least one anchor. Let $H$ be the subgraph of $G_q$ induced by $V'$. Then, the vertices of $H$ that are not simplicial are exactly the vertices $\tilde{u} \in V'$ for which there is another vertex $\tilde{\ell} \in V'$ for which $\ell \prec u$.
\end{lemma}

\begin{proof}
Notice that the lower vertices and the anchors are always simplicial, so we only need to prove the equivalence for upper vertices. Let us consider an upper vertex $\tilde{u} \in V'$. We also know that there exists an anchor $s \in V'$.

First, suppose that $\tilde{u}$ is simplicial. Because $s$ is a neighbor of $\tilde{u}$ in $H$, every other neighbor of $\tilde{u}$ in $H$ must be adjacent to $s$. Because no lower vertex is connected to $s$, we conclude that there is no $\tilde{\ell} \in V'$ for which $\ell \prec u$. Conversely, suppose that there are no lower vertices in $V'$ connected to $\tilde{u}$. Then, all neighbors of $\tilde{u}$ belong to the clique, and hence $\tilde{u}$ is simplicial.
\end{proof}

Given an ordering $\sigma$ of the elements of $P$, and an integer $m \in \{0, 1, \dots, n\}$, let $\peo(G_q)_{\sigma, m}$ be the number of perfect elimination orderings of $G_q$ for which the following conditions hold:
(i) the induced order on the poset vertices corresponds to $\sigma$; and (ii)
the number of poset vertices eliminated before all anchors are eliminated is exactly $m$.
Let $\ompeo$ be the set of pairs $(\sigma, m)$, with $\sigma$ an ordering of $P$ and $m \in \{0, 1, \dots, n\}$, for which $\peo(G_q)_{\sigma, m} > 0$. Since anchors are always simplicial, \Cref{lemma:anchors_block} implies that whether an ordering of $V_q$ is a PEO depends only on its pair $(\sigma,m)$; in particular, $\ompeo$ does not depend on $q$.
We now give an explicit formula for these counts, normalized as probabilities under a uniformly random vertex ordering.

\begin{lemma}\label{lemma:peos-formula}
Let $(\sigma, m) \in \ompeo$. Then,
\[
\chi_{\sigma, m}(q) \coloneqq \frac{\peo(G_q)_{\sigma, m}}{(n+q)!} = \frac{q}{m! \prod_{j = m}^{n}{(q+j)}}.
\]
\end{lemma}

\begin{proof}
Notice that for a perfect elimination ordering to satisfy the conditions given by the pair $(\sigma, m)$, the ordering of the poset vertices, as well as the relative position (with respect to those poset vertices) of the last anchor, are prescribed. Now, there are $\binom{m+q-1}{q-1}$ ways to insert the other anchors. We also need to consider the possible orderings of the anchors, and so we need to multiply by an additional $q!$ factor. Therefore,
\[
\frac{\peo(G_q)_{\sigma, m}}{(n+q)!} = \frac{q!}{(n+q)!} \binom{m+q-1}{q-1} = \frac{q(m+q-1)!}{m! (n+q)!} = \frac{q}{m! \prod_{j = m}^{n}{(q+j)}}.\qedhere
\]
\end{proof}

\noindent \Cref{lemma:peos-formula} gives us a rational function in $q$ for each pair $(\sigma, m) \in \ompeo$. Thus,
\[
    \frac{\peo(G_q)}{(n+q)!} =  \sum_{(\sigma, m) \in \ompeo} \chi_{\sigma, m}(q).
\]

We now show that this sum can be recovered in polynomial time from oracle evaluations.

\begin{proposition}\label{prop:peos_interpolation}
   Given an oracle for $\peo(\cdot)$, we can compute in polynomial time (in terms of $n$) two polynomials $R$ and $D$ in canonical form such that the following equality between rational functions holds:
   \[ \sum_{(\sigma, m) \in \ompeo} \chi_{\sigma, m}(q) = \frac{R(q)}{D(q)}.
    \]
\end{proposition}

\begin{proof}
Observe that the polynomial
\(
    D(q) := n! \prod_{j = 0}^{n}{(q+j)}
\)
is a multiple of the denominator of $\chi_{\sigma, m}(q)$ for every $(\sigma, m) \in \ompeo$. Moreover, $D(q)$ has degree $n+1$. Therefore,
\[
    \sum_{(\sigma, m) \in \ompeo} D(q) \cdot \chi_{\sigma, m}(q) = D(q) \cdot \sum_{(\sigma, m) \in \ompeo} \chi_{\sigma, m}(q)
\] is also a polynomial, $R(q)$, of degree at most $n+1$, since each summand has degree $(n+ 1) + 1 - (n-m+1) \leq n+1$. Thus, we can obtain its coefficients by first evaluating $\frac{\peo(G_q)}{(n+q)!}$ at the points $1, 2, \ldots, n+2$ (via the oracle) and $D$ at the same points (directly, since we know $D$ explicitly), then applying Lagrange interpolation, and finally putting the result into canonical form. All these steps can be done in polynomial time in terms of $n$.
Notice that $D$ can also be put into canonical form in polynomial time.
\end{proof}

The next proposition shows how to recover $\#\mathsf{LE}(P)$ from the rational function.

\begin{proposition}\label{prop:peos_limit}
    We have that
    \[
        \lim_{q \to \infty} \sum_{(\sigma, m) \in \ompeo} \chi_{\sigma, m}(q) = \frac{\#\mathsf{LE}(P)}{n!}.
    \]
\end{proposition}

\begin{proof}
As a polynomial in $q$, the denominator of
$\chi_{\sigma,m}(q)$ has degree $n-m+1$,
whereas its numerator has degree $1$.
Thus, for every $(\sigma,m)\in\ompeo$ with $m<n$,
we have
\(
    \lim_{q\to\infty}\chi_{\sigma,m}(q)=0.
\)

We claim that $(\sigma,n)\in\ompeo$ if and only if
$\sigma$ is a linear extension of $P$. First, suppose that $(\sigma,n)\in\ompeo$.
At least one anchor remains present throughout the elimination
of the poset vertices. By \Cref{lemma:anchors_block}, each
poset vertex eliminated must correspond to a minimal element
of the remaining poset. Thus, $\sigma$ is a linear extension
of $P$.
Conversely, suppose that $\sigma$ is a linear extension of $P$.
We can eliminate the poset vertices in the order prescribed
by $\sigma$, keeping all anchors until the end.
At each step, the next poset element is minimal among those
remaining, so its vertex is simplicial by
\Cref{lemma:anchors_block}.
The remaining anchors form a clique and can therefore be
eliminated in any order. This gives a perfect elimination
ordering with $(\sigma,n)\in\ompeo$.

Since $\ompeo$ is finite and does not depend on $q$, we obtain
\begin{multline*}
\lim_{q\to\infty}
\sum_{(\sigma,m)\in\ompeo}\chi_{\sigma,m}(q)
=
\sum_{(\sigma,m)\in\ompeo}
\lim_{q\to\infty}\chi_{\sigma,m}(q)
=
\sum_{\sigma\in\LE(P)}
\lim_{q\to\infty}\chi_{\sigma,n}(q)
\ = \\
\sum_{\sigma\in\LE(P)}
\lim_{q\to\infty}\frac{q}{n!(q+n)}
=
\sum_{\sigma\in\LE(P)}\frac{1}{n!}
=
\frac{\#\mathsf{LE}(P)}{n!}.\qedhere
\end{multline*}
\end{proof}

With these ingredients, the proof of the theorem is straightforward.
\begin{proof}[Proof of~\Cref{thm:peo}]
Membership in $\#\mathsf{P}$ is immediate.
For $\#\mathsf{P}$-hardness, we now have a Turing reduction from $\#\mathsf{LE}$ on posets of height~$2$. Given a poset $P$ of height $2$, by \Cref{prop:peos_interpolation} and \Cref{prop:peos_limit}, we can compute in polynomial time two polynomials $R$ and $D$ such that
    \[
        \#\mathsf{LE}(P) = n! \cdot \lim_{q \to \infty} \frac{R(q)}{D(q)}.
    \]
Since $\frac{R(q)}{D(q)}$ is a ratio of polynomials, of size polynomial in $|P|$, the limit can be computed symbolically in polynomial time, as it corresponds simply to the ratio of the leading coefficients.
All oracle queries are made on the graphs $G_q$, which are split graphs for every positive integer $q$.
\end{proof}

%% file: sections/contributions.tex
We now describe the remaining problems and results in this paper, in the order we will prove them. They all concern the counting of orderings subject to constraints imposed by graphs, posets, and planar point sets. Depending on the problem, a valid ordering builds the structure incrementally or removes its elements while preserving a required property.  All hardness results below use polynomial-time Turing reductions.

In general, the \defn{height} of a poset is the maximum number of elements in a chain. Brightwell and Winkler~\cite{BrightwellWinkler} proved that $\#\mathsf{LE}$ is $\sharpP$-hard for general posets and asked whether this persists at height $2$, which Dittmer and Pak~\cite{Dittmer_Pak_2020} confirmed. In~\Cref{sec-height-two}, we give an alternative, elementary proof of their result by a direct reduction from $\#\mathsf{X3C}$, the problem of counting exact covers by $3$-sets (see~\Cref{sec-prelim}), without relying on hardness for general posets.

\begin{restatable*}{thm}{heighttwocounting}\label{thm:h2le}
$\#\mathsf{LE}$ is $\#\mathsf{P}$-complete, even for height-$2$ posets.
\end{restatable*}

We next consider a problem of a different nature, defined by planar geometry. The input is a finite set $\points\subseteq\mathbb{Q}^2$, specified by binary-encoded rational coordinates and in \emph{general position}, meaning that no three points are collinear. A \defn{peeling sequence} of $\points$ is an ordering $(z_1,\ldots,z_{|\points|})$ of its points in which each $z_i$ is a vertex of the convex hull of $\{z_i,\ldots,z_{|\points|}\}$. The function $\peel(\points)$ counts these orderings. Every ordering is valid in convex position, whereas interior points are initially unavailable and restrict the possible removal orderings. This variation has motivated the study of the minimum number of peeling sequences over point sets of a given size. Dumitrescu~\cite{Dumitrescu2022Peeling} investigated this question in the plane, and Dumitrescu and T\'oth~\cite{DumitrescuToth2025Peeling} proved exponential upper and lower bounds for such a minimum, in contrast with the factorial maximum. Simon~\cite{Simon2026Peeling} subsequently improved the bounds for higher-dimensional point sets. In~\Cref{sec-peeling}, we determine the complexity of exact counting of such peelings,
answering a question of Eppstein.\footnote{David Eppstein, personal communication, 2026. We were informed of this question by Igor Pak, also through personal communication.}

\begin{restatable*}{thm}{peelingcounting}\label{thm:peeling}
    $\peel(\cdot)$ is $\#\mathsf{P}$-complete.
\end{restatable*}

Returning to linear extensions, a poset is \defn{N-free} if its Hasse diagram contains no induced copy of the following four-element pattern:
\begin{center}
\begin{tikzpicture}
\node[draw, circle, fill=black, inner sep=2pt] (a) at (0,0) {};

\node[draw, circle, fill=black, inner sep=2pt] (b) at (1,0) {};

\node[draw, circle, fill=black, inner sep=2pt] (c) at (0,-1) {};

\node[draw, circle, fill=black, inner sep=2pt] (d) at (1,-1) {};

\draw[-Latex] (d) -- (b);
\draw[-Latex] (d) -- (a);
\draw[-Latex] (c) -- (a);

\end{tikzpicture}
\end{center}
This class was first studied by Grillet~\cite{Grillet1969}, who showed that, for finite posets, the condition is equivalent to requiring every maximal chain to meet every maximal antichain. The condition has also motivated work on the structure and algorithms of posets. Habib and J\'egou~\cite{HabibJegou1985} gave a recursive construction that generalizes series-parallel posets, together with a linear-time recognition and decomposition algorithm. Some optimization problems become tractable on this class: the minimum number of consecutive incomparable pairs in a linear extension can be computed in polynomial time, although scheduling under N-free precedence constraints can remain $\mathsf{NP}$-hard~\cite{HabibMohring1987}. Its structure also makes counting linear extensions tractable at height $2$~\cite{FelsnerManneville}. Felsner and Manneville~\cite{FelsnerManneville} left the general counting problem open and conjectured $\sharpP$-completeness, suggesting that hardness should already hold at height $3$. This question was also recorded in the survey of Chan and Pak~\cite[Section~12.1]{LE_survery}. We settle it in~\Cref{sec-n-free}, showing that increasing the height from $2$ to $3$ suffices for hardness within this class.

\begin{restatable*}{thm}{nfreecounting}\label{thm:nfree-height-three}\label{thm:nfree}
Given an N-free poset of height $3$, it is $\#\mathsf{P}$-complete to count its number of linear extensions.
\end{restatable*}

We then turn to orderings governed by connectivity. For a finite undirected graph $G=(V,E)$ on $n$ vertices, a \defn{successive vertex ordering} (SVO) is an ordering $(v_1,\ldots,v_n)$ of $V$ such that every prefix induces a connected subgraph. Equivalently, every vertex after the first has an earlier neighbor. We write $\#\mathsf{SVO}(G)$ for their number. This counting question has motivated explicit formulas for structured graph classes and algorithms for general graphs. \cite{FANG2023105776} obtained a formula for fully regular graphs, a class whose independent sets satisfy a uniform neighborhood condition. More recently, Agrawal, Erturk, and Louis~\cite{agrawal2026successivevertexorderingsconnected} gave a formula for every connected graph using inclusion--exclusion over independent sets, together with an algorithm using $O(n2^n)$ arithmetic operations. These results raise the question of whether exact counting admits a polynomial-time algorithm on general graphs. In~\Cref{sec-svos-shellings}, we establish the following complexity result.

\begin{restatable*}{thm}{svocounting}\label{thm:svo-bipartite}\label{thm:main}
$\#\mathsf{SVO}$ is $\sharpP$-complete, with hardness already for bipartite graphs.
\end{restatable*}

A related notion is that of an \defn{$st$-numbering}. Given a graph $G$ and distinct vertices $s,t$, this is an ordering of its vertices with $s$ first and $t$ last, in which every other vertex has both an earlier neighbor and a later neighbor. Such orderings impose the neighbor condition in both directions and have applications in computational geometry~\cite{Papamanthou}. Our reduction for rooted SVOs, whose first vertex is prescribed, also yields the following consequence.

\begin{restatable*}{corollary}{stnumberingcounting}\label{cor:st-numberings}
    Given $G$ and vertices $s$ and $t$, it is $\#\mathsf{P}$-complete to count the $st$-numberings of $G$.
\end{restatable*}

For a finite undirected graph $G=(V,E)$, a \defn{shelling} is an ordering of $E$ in which every edge after the first shares an endpoint with an earlier edge. Equivalently, every nonempty prefix forms a connected subgraph on the endpoints of its edges. We write $\#\mathsf{Shellings}(G)$ for the number of shellings. Gao and Peng~\cite{gaoCountingShellingsComplete2021} obtained exact formulas for complete bipartite graphs and related the count for trees to linear extensions of tree posets. Schild~\cite{masterAle} proved $\sharpP$-hardness for counting shellings that induce a prescribed spanning tree. Each shelling is an SVO of the line graph, but hardness for general SVOs does not imply hardness on line graphs. We resolve the general problem of counting shellings in~\Cref{sec-svos-shellings}.

\begin{restatable*}{thm}{shellingcounting}\label{thm:shellings-bipartite}\label{thm:shellings}
$\#\mathsf{Shellings}$ is $\sharpP$-complete, with hardness already for bipartite graphs.
\end{restatable*}

Graph shellings are the dimension-$1$ case of shellings of pure simplicial complexes, whose facets all have the same dimension. In general, a shelling adds facets one at a time, with each facet after the first meeting the earlier ones in a nonempty union of codimension-one faces; for graphs, this is precisely the requirement that each new edge share an endpoint with an earlier edge. Goaoc, Pat\'ak, Pat\'akov\'a, Tancer, and Wagner~\cite{Shellability} proved that deciding shellability is $\mathsf{NP}$-complete for pure simplicial complexes in every fixed dimension $d\geq 2$. Their reduction establishes an equivalence of existence, but does not recover the number of satisfying assignments from the number of shellings, so it does not directly establish $\sharpP$-hardness. Our graph result implies $\sharpP$-hardness in every fixed dimension $d\geq 2$ by taking cones: adjoining the same $d-1$ new vertices to every edge produces the facets of a pure $d$-dimensional complex, and an ordering of these facets is a shelling exactly when the corresponding edge ordering is a graph shelling. This polynomial-size construction therefore preserves the number of shellings exactly. To the best of our knowledge, this counting hardness in dimensions $d\geq 2$ was not known prior to this work.

%% file: sections/preliminaries.tex
Given a positive integer $n$, we write $[n]$ to denote the set $\{1, 2, \dots, n\}$.
For a finite set $X$, an \emph{ordering} of $X$ is a sequence $(x_1, \dots, x_{|X|})$ in which every element of $X$ appears exactly once.
We denote by $X!$ the set of all orderings of $X$. For $\sigma \in X!$, we write $\sigma(i)$ for its $i$-th element, and $\sigma^{-1}(x)$ for the \emph{rank} of $x$ in $\sigma$.

Several of our reductions start from $\#\mathsf{X3C}$, the problem of counting exact covers by $3$-sets. An instance $I=(\mathcal{S},U)$ consists of a universe $U=\{u_1,\ldots,u_{3h}\}$ and a family $\mathcal{S}=\{S_1,\ldots,S_n\}$ of $3$-element subsets of $U$. An \emph{exact cover} of $I$ is a subfamily $X\subseteq\mathcal{S}$ with $|X|=h$ whose union is $U$. Since $|U|=3h$, the sets of an exact cover are pairwise disjoint. We write $\mathsf{X3C}(I)$ for the set of exact covers of $I$ and $\#\mathsf{X3C}(I) = |\mathsf{X3C}(I)|$ for their number. The problem $\#\mathsf{X3C}$ is $\#\mathsf{P}$-complete under parsimonious reductions~\cite{X3C}.
\begin{ex}\label{ex:x3c}For $h=2$, consider a universe $U := \{u_1, \ldots, u_6\}$, and $n=4$ sets $S_1 = \{u_1, u_2, u_4\}$, $S_2 = \{u_3, u_4, u_6\}$, $S_3 = \{u_1, u_2, u_5\}$, and $S_4 = \{u_2, u_5, u_6\}$.
Then, $\{S_2, S_3\}$ is an exact cover.
\end{ex}

In our reductions from $\#\mathsf{X3C}$, we assume that every element of $U$ belongs to at least one set in $\mathcal{S}$ and that $h<n$. Instances with an uncovered element or $h>n$ have no exact covers. If $h=n$, the only candidate subfamily is $\mathcal{S}$ itself, which can be checked directly.

All our reductions from $\#\mathsf{X3C}$ use the following classical hook-length formula for linear extensions of rooted tree posets (see, e.g.,~\cite[Proposition~1]{gaoCountingShellingsComplete2021}). For the sake of completeness, a proof of this lemma is included in Appendix~\ref{subsec:hook_length}.

\begin{lemma}\label{lemma:tree-le}
    Let $T := (V, \preceq)$ be a poset whose Hasse diagram is a rooted tree.\footnote{The root is the minimum element, and the order increases along paths away from it.} For each element $v \in V$, let $T_v := (\{u \in V \mid v \preceq u\}, \preceq)$ be the poset induced by the subtree rooted at $v$.  Then,
    \[
    \#\mathsf{LE}(T) = \frac{|T|!}{\prod_{v \in V} |T_v|}.
    \]
\end{lemma}

%% file: sections/technique-peo.tex
The example in~\Cref{subsec:peo} illustrates the following general scheme.
Our aim is to prove $\sharpP$-hardness for some counting problem $\Pi$ about orderings by a Turing reduction from another function $F$ which is already known to be $\sharpP$-hard. Given an instance $I$ of $F$, we construct a sequence of instances of $\Pi$, say $\{X_{q}\}_{q \geq 1}$, in which $q$ controls gadget amplification. This construction must take time polynomial in $|I|+q$.

We partition the valid orderings into classes indexed by a finite set $\Omega$ of \emph{profiles}, which is independent of $q$. A profile records the order of the original elements and, when needed, additional information about auxiliary events that affect the admissible interleavings. For example, in~\Cref{subsec:peo}, the auxiliary event is the elimination of the last anchor: the profile $(\sigma,m)$ records the order $\sigma$ of the poset vertices and the number $m$ of these vertices eliminated before that event.
Let $c_\omega(q)$ count the \emph{valid} orderings (where the validity conditions are according to the function $\Pi$) with profile $\omega$, and let $N(q)$ count the elements being ordered. Define
\[
    \chi_\omega(q) \coloneqq \frac{c_\omega(q)}{N(q)!},
    \qquad
    \Phi(q) \coloneqq \sum_{\omega\in\Omega}\chi_\omega(q).
\]
The number $\chi_\omega(q)$ can be interpreted as the probability that a uniformly
random ordering is valid and has profile $\omega$. We next count the admissible interleavings of the gadgets to express each $\chi_\omega$ as a rational function. The key property is that, up to constant factors, all denominators are products of polynomials from a common family computable in polynomial time from $I$, with known polynomial bounds on their multiplicities. One way to obtain such expressions is to represent the admissible interleavings as linear extensions of a tree poset and apply the hook-length formula (\Cref{lemma:tree-le}). Factorial normalization then leaves a product of reciprocal hook lengths, which in our constructions are affine in $q$. The product of the polynomials in the family, each raised to its multiplicity bound, gives a common denominator $D$ without enumerating $\Omega$. We require $D$ to be nonzero at positive integers, and its degree and the bit length of each coefficient must be bounded by a polynomial in $|I|$. Then $R \coloneqq D \cdot \Phi$ is a polynomial. Since $0\leq \Phi(q)\leq1$ at positive integers, we have that $\deg R\leq\deg D$.

The key combinatorial step is to identify solution profiles through their decay at $q_\star=\infty$ or their pole order at a rational $q_\star$. We also need a normalization factor: choose $g(q)=q^k$ in the former case and $g(q)=(q-q_\star)^k$ in the latter, for a nonnegative integer $k$. We then prove that each limit $\lambda_\omega\coloneqq \lim_{q\to q_\star}g(q)\cdot \chi_\omega(q)$ is finite and vanishes for unwanted profiles. We group the surviving profiles by the source solution they encode. We then prove that every source solution receives the same total weight $\beta\neq0$, obtained by summing the weights $\lambda_\omega$ of all profiles encoding it. Since $\Omega$ is finite and
independent of $q$, this gives
\[
    \lim_{q\to q_\star}g(q) \cdot \Phi(q)=\beta \cdot F(I).
\]
Finally, we set $\delta=\deg D$ and interpolate $R$ from the values $D(q) \cdot \Phi(q)$ at $q=1,\ldots,\delta+1$, using an oracle for $\Pi$ to evaluate $\Phi$. We then compute the normalized limit symbolically and divide by $\beta$ to recover $F(I)$. We require $q_\star$, $k$, and $\beta$ to be computable in polynomial time, with $k$ polynomially bounded. Under these conditions, we obtain a polynomial-time Turing reduction from $F$ to $\Pi$, without enumerating the profiles or estimating the rate of convergence.

%% file: sections/height-two-posets.tex
The elements of a poset of height $2$ can be arranged in two layers, with all strict comparabilities going from the lower layer to the upper layer. We give a direct reduction from $\#\mathsf{X3C}$. The proof of Dittmer and Pak~\cite{Dittmer_Pak_2020} reduces from counting linear extensions of general posets, using counts modulo primes and the Chinese Remainder Theorem. Here, making copies of elements and applying the hook-length formula suffice to express the normalized count as a rational function. We then apply the interpolation technique of~\Cref{sec:strategy} and compute a limit symbolically.

The choice of limit is essential. For perfect elimination orderings in~\Cref{subsec:peo}, the desired count is isolated by the behavior at infinity. Here, the leading behavior as the amplification parameter grows does not distinguish the orderings encoding exact covers from the others. Interpolation recovers the entire normalized count as a rational function, allowing us to examine it near a finite point instead. At a suitable negative value, the terms encoding exact covers have a stronger singularity, so an appropriate normalization isolates their contributions.

The proof has three steps. We first construct a poset by making $q$ copies of each universe element. Fixing the order of the set elements then gives a tree poset, so the hook-length formula yields an explicit rational expression for its normalized count. Next, a common denominator of polynomial degree allows us to recover the sum of these expressions from polynomially many oracle queries. Finally, we compute a normalized limit at a negative value of $q$ to recover the number of exact covers.

Let $I = (\mathcal{S}, U)$ be an instance of $\#\mathsf{X3C}$, with $|\mathcal{S}|=n$ and $|U|=3h$, using the notation and assumptions of~\Cref{sec-prelim}. For each integer $q \geq 1$, construct a poset $P_{I, q}$ of height $2$ as follows:
\begin{enumerate}
    \item Create an element $\ns_i$ of the first layer for each $S_i \in \mathcal{S}$.
    \item For each element $u_j \in U$, create $q$ copies $u_j^{(1)}, \dots, u_j^{(q)}$ and add all the relations $\ns_i \prec u_j^{(k)}$ such that $u_j \in S_i$.
\end{enumerate}

\begin{figure}[!t]
    \begin{subfigure}[t]{0.49\linewidth}
        \centering
        \input{figures/height-two-instance}
        \caption{Instance $I$ from~\Cref{ex:x3c}, exact cover in red.}\label{fig:height-two-instance}
    \end{subfigure}
    \hfill
    \begin{subfigure}[t]{0.49\linewidth}
        \centering
        \input{figures/height-two-reduction}
        \caption{$P_{I,2}$, with the copies of each $u_j$ in a dotted group.}\label{fig:height-two-poset}
    \end{subfigure}
    \caption{The reduction from $\#\mathsf{X3C}$ to counting linear extensions of height-$2$ posets, for $q=2$. Each thick gray arrow from $\ns_i$ to the group for $u_j$ represents the two relations $\ns_i \prec u_j^{(k)}$, for $k\in\{1,2\}$.}\label{fig:height-two-reduction}
\end{figure}

\begin{samepage}
In the resulting poset, illustrated in~\Cref{fig:height-two-reduction} for $q=2$, each copy of a universe element must follow all first-layer elements corresponding to sets that contain it. Fixing the order of the first layer makes these constraints equivalent to requiring the copy to follow only the last such element. This gives a tree poset whose linear extensions can be counted using the hook-length formula.

\end{samepage}

For a fixed $\sigma \in [n]!$, let $\#\mathsf{LE}(P_{I, q})_{\sigma}$ count the linear extensions whose first-layer elements appear in the order $\ns_{\sigma(1)}, \ldots, \ns_{\sigma(n)}$. The next lemma expresses this count in terms of the last set covering each universe element. As in~\Cref{subsec:peo}, we introduce a family of normalized functions.

\begin{lemma}\label{lemma:le-formula}
Let $\sigma \in [n]!$. For each $j \in [n]$, let $c_j(\sigma)$ be the number of elements of $U$ whose last covering set is $S_{\sigma(j)}$, that is, those that are covered by $S_{\sigma(j)}$ but are not covered by $S_{\sigma(i)}$ for any $i>j$. Then
\[
    \chi_{\sigma}(q) \coloneqq \frac{\#\mathsf{LE}(P_{I, q})_{\sigma}}{(n+3hq)!} =  \frac{1}{\prod_{i=1}^n \left(i + \sum_{j=1}^{i} q\cdot c_{n-j+1}(\sigma)\right)}.
    \]
\end{lemma}
Summing over the possible orderings of the first layer, \Cref{lemma:le-formula} gives
\[
        \frac{\#\mathsf{LE}(P_{I, q})}{(n+3hq)!} = \sum_{\sigma \in [n]!} \chi_{\sigma}(q).
\]
Although this sum has $n!$ terms, we can recover it without enumerating them. As in~\Cref{prop:peos_interpolation}, we use an explicit common denominator of polynomial degree and interpolate the numerator. Each oracle query gives the value of the whole sum at a positive integer $q$. The next proposition carries out this step using only the height-$2$ posets $P_{I,q}$.

\begin{samepage}
\begin{proposition}\label{prop:le_compute-polynomials}
   Given an instance $I$ of $\#\mathsf{X3C}$, and an oracle for $\#\mathsf{LE}(\cdot)$, we can compute in polynomial time (in terms of the size of the instance $I$) two polynomials $R$ and $D$ in canonical form such that the following equality between rational functions holds:
   \[ \sum_{\sigma \in [n]!} \chi_{\sigma}(q) = \frac{R(q)}{D(q)}.
    \]
\end{proposition}
\end{samepage}

\begin{proof}
The argument is analogous to the proof of \Cref{prop:peos_interpolation}.

Note that all terms $i + \sum_{j=1}^i q \cdot c_{n-j+1}(\sigma)$ are of the form $Cq + \ell$ for integers $0 \leq C \leq |U|$ and $1 \leq \ell \leq n$. Thus,
the polynomial
\[
    D(q) := \prod_{C = 0}^{|U|} \prod_{\ell=1}^{n} (Cq + \ell)
\]
is divisible by the polynomial $\frac{1}{\chi_\sigma(q)}$ for every $\sigma$, making $D(q) \cdot \chi_\sigma(q)$ a polynomial of degree at most $\deg D \leq |U|n$. Therefore,
\[
    \sum_{\sigma \in [n]!} D(q)\cdot \chi_\sigma(q) = D(q) \cdot \sum_{\sigma \in [n]!} \chi_\sigma(q)
\] is also a polynomial, $R(q)$, of degree at most $|U| \cdot n$.  We can obtain its coefficients by evaluating $\frac{\#\mathsf{LE}(P_{I, q})}{(n+3hq)!}$ at the points $1, 2, \ldots, |U|n + 1$ (via the oracle) and $D$ at the same points (directly), then applying Lagrange interpolation, and finally putting the result into canonical form.
\end{proof}

For every ordering $\sigma$, each denominator factor of $\chi_{\sigma}(q)$ has a positive coefficient of $q$, since the last set alone covers three universe elements. Thus all terms decay as $q^{-n}$, and multiplying by $q^n$ before taking the limit at infinity retains a nonzero contribution from every ordering. Interpolation recovers the information needed at a finite point: after multiplication by $D$, the degree bound ensures that the sampled values uniquely determine $R$. The equality in~\Cref{prop:le_compute-polynomials} therefore holds as an identity of rational functions, so it determines the behavior near $q=-1/3$. All oracle queries use positive integers; the limit at this negative value is computed algebraically from $R/D$.

At $q=-1/3$, the pole order distinguishes the desired orderings. In the denominator factor indexed by $i$, the coefficient of $q$ counts the universe elements covered by the last $i$ sets of $\sigma$. This number is at most $3i$, with equality exactly when those sets are pairwise disjoint. The factor vanishes precisely in this case, when it equals $i(3q+1)$. Since only $3h$ universe elements can be covered, no factor with $i>h$ vanishes. Thus the pole has order at most $h$, and has order $h$ exactly when the last $h$ sets form an exact cover.

Multiplication by $(3q+1)^h$ cancels the poles of order $h$ and makes all terms with smaller pole order tend to zero. Each surviving term has limit $1/(h!(n-h)!)$, and each exact cover corresponds to $h!(n-h)!$ orderings. Hence, each exact cover contributes exactly one to the normalized limit, as the following proposition proves.

\begin{proposition}\label{prop:le_limit}
    For every instance $I = (\mathcal{S}, U)$ of $\#\mathsf{X3C}$ with $|U| = 3h$ it holds that
    \[
        \lim_{q \to -\frac{1}{3}} (3q+1)^h  \cdot \sum_{\sigma \in [n]!} \chi_{\sigma}(q) = \#\mathsf{X3C}(I).
    \]
\end{proposition}

\begin{proof}
Let $C_{i, \sigma} := \sum_{j=1}^i c_{n-j+1}(\sigma)$ be the number of elements whose last covering set appears within the last $i$ sets according to $\sigma$. Thus,
\[
    \chi_{\sigma}(q) =  \prod_{i=1}^n \frac{1}{i + q C_{i, \sigma} }.
\]
Let $\sigma$ be an ordering such that its last $h$ sets are an exact cover. Then, the last $h$ sets are all pairwise disjoint, and thus $C_{i, \sigma} = 3i$ for $1 \leq i \leq h$, and for $i > h$ we simply have $C_{i, \sigma} = 3h$.
Then,
 \begin{multline*}
     \lim_{q \to -\frac{1}{3}} (3q+1)^h\chi_{\sigma}(q) \ = \ \lim_{q \to -\frac{1}{3}}  (3q+1)^h \prod_{i=1}^{h} \frac{1}{i + q\cdot 3i }  \cdot \prod_{i=h+1}^{n} \frac{1}{i + q\cdot 3h }\  = \\
    \lim_{q \to -\frac{1}{3}}   \prod_{i=1}^{h} \frac{(3q+1)}{i(3q+1)}  \cdot \prod_{i=h+1}^{n} \frac{1}{i + q\cdot 3h }
     \ = \ \frac{1}{h!} \lim_{q \to -\frac{1}{3}}  \prod_{i=h+1}^{n} \frac{1}{i + q\cdot 3h } \ = \\
     \frac{1}{h!} \prod_{i=h+1}^{n} \lim_{q \to -\frac{1}{3}} \frac{1}{i + q\cdot 3h }
     \ = \ \frac{1}{h!} \cdot \prod_{i=h+1}^{n} \frac{1}{i - h }
     \ = \ \frac{1}{h! (n-h)!}.
 \end{multline*}
 Now let $\sigma$ be an ordering such that its last $h$ sets are \emph{not} an exact cover. Then, define $B_\sigma := \{ i \in [h] \mid C_{i, \sigma} < 3i \}$. The fact that the last $h$ sets of $\sigma$ are not an exact cover implies that $|B_\sigma| > 0$.
 Thus,
\begin{align*}
      \lim_{q \to -\frac{1}{3}} (3q+1)^h\chi_{\sigma}(q) &= \lim_{q \to -\frac{1}{3}}  (3q+1)^h \prod_{i \in B_\sigma} \frac{1}{i + q \cdot C_{i, \sigma}}  \cdot \prod_{i \in [h] \setminus B_\sigma} \frac{1}{i + q \cdot 3i}  \cdot \prod_{i=h+1}^{n} \frac{1}{i + q\cdot C_{i, \sigma} }\\
      &= \lim_{q \to -\frac{1}{3}}  (3q+1)^{|B_\sigma|} \prod_{i \in B_\sigma} \frac{1}{i + q \cdot C_{i, \sigma}}  \cdot \prod_{i \in [h] \setminus B_\sigma} \frac{(3q+1)}{i(3q+1)}  \cdot \prod_{i=h+1}^{n} \frac{1}{i + q\cdot C_{i, \sigma} }\\
        &= \left(\prod_{i \in [h] \setminus B_\sigma} \frac{1}{i}\right) \cdot \lim_{q \to -\frac{1}{3}} (3q+1)^{|B_\sigma|} \prod_{i \in B_\sigma} \frac{1}{i + q \cdot C_{i, \sigma}}  \cdot \prod_{i=h+1}^{n} \frac{1}{i + q\cdot C_{i, \sigma} } \\
      &= \left(\prod_{i \in [h] \setminus B_\sigma} \frac{1}{i}\right) \cdot \lim_{q \to -\frac{1}{3}}  (3q+1)^{|B_\sigma|} \cdot \left(\prod_{i \in B_\sigma \cup \{h+1, \dots, n\}} \frac{3}{3i - C_{i, \sigma}}\right)\\
      &= \left(\prod_{i \in [h] \setminus B_\sigma} \frac{1}{i}\right) \cdot 0 \cdot \left(\prod_{i \in B_\sigma \cup \{h+1, \dots, n\}} \frac{3}{3i - C_{i, \sigma}}\right)\\
      &= 0.
\end{align*}
Let $\omxc \subseteq [n]!$ be the set of orderings $\sigma$ whose last $h$ sets form an exact cover. Notice that each exact cover corresponds to exactly $(n-h)!h!$ orderings in $\omxc$. Therefore,
\[
\lim_{q \to -\frac{1}{3}} (3q+1)^h  \cdot \sum_{\sigma \in [n]!} \chi_{\sigma}(q) = \sum_{\sigma \in [n]!} \lim_{q \to -\frac{1}{3}} (3q+1)^h \chi_{\sigma}(q) = \sum_{\sigma \in \omxc}{\frac{1}{(n-h)!h!}} = \#\mathsf{X3C}(I). \qedhere
\]
\end{proof}

The two propositions complete the reduction: oracle values on polynomially many height-$2$ posets determine $R/D$, and the limit of $(3q+1)^hR(q)/D(q)$ at $q=-1/3$ gives $\#\mathsf{X3C}(I)$. At this finite point, the relevant information is the order of vanishing of the numerator and denominator, together with their first nonzero derivatives. All of these quantities can be computed exactly in polynomial time. We obtain the following alternative proof of the theorem of Dittmer and Pak~\cite{Dittmer_Pak_2020}.

\heighttwocounting

\begin{proof}[Proof of~\Cref{thm:h2le}]
    Membership in $\#\mathsf{P}$ is clear. Let $I$ be an arbitrary instance of $\#\mathsf{X3C}$.
    By~\Cref{prop:le_compute-polynomials} and~\Cref{prop:le_limit}, we can compute in polynomial time polynomials $R$ and $D$ such that
    \[
        \#\mathsf{X3C}(I) = \lim_{q \to -\frac{1}{3}} \frac{(3q+1)^h R(q)}{D(q)}.
    \]
    Since
$\frac{(3q+1)^h R(q)}{D(q)}$
is a ratio of polynomials whose size is polynomial in the size of $I$, the limit
$\lim_{q \to -\frac{1}{3}} \frac{(3q+1)^h R(q)}{D(q)}$
can be computed by simply evaluating the numerator and denominator at $q=-\frac{1}{3}$, except when this evaluation yields the indeterminate form $\frac{0}{0}$. In that case, we repeatedly apply L'Hôpital's rule until the resulting limit is no longer of the form $\frac{0}{0}$, after which it can be computed by evaluating the resulting polynomials at $q=-\frac{1}{3}$.
\end{proof}

%% file: figures/height-two-instance.tex
\begin{tikzpicture}[scale=0.7]
\path (2, 3.8) (2, 12.7);
\node[draw, circle, fill=blue!10] (s0) at (2, 11.25- 2*0) {\small $S_1$};
\node[draw, circle, fill=red!30] (s1) at (2, 11.25- 2*1) {\small $S_2$};
\node[draw, circle, fill=red!30] (s2) at (2, 11.25- 2*2) {\small $S_3$};
\node[draw, circle, fill=blue!10] (s3) at (2, 11.25- 2*3) {\small $S_4$};
\node[draw, circle, fill=orange!10, inner sep=2pt] (e1) at (5, 8*1.5- 1.5*0) {\small $u_1$};
\node[draw, circle, fill=orange!10, inner sep=2pt] (e2) at (5, 8*1.5- 1.5*1) {\small $u_2$};
\node[draw, circle, fill=orange!10, inner sep=2pt] (e3) at (5, 8*1.5- 1.5*2) {\small $u_3$};
\node[draw, circle, fill=orange!10, inner sep=2pt] (e4) at (5, 8*1.5- 1.5*3) {\small $u_4$};
\node[draw, circle, fill=orange!10, inner sep=2pt] (e5) at (5, 8*1.5- 1.5*4) {\small $u_5$};
\node[draw, circle, fill=orange!10, inner sep=2pt] (e6) at (5, 8*1.5- 1.5*5) {\small $u_6$};
\begin{scope}[on background layer]
\draw[] (s0) -- (e1);
\draw[] (s0) -- (e2);
\draw[] (s0) -- (e4);
\draw[very thick, red!50!gray] (s1) -- (e3);
\draw[very thick,red!50!gray] (s1) -- (e4);
\draw[very thick,red!50!gray] (s1) -- (e6);
\draw[very thick, red!50!gray] (s2) -- (e1);
\draw[very thick, red!50!gray] (s2) -- (e2);
\draw[very thick, red!50!gray] (s2) -- (e5);
\draw[] (s3) -- (e2);
\draw[] (s3) -- (e5);
\draw[] (s3) -- (e6);
\end{scope}
\end{tikzpicture}

%% file: figures/height-two-reduction.tex
\usetikzlibrary{shapes.geometric} 
\begin{tikzpicture}[scale=0.7]
\path (2, 3.8) (2, 12.7);
\node[draw, circle, fill=blue!10] (s1) at (2, 11.25) {\small $\ns_1$};
\node[draw, circle, fill=blue!10] (s2) at (2, 9.25) {\small $\ns_2$};
\node[draw, circle, fill=blue!10] (s3) at (2, 7.25) {\small $\ns_3$};
\node[draw, circle, fill=blue!10] (s4) at (2, 5.25) {\small $\ns_4$};

\foreach \i in {1,...,6} {
    \node[draw, dotted, very thick, ellipse, fill=yellow!10, opacity=0.4,
          inner sep=0pt, minimum width=54pt, minimum height=26pt] (group\i)
          at (8, {12-1.5*(\i-1)}) {};
    \node[draw, circle, fill=orange!10, inner sep=1pt] (u\i-1)
          at (7.45, {12-1.5*(\i-1)}) {\tiny $u_{\i}^{(1)}$};
    \node[draw, circle, fill=orange!10, inner sep=1pt] (u\i-2)
          at (8.55, {12-1.5*(\i-1)}) {\tiny $u_{\i}^{(2)}$};
}

\begin{scope}[on background layer]
\draw[-Latex, very thick, gray!70] (s1) -- (group1);
\draw[-Latex, very thick, gray!70] (s1) -- (group2);
\draw[-Latex, very thick, gray!70] (s1) -- (group4);
\draw[-Latex, very thick, gray!70] (s2) -- (group3);
\draw[-Latex, very thick, gray!70] (s2) -- (group4);
\draw[-Latex, very thick, gray!70] (s2) -- (group6);
\draw[-Latex, very thick, gray!70] (s3) -- (group1);
\draw[-Latex, very thick, gray!70] (s3) -- (group2);
\draw[-Latex, very thick, gray!70] (s3) -- (group5);
\draw[-Latex, very thick, gray!70] (s4) -- (group2);
\draw[-Latex, very thick, gray!70] (s4) -- (group5);
\draw[-Latex, very thick, gray!70] (s4) -- (group6);
\end{scope}
\end{tikzpicture}

%% file: sections/peeling.tex
Our reduction starts from counting linear extensions of posets of dimension $2$, a $\sharpP$-complete problem~\cite[Theorem~1.3]{dittmer2018counting}. A poset $(P, \preceq)$ is said to have \emph{dimension} (at most) $2$ if there exist two total orders of $P$ whose intersection (viewed as an intersection of relations) is exactly the partial order $\preceq$. We represent each poset element by a point and add two families of auxiliary points, called northern and eastern anchors. As long as at least one anchor of each type remains, a poset point can be peeled exactly when its element is minimal among the remaining elements. Thus, while both anchor families are present, peeling the poset points follows a linear extension.

The anchors themselves can be peeled at any time, so the construction also admits sequences that do not encode linear extensions. The interpolation technique plays a key role here: it lets us keep these extra sequences in the construction and separate their contributions algebraically. As the number of anchors grows, their contribution to the normalized count tends to zero, while a uniform contribution from each linear extension survives. Here the desired count is therefore isolated by a limit at infinity. This simplifies the reduction by avoiding the need for additional geometric constraints that force the anchors to remain. Following~\Cref{sec:strategy}, we obtain the limit exactly by reconstructing a rational function from polynomially many oracle queries.

Let $(P, \preceq)$ be a poset of dimension $2$, and let $n = |P|$ with $n\geq 1$. We describe how to compute $\#\mathsf{LE}(P)$ in polynomial time using queries to $\peel(\points)$. Let $\rho_1, \rho_2\colon P \rightarrow [n]$ be the rank functions of two total orders of the elements of $P$ such that, for every $a, b \in P$, it holds that $a \preceq b$ if and only if $\rho_1(a) \leq \rho_1(b)$ and $\rho_2(a) \leq \rho_2(b)$. It is known that such a pair of total orders can be computed in time $O(n^2)$ \cite{2d_listings}.

\begin{figure}[!t]
    \begin{subfigure}{0.49\linewidth}
        \centering
        \input{figures/peeling-poset}
        \caption{The diamond poset, whose unique minimal element is $a$.}\label{fig:peeling-poset}
    \end{subfigure}
    \hfill
    \begin{subfigure}{0.49\linewidth}
        \centering
        \input{figures/peeling-configuration}
        \caption{The configuration $\points_3$, with three anchors of each type and the convex hull in gray.}\label{fig:peeling-configuration}
    \end{subfigure}
    \caption{The reduction for the diamond poset and $q=3$, using the orderings $(a,b,c,d)$ and $(a,c,b,d)$, whose rank functions are $\rho_1$ and $\rho_2$, respectively. Northern anchors lie above and to the left of all poset points; eastern anchors lie below and to their right.
    }\label{fig:peeling-reduction}
\end{figure}

The two total orders determine the coordinates of the poset points. We space these coordinates exponentially and place the anchors on two short parabola arcs, one above and one to the right of all poset points. For every element $a \in P$, define the \emph{poset point} $p_a \coloneqq (2^{\rho_1(a)},\, 2^{(n+1)\rho_2(a)})$. For every integer $q \geq 1$ and every $j \in \{0, 1, \dots, q-1\}$, define the \emph{northern anchors} $\noan_j^{(q)} = (\frac{j}{q},\, 2^{(n+1)^2}+1-\frac{j^2}{q^2})$ and the \emph{eastern anchors} $\eaan_j^{(q)} = (2^{n+1}+1-\frac{j^2}{2^{(n+1)^2}q^2},\, \frac{j}{q})$. The resulting point configuration is \[\points_q \coloneqq \{p_a \mid a \in P \} \cup \{\noan_j^{(q)} \mid 0\leq j < q \} \cup \{\eaan_j^{(q)} \mid 0\leq j < q \}.\]

\Cref{fig:peeling-reduction} illustrates the role of the anchors. The configuration $\points_q$ has $n+2q$ points, each with coordinate bit length $O(n^2+\log q)$. Thus, polynomial values of $q$ give configurations of polynomial encoding size. We first verify that the construction is in general position and that every anchor is a \emph{hull point}, that is, a vertex of the convex hull; an anchor therefore remains available for peeling after any other points have been removed.

\begin{lemma}\label{lem:general_position}
 For every $q \geq 1$, the point configuration $\points_q$ is in general position. Moreover, all the anchor points of $\points_q$ are hull points.
\end{lemma}
The next lemma connects the geometry to the partial order. While both anchor types are present, the available poset points are exactly those corresponding to minimal elements of the remaining poset. This is the property that makes it possible to recover linear extensions from peeling sequences.

\begin{lemma}\label{lemma:minimal_elements}
 Let $q \geq 1$ and let $S \subseteq \points_q$ be a subconfiguration that contains at least one northern anchor and at least one eastern anchor. Then, for any poset point $p_a \in S$, we have that $p_a$ is a hull point of $S$ if and only if there is no other poset point $p_b \in S$ such that $b \preceq a$.
\end{lemma}

Once one anchor family is exhausted, additional poset points may become available. To count these sequences as well, we need to know how the remaining anchors affect the convex hull. The next lemma shows that, within either family, their identities and number do not matter as long as the family is nonempty.

\begin{lemma}\label{lemma:anchors_indist}
Let $S \subseteq \{p_a \mid a \in P \}$ be a subconfiguration of poset points, and let $q \geq 1$. The following statements hold:
\begin{enumerate}
\item Let $T_1$ and $T_2$ be two non-empty subsets of $\{\noan_j^{(q)} \mid 0\leq j < q \}$. Then, the subset of points of $S$ that are hull points of the configuration $S \cup T_1$ is the same as the subset of points of $S$ that are hull points of the configuration $S \cup T_2$.
\item Let $T_1$ and $T_2$ be two non-empty subsets of $\{\eaan_j^{(q)} \mid 0\leq j < q \}$. Then, the subset of points of $S$ that are hull points of the configuration $S \cup T_1$ is the same as the subset of points of $S$ that are hull points of the configuration $S \cup T_2$.
\end{enumerate}
\end{lemma}

We can now group peeling sequences by the order of their poset points and the stages at which the two anchor families are exhausted. Given an ordering $\sigma \in P!$ of the elements of $P$, and two integers $m, k \in \{0, 1, \dots, n\}$, let $\peel(\points_q)_{\sigma, m, k}$ be the number of peeling sequences of $\points_q$ for which the following conditions hold:
\begin{itemize}
\item The induced order on the poset points corresponds to $\sigma$.
\item The number of poset points peeled before all northern anchors are peeled is exactly $m$.
\item The number of poset points peeled before all eastern anchors are peeled is exactly $k$.
\end{itemize}
Let $\ompe \subseteq P! \times \{0, 1, \dots, n\}^2$ be the subset of triples $(\sigma, m, k)$ for which $\peel(\points_q)_{\sigma, m, k} > 0$. This set does not depend on $q$: by~\Cref{lemma:minimal_elements,lemma:anchors_indist}, the available poset points depend only on which poset points remain and whether northern anchors, eastern anchors, both, or neither are still present; moreover, by~\Cref{lemma:anchors_indist}, a nonempty set of northern (eastern) anchors has the same effect as $\noan_0^{(q)}$ ($\eaan_0^{(q)}$), which is the same point for every $q$. Increasing $q$ changes how many sequences realize each triple, but not which triples are possible.

For a fixed triple, counting the ways to interleave and order the anchors gives the following formula. Normalizing by $(n+2q)!$, the number of all point orderings, gives the rational function below.

\begin{lemma}\label{lemma:peels-formula}
Let $(\sigma, m, k) \in \ompe$ and define $a \coloneqq \min(m, k)$ and $b \coloneqq \max(m, k)$. Then,
\[
\chi_{\sigma, m, k}(q) \coloneqq \frac{\peel(\points_q)_{\sigma, m, k}}{(n+2q)!} = \frac{\left(1+\mathds{1}_{m=k} \right)q^2}{a! \left(\prod_{j = a}^b{(q+j)} \right) \left( \prod_{j = b}^{n}{(2q+j)}\right)}.
\]
\end{lemma}

Summing the expressions from~\Cref{lemma:peels-formula} over all possible triples gives
\[
    \frac{\peel(\points_q)}{(n+2q)!} =  \sum_{(\sigma, m, k) \in \ompe} \chi_{\sigma, m, k}(q).
\]
Each oracle query gives the value of the entire sum. Its denominators divide an explicit polynomial of degree $2n+2$, so multiplying by this common denominator allows us to recover the numerator by interpolation.

\begin{proposition}\label{prop:peels_interpolation}
   Given an oracle for $\peel(\cdot)$, we can compute in polynomial time (in terms of $n$) two polynomials $R$ and $D$ in canonical form such that the following equality between rational functions holds:
   \[ \sum_{(\sigma, m, k) \in \ompe} \chi_{\sigma, m, k}(q) = \frac{R(q)}{D(q)}.
    \]
\end{proposition}

The parameters $m$ and $k$ now identify the part of the count that encodes linear extensions. When $m=k=n$, anchors may be removed between poset points, but at least one anchor of each type survives until the last poset point is peeled. By~\Cref{lemma:minimal_elements}, these triples correspond exactly to the linear extensions of $P$. For every other triple, one anchor family is exhausted earlier. In the formula for $\chi_{\sigma,m,k}(q)$, this leaves more powers of $q$ in the denominator than in the numerator, making its contribution vanish as $q\to\infty$. The surviving terms all tend to $1/n!$, so the limit gives the number of linear extensions up to this known factor.

Since $\ompe$ is finite and independent of $q$, we can sum these limits term by term. Moreover,~\Cref{prop:peels_interpolation} recovers an identity of rational functions, so the limit is determined exactly by $R/D$. The reduction uses only the positive integer values $q=1,\ldots,2n+3$; no estimate of the rate of convergence or queries on larger configurations are needed.

\begin{proposition}\label{prop:peel_limit}
    We have that 
    \[
        \lim_{q \to \infty} \sum_{(\sigma, m, k) \in \ompe} \chi_{\sigma, m, k}(q) = \frac{\#\mathsf{LE}(P)}{n!}.
    \]
\end{proposition}

The geometric construction, interpolation, and limit calculation together give the hardness reduction. All queried configurations have polynomial encoding size, and the recovered limit yields $\#\mathsf{LE}(P)$ after multiplication by $n!$. Combining this with polynomial-time verification of a peeling sequence gives the following result.

\peelingcounting

%% file: figures/peeling-poset.tex
\begin{tikzpicture}[xscale=0.62, yscale=0.46]
\path[use as bounding box] (-0.9,-0.9) rectangle (10.9,10.9);
\node[draw, circle, fill=red!30, minimum size=5.8mm, inner sep=1pt] (a) at (5,0.8) {\small $a$};
\node[draw, circle, fill=blue!10, minimum size=5.8mm, inner sep=1pt] (b) at (1.8,5) {\small $b$};
\node[draw, circle, fill=blue!10, minimum size=5.8mm, inner sep=1pt] (c) at (8.2,5) {\small $c$};
\node[draw, circle, fill=blue!10, minimum size=5.8mm, inner sep=1pt] (d) at (5,9.2) {\small $d$};
\begin{scope}[on background layer]
\draw[-Latex, very thick, gray!70] (a) -- (b);
\draw[-Latex, very thick, gray!70] (a) -- (c);
\draw[-Latex, very thick, gray!70] (b) -- (d);
\draw[-Latex, very thick, gray!70] (c) -- (d);
\end{scope}
\end{tikzpicture}

%% file: figures/peeling-configuration.tex
\begin{tikzpicture}[xscale=0.46, yscale=0.35]
\path[use as bounding box] (-4.9,-4.9) rectangle (10.9,10.9);
\node[draw, circle, fill=red!30, minimum size=5.4mm, inner sep=0.8pt] (a) at (0.5,0.5) {\footnotesize $p_a$};
\node[draw, circle, fill=blue!10, minimum size=5.4mm, inner sep=0.8pt] (b) at (1.5,1.9) {\footnotesize $p_b$};
\node[draw, circle, fill=blue!10, minimum size=5.4mm, inner sep=0.8pt] (c) at (2.4,0.7) {\footnotesize $p_c$};
\node[draw, circle, fill=blue!10, minimum size=5.4mm, inner sep=0.8pt] (d) at (2.42,3.4) {\footnotesize $p_d$};
\node[draw, circle, fill=orange!10, minimum size=6.2mm, inner sep=0.1pt] (n0) at (-4.3,10.3) {\scriptsize $\noan_0^{(3)}$};
\node[draw, circle, fill=orange!10, minimum size=6.2mm, inner sep=0.1pt] (n1) at (-2.14,9.13) {\scriptsize $\noan_1^{(3)}$};
\node[draw, circle, fill=orange!10, minimum size=6.2mm, inner sep=0.1pt] (n2) at (-0.16,7.96) {\scriptsize $\noan_2^{(3)}$};
\node[draw, circle, fill=orange!10, minimum size=6.2mm, inner sep=0.1pt] (e0) at (10.3,-4.3) {\scriptsize $\eaan_0^{(3)}$};
\node[draw, circle, fill=orange!10, minimum size=6.2mm, inner sep=0.1pt] (e1) at (9.13,-2.14) {\scriptsize $\eaan_1^{(3)}$};
\node[draw, circle, fill=orange!10, minimum size=6.2mm, inner sep=0.1pt] (e2) at (7.96,-0.16) {\scriptsize $\eaan_2^{(3)}$};
\begin{scope}[on background layer]
\draw[very thick, gray!70] (n0.center) -- (n1.center) -- (n2.center)
    -- (e2.center) -- (e1.center) -- (e0.center) -- (a.center) -- cycle;
\end{scope}
\end{tikzpicture}

%% file: sections/n-free.tex
We construct a poset as the \defn{edge poset} of a directed acyclic graph. Each edge becomes a poset element, and every edge entering a vertex must precede every edge leaving that vertex. Linear extensions are therefore exactly the edge orderings that satisfy this rule. This edge-poset construction guarantees that the poset is N-free~\cite{FelsnerManneville}. Our graph represents an instance of $\#\mathsf{X3C}$. At each vertex representing a set or a universe element, we add $q$ edges to $q$ new vertices, each with no outgoing edges. The longest directed paths then have three edges. Since the height of an edge poset equals this maximum path length, the poset has height $3$ for every positive integer $q$. Interpolation plays a key role because the construction also admits edge orderings that do not encode exact covers. As in~\Cref{sec-height-two}, the leading behavior at infinity does not distinguish the desired terms. Recovering the entire normalized count as a rational function lets us examine it near a finite point instead, where the terms encoding exact covers have poles of the largest order. A suitable normalization then removes the other contributions. This allows us to use a simple construction that preserves the required poset structure and extract the desired count algebraically.

The proof makes this separation explicit in three steps. We first fix the order of the original edges. The resulting constraints form a tree poset, so the hook-length formula gives a rational expression for the normalized count. Next, we identify the orderings that survive the normalized limit and show that their total contribution is the same for each exact cover. Finally, we recover the rational function from polynomially many oracle queries and compute the limit exactly.

Let $I = (\mathcal{S}, U)$ be an instance of $\#\mathsf{X3C}$, with $|\mathcal{S}| = n$ and $|U| = 3h$.
For each positive integer $q$, we construct a directed acyclic graph $A_{I, q}$ as follows:
\begin{enumerate}
    \item Create a vertex $\ns_i$ for each $S_i \in \mathcal{S}$, and a vertex $\nou_j$ for each $u_j \in U$.
    \item Create an edge $(\ns_i, \nou_j)$ whenever $u_j \in S_i$.
    \item Create a vertex $r$, and for every $S_i \in \mathcal{S}$, create an edge $(r, \ns_i)$.
    \item For each $u_j \in U$, create $q$ additional copies $u_j^{(k)}$ ($1 \leq k \leq q$) and the edges $(\nou_j, u_j^{(k)})$.
    \item For each $S_i \in \mathcal{S}$, create $q$ additional copies $s_i^{(k)}$ ($1 \leq k \leq q$) and the edges $(\ns_i, \ns_i^{(k)})$.
\end{enumerate}
Now let $P_{I, q}$ be the edge poset of $A_{I, q}$. Note that $P_{I, q}$ has height $3$, and since $A_{I, q}$ is a DAG, the characterization of~\cite{FelsnerManneville} implies $P_{I, q}$ is always N-free. 
We work directly with the corresponding edge orderings of $A_{I, q}$.

\begin{figure}[!t]
    \begin{subfigure}[t]{0.158\linewidth}
        \centering
        \input{figures/nfree-instance}%
        \caption{Instance $I$.}\label{fig:nfree-instance}
    \end{subfigure}
    \hfill
    \begin{subfigure}[t]{0.268\linewidth}
        \centering
        \input{figures/nfree-dag}%
        \caption{Part of $A_{I,2}$ around $\nou_j$.}\label{fig:nfree-dag}
    \end{subfigure}
    \hfill
    \begin{subfigure}[t]{0.531\linewidth}
        \centering
        \input{figures/nfree-edge-poset}%
        \caption{The corresponding part of $P_{I,2}$.}\label{fig:nfree-edge-poset}
    \end{subfigure}
    \caption{The reduction from $\#\mathsf{X3C}$ for $q=2$, using the instance in~\Cref{ex:x3c}, whose exact cover $\{S_2,S_3\}$ is highlighted in red in (a). Panel (b) shows all paths from $r$ to copies of $\ns_i$, $\ns_\ell$ and $u_j$, where $u_j$ lies only in $S_i$ and $S_\ell$, as $u_1$ does in $S_1$ and $S_3$. Each element of (c) is an edge of (b), colored like the vertex it enters, and the arrows are all cover relations among them. Dotted groups contain copies and copied edges. The longest paths of $A_{I,2}$ have three edges, so $P_{I,2}$ has height~$3$. Each edge entering a vertex is covered by exactly the edges leaving it, as the four arrows into level~$3$ show; hence elements with a common upper cover have the same upper covers, and the poset is N-free.}\label{fig:nfree-reduction}
\end{figure}

\Cref{fig:nfree-reduction} shows how the graph gives an N-free poset of height $3$. Each edge becomes a poset element, and two consecutive edges on a directed path give a cover relation. The copied edges leaving set vertices lie on the second level, while those leaving universe vertices lie on the third.

Let $E_q := E(A_{I, q})$,
and let $E' \subsetneq E_q$ be the subset of edges introduced in steps (2) and (3). Note that $|E'| = 4|\mathcal{S}| = 4n$.
Also, for each ordering $\sigma$ of $E'$, we consider the number $\#\mathsf{LE}(P_{I, q})_{\sigma}$ of linear extensions of $P_{I, q}$ whose projection down to $E'$ coincides with $\sigma$. We will write $\omle$ to denote the subset of orderings $\sigma$ of $E'$ for which $\#\mathsf{LE}(P_{I, q})_{\sigma} > 0$. Every edge of $E_q \setminus E'$ ends at a vertex with no outgoing edges, so it is a maximal element of $P_{I,q}$. Hence, an ordering $\sigma$ of $E'$ belongs to $\omle$ if and only if it is a linear extension of the edge poset of $E'$; in particular, $\omle$ does not depend on $q$.

Fixing $\sigma$ determines when each group of copied edges becomes available. An edge leaving $\nou_j$ must follow the last original edge entering $\nou_j$, while a copied edge leaving $\ns_i$ must follow $(r,\ns_i)$. Thus, each copied edge can be attached as a leaf to a single position in the chain specified by $\sigma$.

We now define a new poset $T_{\sigma} := (E_q, \preceq)$ for each fixed $\sigma \in \omle$. For ease of notation, let us write $e_i^\sigma \coloneqq \sigma(i)$ for each $i \in [4n]$, so that $\sigma = (e_1^\sigma, \dots, e_{4n}^\sigma)$. Also, for an element $u_j \in U$, let
\[
\rho_{\sigma}(u_j) := \max \{ i \mid e_{i}^\sigma = (\ns_\ell, \nou_j) \text{ for some } \ell \}.
\]
The order $\preceq$ is defined by:
\begin{enumerate}
    \item $e_1^\sigma \preceq e_2^\sigma \preceq \dots  \preceq e_{4n-1}^\sigma \preceq e_{4n}^\sigma$.
    \item For each edge $e := (\nou_j, u_j^{(k)}) \in E_q \setminus E'$, we have $e_{\rho_\sigma(u_j)}^\sigma \preceq e$.
    \item For each edge $e := (\ns_i, \ns_i^{(k)}) \in E_q \setminus E'$, we have $(r, \ns_i) \preceq e$.
\end{enumerate}
Because $\sigma \in \omle$, the linear extensions of $T_\sigma$ are exactly the linear extensions of $P_{I, q}$ whose projection down to $E'$ coincides with $\sigma$. Thus, by~\Cref{lemma:tree-le}, we have
\[
\#\mathsf{LE}(P_{I, q})_{\sigma} = (4n+3hq+nq)! \prod_{i=1}^{4n}\frac{1}{i + q \cdot C_{i, \sigma}},
\]
where 
\[ 
C_{i, \sigma} := |\{ u_j \in U \mid \rho_{\sigma}(u_j) \geq 4n + 1 -i \}| + |\{\ell \in [n] \mid (r, \ns_\ell) = e_k^\sigma \text{ for some } k \geq 4n+1-i\}| .
\]
Here $C_{i,\sigma}$ counts the groups of copied edges attached to the last $i$ edges of the chain. Each group contributes $q$ leaves, so $i+qC_{i,\sigma}$ is the size of the subtree rooted at $e_{4n+1-i}^\sigma$.
Notice that $C_{i, \sigma} \leq C_{j, \sigma}$ whenever $i < j$. 
We again write 
\(
\chi_{\sigma}(q) \coloneqq \prod_{i=1}^{4n}\frac{1}{i + q \cdot C_{i, \sigma}},
\)
and thus have
\[
\frac{\#\mathsf{LE}(P_{I, q})}{(4n+3hq + nq)!} = \sum_{\sigma \in \omle} \chi_{\sigma}(q).
\]

The last edge in $\sigma$ is the last incoming edge of a universe vertex, so every $C_{i,\sigma}$ is positive. All terms therefore decay as $q^{-4n}$ at infinity, and normalization by $q^{4n}$ retains every ordering. We instead examine $q=-1$. As in~\Cref{prop:le_compute-polynomials}, a common denominator of polynomial degree allows exact reconstruction from oracle values at positive integers. The resulting identity of rational functions determines the normalized limit at this negative value algebraically.

Given an exact cover $X \subseteq \mathcal{S}$, define the edge set
$$E_X \coloneqq \{(r, \ns_i) \mid  S_i \in X\} \cup \{(\ns_i, \nou_j) \mid  S_i \in X, u_j \in S_i\}.$$

We will say that an ordering $\sigma \in \omle$ \emph{induces} the exact cover $X$, and write $\sigma \in \omxc(X)$, if the last $4h$ edges according to $\sigma$ correspond to the edges $E_X$ in some order.

At $q=-1$, a factor $i+qC_{i,\sigma}$ vanishes exactly when $C_{i,\sigma}=i$, contributing one to the pole order. The next lemma shows that the pole has order $4h$ exactly when $\sigma$ induces an exact cover. All other orderings have smaller pole order, so multiplying by $(q+1)^{4h}$ makes their contributions tend to zero.

\begin{lemma}\label{lem:le_cover_characterization}
Given any $\sigma \in \omle$, the following statements hold:
\begin{enumerate}
\item If $\sigma \in \omxc(X)$ for some exact cover $X$, then $C_{i, \sigma} = i$ for $i \in \{1, \dots, 4h\}$ and $C_{i, \sigma} < i$ for every $i \in \{4h+1, \dots, 4n\}$.
\item If $\sigma$ does not induce an exact cover, then the set $B_\sigma := \{ i \in [4h] \mid C_{i, \sigma} < i \}$ is non-empty.
\end{enumerate}
\end{lemma}

Now let $\sigma \in \omxc(X)$ for some exact cover $X$. The first part of \Cref{lem:le_cover_characterization} implies that
\begin{align*}
     \lim_{q \to -1} (q+1)^{4h}\chi_{\sigma}(q) &= \lim_{q \to -1}  (q+1)^{4h} \prod_{i=1}^{4h} \frac{1}{i + q\cdot i }  \cdot \prod_{i=4h+1}^{4n} \frac{1}{i + q\cdot C_{i, \sigma} }\\
     &= \lim_{q \to -1}   \prod_{i=1}^{4h} \frac{(q+1)}{i(q+1)}  \cdot \prod_{i=4h+1}^{4n} \frac{1}{i + q\cdot C_{i, \sigma} }\\
     &= \frac{1}{(4h)!} \lim_{q \to -1}  \prod_{i=4h+1}^{4n} \frac{1}{i + q\cdot C_{i, \sigma} }\\
     &= \frac{1}{(4h)!} \cdot \prod_{i=4h+1}^{4n} \frac{1}{i - C_{i, \sigma}}.
\end{align*}

Now let $\sigma \in \omle$ be an ordering that does not induce an exact cover. By the second part of \Cref{lem:le_cover_characterization}, we have that $|B_\sigma| > 0$. Since $C_{i+1, \sigma} - C_{i, \sigma} \in \{0, 1\}$ for every $i \in [4n-1]$, it follows that $C_{i, \sigma} < i$ also holds for every $i \in \{4h+1, \dots, 4n\}$. Thus,
\begin{align*}
     \lim_{q \to -1} (q+1)^{4h}\chi_{\sigma}(q) &= \lim_{q \to -1}  (q+1)^{4h} \prod_{i \in B_\sigma} \frac{1}{i + q \cdot C_{i, \sigma}}  \cdot \prod_{i \in [4h] \setminus B_\sigma} \frac{1}{i + q \cdot i}  \cdot \prod_{i=4h+1}^{4n} \frac{1}{i + q\cdot C_{i, \sigma} }\\
      &= \lim_{q \to -1}  (q+1)^{|B_\sigma|} \prod_{i \in B_\sigma} \frac{1}{i + q \cdot C_{i, \sigma}}  \cdot \prod_{i \in [4h] \setminus B_\sigma} \frac{(q+1)}{i(q+1)}  \cdot \prod_{i=4h+1}^{4n} \frac{1}{i + q\cdot C_{i, \sigma} }\\
        &= \left(\prod_{i \in [4h] \setminus B_\sigma} \frac{1}{i}\right) \cdot \lim_{q \to -1} (q+1)^{|B_\sigma|} \prod_{i \in B_\sigma} \frac{1}{i + q \cdot C_{i, \sigma}}  \cdot \prod_{i=4h+1}^{4n} \frac{1}{i + q\cdot C_{i, \sigma} } \\
      &= \left(\prod_{i \in [4h] \setminus B_\sigma} \frac{1}{i}\right) \cdot \lim_{q \to -1}  (q+1)^{|B_\sigma|} \cdot \left(\prod_{i \in B_\sigma \cup \{4h+1, \dots, 4n\}} \frac{1}{i - C_{i, \sigma}}\right)\\
      &= \left(\prod_{i \in [4h] \setminus B_\sigma} \frac{1}{i}\right) \cdot 0 \cdot \left(\prod_{i \in B_\sigma \cup \{4h+1, \dots, 4n\}} \frac{1}{i - C_{i, \sigma}}\right)\\
      &= 0.
\end{align*}

The surviving limits need not be equal, since the factors $i-C_{i,\sigma}$ depend on $\sigma$. What matters is their sum over the orderings inducing a fixed exact cover. The next lemma shows that this sum is independent of the cover, giving the uniform total contribution needed for the reduction. 
\begin{lemma}\label{lem:nfree_claim}
For every exact cover $X \subseteq \mathcal{S}$ it holds that
    \[
\sum_{\sigma \in\omxc(X)} \prod_{i=4h+1}^{4n} \frac{1}{i - C_{i, \sigma}} = \frac{(4h)!}{4^h3^{n-h}}.
\]
\end{lemma}

We now combine the two lemmas with interpolation to obtain the following result.

\nfreecounting

%% file: figures/nfree-instance.tex
\begin{tikzpicture}
\path[use as bounding box] (0,-0.25) rectangle (2.6,5.7);
\node[draw, circle, fill=blue!10, minimum size=6.2mm, inner sep=1pt]
    (s1) at (0.33,4.75) {\small $S_1$};
\node[draw, circle, fill=red!30, minimum size=6.2mm, inner sep=1pt]
    (s2) at (0.33,3.45) {\small $S_2$};
\node[draw, circle, fill=red!30, minimum size=6.2mm, inner sep=1pt]
    (s3) at (0.33,2.15) {\small $S_3$};
\node[draw, circle, fill=blue!10, minimum size=6.2mm, inner sep=1pt]
    (s4) at (0.33,0.85) {\small $S_4$};
\foreach \i/\y in {1/5.1,2/4.18,3/3.26,4/2.34,5/1.42,6/0.5} {
    \node[draw, circle, fill=orange!10, minimum size=6.2mm, inner sep=1pt]
        (u\i) at (2.27,\y) {\small $u_\i$};
}
\begin{scope}[on background layer]
\draw (s1) -- (u1);
\draw (s1) -- (u2);
\draw (s1) -- (u4);
\draw (s4) -- (u2);
\draw (s4) -- (u5);
\draw (s4) -- (u6);
\draw[very thick, red!50!gray] (s2) -- (u3);
\draw[very thick, red!50!gray] (s2) -- (u4);
\draw[very thick, red!50!gray] (s2) -- (u6);
\draw[very thick, red!50!gray] (s3) -- (u1);
\draw[very thick, red!50!gray] (s3) -- (u2);
\draw[very thick, red!50!gray] (s3) -- (u5);
\end{scope}
\end{tikzpicture}

%% file: figures/nfree-dag.tex
\begin{tikzpicture}[
    vertex/.style={draw, circle, minimum size=6.2mm, inner sep=1pt},
    edge/.style={-Latex, very thick, gray!70},
    copies/.style={draw=black, dotted, very thick, fill=yellow!10,
        opacity=0.4}]
\path[use as bounding box] (-2.21,-0.25) rectangle (2.21,5.7);
\begin{scope}[on background layer]
\draw[copies] (0,5.05) ellipse (1.1 and 0.56);
\draw[copies] (-1.32,3.5) ellipse (0.862 and 0.52);
\draw[copies] (1.32,3.5) ellipse (0.862 and 0.52);
\end{scope}
\node[vertex, fill=green!10] (r) at (0,0.4) {\small $r$};
\node[vertex, fill=blue!10] (si) at (-0.905,1.95) {\small $\ns_i$};
\node[vertex, fill=blue!10] (sl) at (0.905,1.95) {\small $\ns_\ell$};
\node[vertex, fill=orange!10] (u) at (0,3.5) {\small $\nou_j$};
\node[vertex, fill=orange!10] (u1) at (-0.6,5.05) {\scriptsize $u_j^{(1)}$};
\node[vertex, fill=orange!10] (u2) at (0.6,5.05) {\scriptsize $u_j^{(2)}$};
\node[vertex, fill=blue!10] (si1) at (-1.735,3.5) {\scriptsize $\ns_i^{(1)}$};
\node[vertex, fill=blue!10] (si2) at (-0.905,3.5) {\scriptsize $\ns_i^{(2)}$};
\node[vertex, fill=blue!10] (sl1) at (0.905,3.5) {\scriptsize $\ns_\ell^{(1)}$};
\node[vertex, fill=blue!10] (sl2) at (1.735,3.5) {\scriptsize $\ns_\ell^{(2)}$};
\begin{scope}[on background layer]
\foreach \a/\b in {r/si, r/sl, si/u, sl/u, u/u1, u/u2,
        si/si1, si/si2, sl/sl1, sl/sl2} {
    \draw[edge] (\a) -- (\b);
}
\end{scope}
\end{tikzpicture}

%% file: figures/nfree-edge-poset.tex
\begin{tikzpicture}[
    element/.style={draw, rectangle, rounded corners=3.1mm,
        minimum height=6.2mm, minimum width=10mm, inner xsep=3pt,
        inner ysep=1pt, font=\scriptsize},
    cover/.style={-Latex, very thick, gray!70},
    copies/.style={draw=black, dotted, very thick, fill=yellow!10,
        opacity=0.4}]
\path[use as bounding box] (-4.51,-0.25) rectangle (4.25,5.7);
\begin{scope}[on background layer]
\foreach \y in {1.175,2.725,4.275} {
    \draw[gray!25, dashed] (-4.2,\y) -- (4.24,\y);
}
\draw[copies] (0,4.275) ellipse (1.49 and 0.55);
\draw[copies] (-2.763,2.725) ellipse (1.387 and 0.55);
\draw[copies] (2.806,2.725) ellipse (1.419 and 0.55);
\end{scope}
\node[font=\scriptsize, text=gray!80, anchor=west, inner xsep=0pt]
    at (-4.478,4.95) {Level};
\foreach \l/\y in {1/1.175,2/2.725,3/4.275} {
    \node[font=\small, text=gray!80] at (-4.39,\y) {$\l$};
}
\node[element, fill=blue!10] (rsi) at (-2.108,1.175) {$(r,\ns_i)$};
\node[element, fill=blue!10] (rsl) at (2.134,1.175) {$(r,\ns_\ell)$};
\node[element, fill=blue!10] (sis1) at (-3.419,2.725) {$(\ns_i,\ns_i^{(1)})$};
\node[element, fill=blue!10] (sis2) at (-2.108,2.725) {$(\ns_i,\ns_i^{(2)})$};
\node[element, fill=orange!10] (siu) at (-0.618,2.725) {$(\ns_i,\nou_j)$};
\node[element, fill=orange!10] (slu) at (0.612,2.725) {$(\ns_\ell,\nou_j)$};
\node[element, fill=blue!10] (sls1) at (2.134,2.725) {$(\ns_\ell,\ns_\ell^{(1)})$};
\node[element, fill=blue!10] (sls2) at (3.478,2.725) {$(\ns_\ell,\ns_\ell^{(2)})$};
\node[element, fill=orange!10] (uu1) at (-0.711,4.275) {$(\nou_j,u_j^{(1)})$};
\node[element, fill=orange!10] (uu2) at (0.711,4.275) {$(\nou_j,u_j^{(2)})$};
\coordinate (in11) at ([xshift=-1.5mm]uu1.south);
\coordinate (in12) at ([xshift=1.5mm]uu1.south);
\coordinate (in21) at ([xshift=-1.5mm]uu2.south);
\coordinate (in22) at ([xshift=1.5mm]uu2.south);
\begin{scope}[on background layer]
\draw[cover] (rsi.center) -- ([xshift=2.2mm]sis1.south);
\draw[cover] (rsi.center) -- (sis2.south);
\draw[cover] (rsi.center) -- ([xshift=-1.5mm]siu.south);
\draw[cover] (rsl.center) -- ([xshift=1.5mm]slu.south);
\draw[cover] (rsl.center) -- (sls1.south);
\draw[cover] (rsl.center) -- ([xshift=-2.2mm]sls2.south);
\draw[cover] (in11 |- siu.center) -- (in11);
\draw[cover] (siu.center) -- (in21);
\draw[cover] (slu.center) -- (in12);
\draw[cover] (in22 |- slu.center) -- (in22);
\end{scope}
\end{tikzpicture}

%% file: sections/svos-shellings.tex
\subsection{Successive Vertex Orderings}

Every connected graph admits an SVO; the count measures how many ways its vertices can be added while maintaining connectivity. We first count \defn{rooted SVOs}: for a graph $G$ and a prescribed vertex $r\in V(G)$, let $\#\mathsf{SVO}(G;r)$ count the SVOs of $G$ whose first vertex is $r$.

Rooted SVOs differ from linear extensions in the constraints they impose on an ordering. In a rooted SVO, one earlier neighbor is enough to make a vertex available; in a linear extension, all predecessors must appear first. In~\Cref{fig:difference}, either $a$ or $b$ makes $c$ available in the graph, whereas both must precede $c$ in the poset. Our reduction handles these disjunctive constraints by allowing extra orderings and separating their contributions algebraically, while preserving bipartiteness.

\begin{figure}[!hb]
    \begin{subfigure}{0.49\linewidth}
      \centering
        \input{figures/svo_vs_h2le}
        \caption{}\label{subfig:diff-svo}
    \end{subfigure}
    \begin{subfigure}{0.49\linewidth}
      \centering
        \input{figures/svo_vs_h2le2}
        \caption{}\label{subfig:diff-h2le}
    \end{subfigure}
    \caption{Comparison of the constraints for \emph{rooted} $\mathsf{SVO}$ and $\mathsf{LE}$: $(r, a,c,b)$ is a valid SVO, but $(a,c,b)$ is not a valid linear extension.}
    \label{fig:difference}
\end{figure}

Our reduction starts from an instance $I=(\mathcal{S},U)$ of $\#\mathsf{X3C}$, with $|\mathcal{S}|=n$ and $|U|=3h$. We represent each set by a vertex and replace each universe element by $q$ copies, which become available when a set containing that element appears. The normalized count is a rational function of $q$ whose terms correspond to orderings of the set vertices. In this construction, the relevant distinction is how early the set vertices cover the universe: the first $h$ sets suffice precisely when they form an exact cover, while every other ordering requires more sets. After factorial normalization, the terms associated with exact covers decay as $q^{-h}$, while all other terms decay faster. We therefore multiply by $q^h$ and take the limit as $q\to\infty$, retaining the exact-cover contributions and eliminating the others. We recover this rational function by interpolation from polynomially many queries at positive integers and compute the limit algebraically.

\begin{samepage}
Formally, for each positive integer $q$, we construct the bipartite graph $G_{I,q}$ as follows:
\begin{enumerate}
    \item Create a vertex $\ns_i$ for each $S_i \in \mathcal{S}$.
    \item Create a root vertex $r$, and connect it to each of the vertices $\ns_1, \dots, \ns_n$.
    \item For each element $u_j \in U$, create $q$ copies $u_j^{(1)}, \dots, u_j^{(q)}$, and connect all copies to every $\ns_i$ such that $u_j \in S_i$.
\end{enumerate}
\end{samepage} 

\Cref{fig:example-b} illustrates the construction. Since every universe element belongs to some set, $G_{I,q}$ is connected. After $r$, every set vertex is available, and each copied vertex becomes available with its first neighboring set vertex. Fixing the order of the set vertices thus gives one predecessor per copied vertex (see~\Cref{fig:poset-tree}), allowing us to apply the hook-length formula.

\begin{figure}
    \begin{subfigure}{0.49\linewidth}
         \centering
         \input{figures/subfig1}
    \caption{Instance $I$ from~Example~\ref{ex:x3c}, with the exact cover $\{S_2, S_3\}$ colored.}\label{fig:example-a}
    \end{subfigure}
    \hfill
    \begin{subfigure}{0.49\linewidth}
        \centering
        \input{figures/subfig2}
    \caption{$G_{I, 3}$. Thick gray edges between a vertex $\ns_i$ and a ``group'' $U_j := \{u_j^1, u_j^2, u_j^3\}$ represent all edges $\{\ns_i, u_j^k\}$ for $k \in \{1,2,3\}$.}\label{fig:example-b}
    \end{subfigure}  
    \caption{Illustrations of an instance of $\#\mathsf{X3C}$ and one of its corresponding graphs.}\label{fig:example}
\end{figure}

For $\sigma\in[n]!$, let $\#\mathsf{SVO}(G_{I,q};r)_\sigma$ count the rooted SVOs whose set vertices appear in the order $\ns_{\sigma(1)},\ldots,\ns_{\sigma(n)}$. The next lemma gives the normalized contribution in terms of when each element of the universe $U$ is first covered.
 
\begin{lemma}\label{lemma:svo-formula}
Let $\sigma \in [n]!$. For each $j \in [n]$, let $c_j(\sigma)$ be the number of elements of $U$ whose earliest covering set is $S_{\sigma(j)}$, that is, those that are covered by $S_{\sigma(j)}$ but are not covered by $S_{\sigma(i)}$ for any $i<j$. Then 
\[
    \chi_{\sigma}(q) \coloneqq \frac{\#\mathsf{SVO}(G_{I, q}; r)_{\sigma}}{(n+3hq)!} =  \frac{1}{\prod_{i=1}^n\left( i + \sum_{j=1}^{i} q\cdot c_{n-j+1}(\sigma)\right)}.
    \]
\end{lemma}

\begin{figure}
    \centering
    \input{figures/poset_tree}
    \caption{Poset $T_\sigma$, defined in Appendix~\ref{app:proof-svo-formula}, for $\sigma = (2, 4, 1, 3)$, coming from the instance $I$ of~\Cref{ex:x3c} with $q = 3$.}
    \label{fig:poset-tree}
\end{figure}

\Cref{lemma:svo-formula} gives us a rational function for each fixed $\sigma \in [n]!$. Thus, 
\[
        \frac{\#\mathsf{SVO}(G_{I, q}; r)}{(n+3hq)!} = \sum_{\sigma \in [n]!} \chi_{\sigma}(q).
\]
Although the sum has $n!$ terms, we can recover it without enumerating the orderings. All denominator factors have the form $Cq+\ell$, with $0\leq C\leq |U|$ and $1\leq\ell\leq n$. A common denominator therefore has polynomial degree, and multiplying by it turns the normalized count into a polynomial that can be recovered by interpolation.

\begin{lemma}\label{lem:compute-polynomials}
   Given an instance $I = (\mathcal{S}, U)$ of $\#\mathsf{X3C}$, and an oracle for $\#\mathsf{SVO}(\cdot;\cdot)$, we can compute in polynomial time (in terms of the size of the instance $I$) two polynomials $R$ and $D$ in canonical form such that the following equality between rational functions holds:
   \[ \sum_{\sigma \in [n]!} \chi_{\sigma}(q) = \frac{R(q)}{D(q)}.
    \]
\end{lemma}

The limit at infinity distinguishes how quickly the set ordering covers $U$. Covering all $3h$ elements requires at least $h$ sets, with equality precisely when the first $h$ sets form an exact cover. For these orderings, exactly $h$ denominator factors of $\chi_\sigma(q)$ depend on $q$; every other ordering has at least $h+1$ such factors. Multiplication by $q^h$ therefore leaves a nonzero limit precisely for the desired orderings. Each exact cover contributes the same total weight, as the next lemma shows.

\begin{lemma}\label{lem:reduction}
    For every instance $I = (\mathcal{S}, U)$ of $\#\mathsf{X3C}$ with $|U| = 3h$ it holds that
    \[
        \lim_{q \to \infty} q^h \cdot \sum_{\sigma \in [n]!} \chi_{\sigma}(q) = 3^{-h} \cdot \#\mathsf{X3C}(I).
    \]
\end{lemma}

Combining the rational-function reconstruction with this limit yields the rooted hardness result.

\begin{thm}\label{thm:main2}
    Given $G$ and $r$, computing
    $\#\mathsf{SVO}(G; r)$ is $\#\mathsf{P}$-hard, even for bipartite graphs.
\end{thm}

The same construction also gives~\Cref{cor:st-numberings}. Add a vertex $t$ adjacent to every vertex of $G_{I,q}$, obtaining $\widetilde{G}_{I,q}$. Appending $t$ to a rooted SVO gives an $st$-numbering with first vertex $r$ and last vertex $t$: the SVO condition supplies an earlier neighbor, and $t$ supplies a later neighbor for every other vertex. Removing $t$ gives the inverse correspondence, so
\[
\#\mathsf{SVO}(G_{I,q};r)=\#st\text{-numberings}(\widetilde{G}_{I,q};\,r;\,t).
\]

It remains to remove the prescribed first vertex. Attaching $\ell$ new leaves to $r$ preserves bipartiteness and makes the unrooted count, divided by $\ell!$, a polynomial in $\ell$. Its leading coefficient determines the rooted count, so a second interpolation gives the following reduction.

\begin{proposition}\label{prop:unrooting_SVO}
Given a bipartite graph $G = (V, E)$ and a vertex $r \in V$, it is possible to compute $\#\mathsf{SVO}(G; r)$ in polynomial time, given access to an oracle for counting SVOs of bipartite graphs.
\end{proposition}

An ordering can be checked to be an SVO, or an $st$-numbering, in polynomial time, so both counting problems belong to $\sharpP$. Together with~\Cref{thm:main2},~\Cref{prop:unrooting_SVO}, and the identity above, this proves the following two results.

\svocounting

\stnumberingcounting

\subsection{Shellings}



The shellings of a graph $G$ are precisely the SVOs of its line graph, whose vertices represent edges of $G$ and whose adjacency records a shared endpoint. For graphs without isolated vertices, the existence of a shelling is equivalent to connectivity; the count measures how many ways the edges can be added while maintaining connectivity.

We first consider the rooted count $\#\mathsf{Shellings}(G;r)$, in which the first edge must be incident to a prescribed vertex $r$. The reduction again starts from an instance $I=(\mathcal{S},U)$ of $\#\mathsf{X3C}$, with $|\mathcal{S}|=n$ and $|U|=3h$. We represent each set and each universe element by a vertex, join $r$ to every set vertex and each set vertex to the vertices of its three elements, and attach $q$ new pendant edges to each universe vertex. A group of pendant edges becomes available when the ordering first reaches its universe vertex. Fixing the order of the original edges therefore gives a tree poset and a rational expression for the normalized count. Reaching all universe vertices requires at least $4h$ original edges, with equality precisely for the orderings encoding exact covers. Multiplication by $q^{4h}$ isolates these orderings in the limit at infinity. Their individual weights may differ, but their sum is the same for every exact cover. Interpolation then recovers the desired count from polynomially many oracle queries.

\begin{samepage}
Formally, for each positive integer $q$, we construct a graph $H_{I, q}$ as follows:
\begin{enumerate}
    \item Create a vertex $\ns_i$ for each $S_i \in \mathcal{S}$, and a vertex $\nou_j$ for each $u_j \in U$.
    \item Create an edge $\{\ns_i, \nou_j\}$ whenever $u_j \in S_i$.
    \item Create a vertex $r$ connected to every $\ns_i$.
    \item For each $u_j \in U$, create $q$ additional copies $u_j^{(k)}$ ($1 \leq k \leq q$), and connect them all to $\nou_j$.
\end{enumerate}
\end{samepage}
\Cref{fig:shellings} shows $H_{I,3}$ for the instance $I$ of~\Cref{ex:x3c}, with four sets and six universe elements. Note that $H_{I, q}$ is always bipartite.
Let $E_q := E(H_{I, q})$,
and let $E' \subsetneq E_q$ be the subset of edges introduced in steps (2) and (3). Note that $|E'| = 4|\mathcal{S}| = 4n$.
Also, for each ordering $\sigma$ of $E'$, we consider the number $\#\mathsf{Shellings}(H_{I, q}; r)_{\sigma}$ of shellings of $H_{I, q}$ whose first edge is incident to $r$ and whose projection down to $E'$ coincides with $\sigma$. We will write $\omsh$ to denote the subset of orderings $\sigma$ of $E'$ for which $\#\mathsf{Shellings}(H_{I, q}; r)_{\sigma} > 0$. Since every pendant edge has a leaf endpoint, an ordering $\sigma$ of $E'$ belongs to $\omsh$ if and only if it is a shelling of the subgraph of $H_{I,q}$ induced by the edges in $E'$ and its first edge is incident to $r$; in particular, $\omsh$ does not depend on $q$.

\begin{figure}
    \centering
    \input{figures/shellings}
    \caption{The graph $H_{I,q}$ for the instance $I$ in~\Cref{ex:x3c}, with $q=3$.}
    \label{fig:shellings}
\end{figure}

For $\sigma\in\omsh$, write $e_i^\sigma \coloneqq \sigma(i)$ for $i \in [4n]$, so that $\sigma = (e_1^\sigma,\ldots,e_{4n}^\sigma)$, and let
\[
\rho_\sigma(u_j):=\min\{i\mid e_i^\sigma=\{\ns_\ell,\nou_j\}\text{ for some }\ell\}.
\]
Thus, $\rho_\sigma(u_j)$ records when the ordering first reaches $\nou_j$ and makes its pendant edges available. As for SVOs, fixing the order of the original edges gives one predecessor for each pendant edge. The hook-length formula then yields the normalized contribution below.

\begin{lemma}\label{lem:shellings-formula}
Let $\sigma\in\omsh$. For each $i\in[4n]$, let
\[
C_{i,\sigma}:=|\{u_j\in U\mid \rho_\sigma(u_j)\geq 4n+1-i\}|.
\]
Then
\[
\chi_{\sigma}(q) \coloneqq \frac{\#\mathsf{Shellings}(H_{I,q};r)_\sigma}{(4n+3hq)!}
=\prod_{i=1}^{4n}\frac{1}{i+q\cdot C_{i,\sigma}}.
\]
\end{lemma}

\begin{samepage}
We now have that
\[
\frac{\#\mathsf{Shellings}(H_{I, q}; r)}{(4n+3hq)!} = \sum_{\sigma \in \omsh} \chi_{\sigma}(q).
\]
\end{samepage}

The sum can be recovered without enumerating the orderings in $\omsh$. Every denominator factor has the form $Cq+\ell$, with $0\leq C\leq |U|$ and $1\leq\ell\leq 4n$. As in the preceding subsection, a common denominator has polynomial degree, so interpolation recovers the normalized count from polynomially many oracle queries.

\begin{lemma}\label{lem:shellings-compute-polynomials}
Given an instance $I=(\mathcal{S},U)$ of $\#\mathsf{X3C}$ and an oracle for $\#\mathsf{Shellings}(\cdot;\cdot)$, we can compute in polynomial time two polynomials $R$ and $D$ in canonical form such that the following equality between rational functions holds:
\[
\sum_{\sigma\in\omsh}\chi_\sigma(q)=\frac{R(q)}{D(q)}.
\]
\end{lemma}

The limit at infinity distinguishes how quickly the original edges reach all universe vertices. At least $4h$ edges are needed, with equality precisely when these edges consist of the root edges and incidence edges of an exact cover. For these orderings, exactly $4h$ denominator factors of $\chi_\sigma(q)$ depend on $q$; every other ordering has at least $4h+1$ such factors. Multiplication by $q^{4h}$ therefore retains only the exact-cover contributions. Although the surviving weights depend on the ordering, their total is the same for every exact cover.

\begin{lemma}\label{lem:shellings-reduction}
For every instance $I=(\mathcal{S},U)$ of $\#\mathsf{X3C}$ with $|U|=3h$, it holds that
\[
\lim_{q\to\infty}q^{4h}\cdot\sum_{\sigma\in\omsh}\chi_\sigma(q)
=3^{-h}\cdot\#\mathsf{X3C}(I).
\]
\end{lemma}

Combining the rational-function reconstruction with this limit yields the rooted hardness result.

\begin{thm}\label{thm:shellings2}
Given a bipartite graph $G = (V, E)$ and a vertex $r \in V$, 
computing $\#\mathsf{Shellings}(G; r)$ is $\sharpP$-hard.
\end{thm}

To remove the condition on the first edge, we attach $\ell$ new pendant edges at $r$ and divide the unrooted count by $\ell!$. As for SVOs, the resulting polynomial has a leading coefficient that determines the rooted count. The next proposition makes this reduction explicit.

\begin{proposition}\label{prop:unrooting_shellings}
Given a bipartite graph $G = (V, E)$ and a vertex $r \in V$, it is possible to compute $\#\mathsf{Shellings}(G; r)$ in polynomial time, given access to an oracle for counting shellings of bipartite graphs.
\end{proposition}

A shelling can be verified in polynomial time, so the counting problem belongs to $\sharpP$. Together with~\Cref{thm:shellings2} and~\Cref{prop:unrooting_shellings}, this proves the final result.

\shellingcounting

%% file: figures/svo_vs_h2le.tex
\begin{tikzpicture}
    \node[draw, circle, fill=green!10] (r) at (0, 0) {$r$};

    \node[draw, circle, fill=blue!10] (a) at (2, 1) {$a$};

    \node[draw, circle, fill=blue!10] (b) at (2, -1) {$b$};

    \node[draw, circle, fill=orange!10] (c) at (4, 0) {$c$};

    \node[] (c-or) at (4.7, 0.4) {`or'};

    \draw[-] (r) -- (a);
    \draw[-] (r) -- (b);

    \draw[-] (c) -- (a);
     \draw[-] (c) -- (b);

\end{tikzpicture}

%% file: figures/svo_vs_h2le2.tex
\begin{tikzpicture}

    \node[draw, circle, fill=blue!10] (a) at (2, 1) {$a$};

    \node[draw, circle, fill=blue!10] (b) at (2, -1) {$b$};

    \node[draw, circle, fill=orange!10] (c) at (4, 0) {$c$};


     \node[] (c-or) at (4.7, 0.4) {`and'};

    \draw[-Latex] (a) -- (c);
     \draw[-Latex] (b) -- (c);

\end{tikzpicture}

%% file: figures/subfig1.tex
\begin{tikzpicture}[scale=0.7]
\node[draw, circle, fill=blue!10] (s0) at (2, 7*1.5- 1.5*0) {\small $S_1$};
\node[draw, circle, fill=red!30] (s1) at (2, 7*1.5- 1.5*1) {\small $S_2$};
\node[draw, circle, fill=red!30] (s2) at (2, 7*1.5- 1.5*2) {\small $S_3$};
\node[draw, circle, fill=blue!10] (s3) at (2, 7*1.5- 1.5*3) {\small $S_4$};
\node[draw, circle, fill=orange!10, inner sep=1pt] (e1) at (5, 8*1.5- 1.5*0) {\small $u_1$};
\node[draw, circle, fill=orange!10, inner sep=1pt] (e2) at (5, 8*1.5- 1.5*1) {\small $u_2$};
\node[draw, circle, fill=orange!10, inner sep=1pt] (e3) at (5, 8*1.5- 1.5*2) {\small $u_3$};
\node[draw, circle, fill=orange!10, inner sep=1pt] (e4) at (5, 8*1.5- 1.5*3) {\small $u_4$};
\node[draw, circle, fill=orange!10, inner sep=1pt] (e5) at (5, 8*1.5- 1.5*4) {\small $u_5$};
\node[draw, circle, fill=orange!10, inner sep=1pt] (e6) at (5, 8*1.5- 1.5*5) {\small $u_6$};
\begin{scope}[on background layer]
\draw[] (s0) -- (e1);
\draw[] (s0) -- (e2);
\draw[] (s0) -- (e4);
\draw[very thick, red!50!gray] (s1) -- (e3);
\draw[very thick,red!50!gray] (s1) -- (e4);
\draw[very thick,red!50!gray] (s1) -- (e6);
\draw[very thick, red!50!gray] (s2) -- (e1);
\draw[very thick, red!50!gray] (s2) -- (e2);
\draw[very thick, red!50!gray] (s2) -- (e5);
\draw[] (s3) -- (e2);
\draw[] (s3) -- (e5);
\draw[] (s3) -- (e6);
\end{scope}
\end{tikzpicture}

%% file: figures/subfig2.tex
\begin{tikzpicture}[scale=0.7]
\node[draw, circle, fill=green!10] (r) at (0, 10.5 - 1.5*1.5) {$r$};
\node[draw, circle, fill=blue!10] (s0) at (2, 7*1.5- 1.5*0) {\small $\ns_1$};
\node[draw, circle, fill=blue!10] (s1) at (2, 7*1.5- 1.5*1) {\small $\ns_2$};
\node[draw, circle, fill=blue!10] (s2) at (2, 7*1.5- 1.5*2) {\small $\ns_3$};
\node[draw, circle, fill=blue!10] (s3) at (2, 7*1.5- 1.5*3) {\small $\ns_4$};
\node[draw, circle, fill=orange!10, inner sep=0pt] (e1_0) at (5.5, 12.93) {\tiny $u_1^{1}$};
\node[draw, circle, fill=orange!10, inner sep=0pt] (e1_1) at (5.19, 12.379999999999999) {\tiny $u_1^{2}$};
\node[draw, circle, fill=orange!10, inner sep=0pt] (e1_2) at (5.81, 12.379999999999999) {\tiny $u_1^{3}$};
\node[draw, dotted, very thick, circle, fill=yellow!10, opacity=0.4, inner sep=10pt] (1_top) at (5.5, 12.579999999999998) {};
\node[draw, circle, fill=orange!10, inner sep=0pt] (e2_0) at (5.5, 11.18) {\tiny $u_2^{1}$};
\node[draw, circle, fill=orange!10, inner sep=0pt] (e2_1) at (5.19, 10.629999999999999) {\tiny $u_2^{2}$};
\node[draw, circle, fill=orange!10, inner sep=0pt] (e2_2) at (5.81, 10.629999999999999) {\tiny $u_2^{3}$};
\node[draw, dotted, very thick, circle, fill=yellow!10, opacity=0.4, inner sep=10pt] (2_top) at (5.5, 10.829999999999998) {};
\node[draw, circle, fill=orange!10, inner sep=0pt] (e3_0) at (5.5, 9.43) {\tiny $u_3^{1}$};
\node[draw, circle, fill=orange!10, inner sep=0pt] (e3_1) at (5.19, 8.879999999999999) {\tiny $u_3^{2}$};
\node[draw, circle, fill=orange!10, inner sep=0pt] (e3_2) at (5.81, 8.879999999999999) {\tiny $u_3^{3}$};
\node[draw, dotted, very thick, circle, fill=yellow!10, opacity=0.4, inner sep=10pt] (3_top) at (5.5, 9.079999999999998) {};
\node[draw, circle, fill=orange!10, inner sep=0pt] (e4_0) at (5.5, 7.68) {\tiny $u_4^{1}$};
\node[draw, circle, fill=orange!10, inner sep=0pt] (e4_1) at (5.19, 7.13) {\tiny $u_4^{2}$};
\node[draw, circle, fill=orange!10, inner sep=0pt] (e4_2) at (5.81, 7.13) {\tiny $u_4^{3}$};
\node[draw, dotted, very thick, circle, fill=yellow!10, opacity=0.4, inner sep=10pt] (4_top) at (5.5, 7.33) {};
\node[draw, circle, fill=orange!10, inner sep=0pt] (e5_0) at (5.5, 5.93) {\tiny $u_5^{1}$};
\node[draw, circle, fill=orange!10, inner sep=0pt] (e5_1) at (5.19, 5.38) {\tiny $u_5^{2}$};
\node[draw, circle, fill=orange!10, inner sep=0pt] (e5_2) at (5.81, 5.38) {\tiny $u_5^{3}$};
\node[draw, dotted, very thick, circle, fill=yellow!10, opacity=0.4, inner sep=10pt] (5_top) at (5.5, 5.58) {};
\node[draw, circle, fill=orange!10, inner sep=0pt] (e6_0) at (5.5, 4.18) {\tiny $u_6^{1}$};
\node[draw, circle, fill=orange!10, inner sep=0pt] (e6_1) at (5.19, 3.63) {\tiny $u_6^{2}$};
\node[draw, circle, fill=orange!10, inner sep=0pt] (e6_2) at (5.81, 3.63) {\tiny $u_6^{3}$};
\node[draw, dotted, very thick, circle, fill=yellow!10, opacity=0.4, inner sep=10pt] (6_top) at (5.5, 3.83) {};
\begin{scope}[on background layer]
\draw[] (r) -- (s0);
\draw[] (r) -- (s1);
\draw[] (r) -- (s2);
\draw[] (r) -- (s3);
\draw[very thick, gray!70] (s0) -- (1_top);
\draw[very thick, gray!70] (s0) -- (2_top);
\draw[very thick, gray!70] (s0) -- (4_top);
\draw[very thick, gray!70] (s1) -- (3_top);
\draw[very thick, gray!70] (s1) -- (4_top);
\draw[very thick, gray!70] (s1) -- (6_top);
\draw[very thick, gray!70] (s2) -- (1_top);
\draw[very thick, gray!70] (s2) -- (2_top);
\draw[very thick, gray!70] (s2) -- (5_top);
\draw[very thick, gray!70] (s3) -- (2_top);
\draw[very thick, gray!70] (s3) -- (5_top);
\draw[very thick, gray!70] (s3) -- (6_top);
\end{scope}
\end{tikzpicture}

%% file: figures/poset_tree.tex
\begin{tikzpicture}[scale=0.7]
\node[draw, circle, fill=blue!10, ] (S_2) at (0*5.0, -0.0) {$\ns_2$};
\node[draw, circle, fill=blue!10, ] (S_4) at (1*5.0, -0.0) {$\ns_4$};
\node[draw, circle, fill=orange!10, inner sep=1pt] (u_3^{1}) at (1*5.0, -1.5) {\tiny $u_3^{1}$};
\node[draw, circle, fill=orange!10, inner sep=1pt] (u_3^{2}) at (1*5.0-0.4, -2.2) {\tiny $u_3^{2}$};
\node[draw, circle, fill=orange!10, inner sep=1pt] (u_3^{3}) at (1*5.0-0.8, -2.9) {\tiny $u_3^{3}$};
\node[draw, circle, fill=orange!10, inner sep=1pt] (u_4^{1}) at (1*5.0-0.6, -4.2) {\tiny $u_4^{1}$};
\node[draw, circle, fill=orange!10, inner sep=1pt] (u_4^{2}) at (1*5.0-1.0, -4.9) {\tiny $u_4^{2}$};
\node[draw, circle, fill=orange!10, inner sep=1pt] (u_4^{3}) at (1*5.0-1.4, -5.6) {\tiny $u_4^{3}$};
\node[draw, circle, fill=orange!10, inner sep=1pt] (u_6^{1}) at (1*5.0-1.7, -6.9) {\tiny $u_6^{1}$};
\node[draw, circle, fill=orange!10, inner sep=1pt] (u_6^{2}) at (1*5.0-2.2, -7.6) {\tiny $u_6^{2}$};
\node[draw, circle, fill=orange!10, inner sep=1pt] (u_6^{3}) at (1*5.0-2.6, -8.3) {\tiny $u_6^{3}$};
\node[draw, circle, fill=blue!10, ] (S_1) at (2*5.0, -0.0) {$\ns_1$};
\node[draw, circle, fill=orange!10, inner sep=1pt] (u_2^{1}) at (2*5.0, -1.5) {\tiny $u_2^{1}$};
\node[draw, circle, fill=orange!10, inner sep=1pt] (u_2^{2}) at (2*5.0-0.4, -2.2) {\tiny $u_2^{2}$};
\node[draw, circle, fill=orange!10, inner sep=1pt] (u_2^{3}) at (2*5.0-0.8, -2.9) {\tiny $u_2^{3}$};
\node[draw, circle, fill=orange!10, inner sep=1pt] (u_5^{1}) at (2*5.0-0.6, -4.2) {\tiny $u_5^{1}$};
\node[draw, circle, fill=orange!10, inner sep=1pt] (u_5^{2}) at (2*5.0-1.0, -4.9) {\tiny $u_5^{2}$};
\node[draw, circle, fill=orange!10, inner sep=1pt] (u_5^{3}) at (2*5.0-1.4, -5.6) {\tiny $u_5^{3}$};
\node[draw, circle, fill=blue!10, ] (S_3) at (3*5.0, -0.0) {$\ns_3$};
\node[draw, circle, fill=orange!10, inner sep=1pt] (u_1^{1}) at (3*5.0, -1.5) {\tiny $u_1^{1}$};
\node[draw, circle, fill=orange!10, inner sep=1pt] (u_1^{2}) at (3*5.0-0.4, -2.2) {\tiny $u_1^{2}$};
\node[draw, circle, fill=orange!10, inner sep=1pt] (u_1^{3}) at (3*5.0-0.8, -2.9) {\tiny $u_1^{3}$};
\begin{scope}[on background layer]
\draw[ rounded corners=3pt, thick, orange, fill=yellow!5, opacity=0.4] (5.6, -1.3) -- (4.9, -0.9) -- (3.6, -3.1) -- (4.3, -3.5) -- cycle;
\draw[ rounded corners=3pt, thick, orange, fill=yellow!5, opacity=0.4] (5.0, -4.0) -- (4.3, -3.6) -- (3.0, -5.8) -- (3.7, -6.2) -- cycle;
\draw[ rounded corners=3pt, thick, orange, fill=yellow!5, opacity=0.4] (3.9, -6.7) -- (3.2, -6.3) -- (1.8, -8.5) -- (2.5, -8.9) -- cycle;
\draw[ rounded corners=3pt, thick, orange, fill=yellow!5, opacity=0.4] (10.6, -1.3) -- (9.9, -0.9) -- (8.6, -3.1) -- (9.3, -3.5) -- cycle;
\draw[ rounded corners=3pt, thick, orange, fill=yellow!5, opacity=0.4] (10.0, -4.0) -- (9.3, -3.6) -- (8.0, -5.8) -- (8.7, -6.2) -- cycle;
\draw[ rounded corners=3pt, thick, orange, fill=yellow!5, opacity=0.4] (15.6, -1.3) -- (14.9, -0.9) -- (13.6, -3.1) -- (14.3, -3.5) -- cycle;
\draw[-Latex] (S_1) -- (S_3);
\draw[-Latex] (S_1) -- (u_1^{1});
\draw[-Latex] (S_1) -- (u_1^{2});
\draw[-Latex] (S_1) -- (u_1^{3});
\draw[-Latex] (S_2) -- (S_4);
\draw[-Latex] (S_2) -- (u_3^{1});
\draw[-Latex] (S_2) -- (u_3^{2});
\draw[-Latex] (S_2) -- (u_3^{3});
\draw[-Latex] (S_2) -- (u_4^{1});
\draw[-Latex] (S_2) -- (u_4^{2});
\draw[-Latex] (S_2) -- (u_4^{3});
\draw[-Latex] (S_2) -- (u_6^{1});
\draw[-Latex] (S_2) -- (u_6^{2});
\draw[-Latex] (S_2) -- (u_6^{3});
\draw[-Latex] (S_4) -- (S_1);
\draw[-Latex] (S_4) -- (u_2^{1});
\draw[-Latex] (S_4) -- (u_2^{2});
\draw[-Latex] (S_4) -- (u_2^{3});
\draw[-Latex] (S_4) -- (u_5^{1});
\draw[-Latex] (S_4) -- (u_5^{2});
\draw[-Latex] (S_4) -- (u_5^{3});
\end{scope}
\end{tikzpicture}

%% file: figures/shellings.tex
\begin{tikzpicture}[scale=0.7]
\node[draw, circle, fill=green!10] (r) at (0, 10.5 - 1.5*1.5) {$r$};
\node[draw, circle, fill=blue!10] (s0) at (2, 7*1.5- 1.5*0) {\small $s_1$};
\node[draw, circle, fill=blue!10] (s1) at (2, 7*1.5- 1.5*1) {\small $s_2$};
\node[draw, circle, fill=blue!10] (s2) at (2, 7*1.5- 1.5*2) {\small $s_3$};
\node[draw, circle, fill=blue!10] (s3) at (2, 7*1.5- 1.5*3) {\small $s_4$};
\node[draw, circle, fill=orange!10, inner sep=0pt] (e1_0) at (7.5, 12.93) {\tiny $u_1^{1}$};
\node[draw, circle, fill=orange!10, inner sep=0pt] (e1_1) at (7.19, 12.379999999999999) {\tiny $u_1^{2}$};
\node[draw, circle, fill=orange!10, inner sep=0pt] (e1_2) at (7.81, 12.379999999999999) {\tiny $u_1^{3}$};
\node[draw, circle, fill=orange!30, opacity=0.7, inner sep=2pt] (1_main) at (5.0, 12.579999999999998) {$\nou_1$};
\node[draw, dotted, very thick, circle, fill=yellow!10, opacity=0.4, inner sep=10pt] (1_top) at (7.5, 12.579999999999998) {};
\node[draw, circle, fill=orange!10, inner sep=0pt] (e2_0) at (7.5, 11.18) {\tiny $u_2^{1}$};
\node[draw, circle, fill=orange!10, inner sep=0pt] (e2_1) at (7.19, 10.629999999999999) {\tiny $u_2^{2}$};
\node[draw, circle, fill=orange!10, inner sep=0pt] (e2_2) at (7.81, 10.629999999999999) {\tiny $u_2^{3}$};
\node[draw, circle, fill=orange!30, opacity=0.7, inner sep=2pt] (2_main) at (5.0, 10.829999999999998) {$\nou_2$};
\node[draw, dotted, very thick, circle, fill=yellow!10, opacity=0.4, inner sep=10pt] (2_top) at (7.5, 10.829999999999998) {};
\node[draw, circle, fill=orange!10, inner sep=0pt] (e3_0) at (7.5, 9.43) {\tiny $u_3^{1}$};
\node[draw, circle, fill=orange!10, inner sep=0pt] (e3_1) at (7.19, 8.879999999999999) {\tiny $u_3^{2}$};
\node[draw, circle, fill=orange!10, inner sep=0pt] (e3_2) at (7.81, 8.879999999999999) {\tiny $u_3^{3}$};
\node[draw, circle, fill=orange!30, opacity=0.7, inner sep=2pt] (3_main) at (5.0, 9.079999999999998) {$\nou_3$};
\node[draw, dotted, very thick, circle, fill=yellow!10, opacity=0.4, inner sep=10pt] (3_top) at (7.5, 9.079999999999998) {};
\node[draw, circle, fill=orange!10, inner sep=0pt] (e4_0) at (7.5, 7.68) {\tiny $u_4^{1}$};
\node[draw, circle, fill=orange!10, inner sep=0pt] (e4_1) at (7.19, 7.13) {\tiny $u_4^{2}$};
\node[draw, circle, fill=orange!10, inner sep=0pt] (e4_2) at (7.81, 7.13) {\tiny $u_4^{3}$};
\node[draw, circle, fill=orange!30, opacity=0.7, inner sep=2pt] (4_main) at (5.0, 7.33) {$\nou_4$};
\node[draw, dotted, very thick, circle, fill=yellow!10, opacity=0.4, inner sep=10pt] (4_top) at (7.5, 7.33) {};
\node[draw, circle, fill=orange!10, inner sep=0pt] (e5_0) at (7.5, 5.93) {\tiny $u_5^{1}$};
\node[draw, circle, fill=orange!10, inner sep=0pt] (e5_1) at (7.19, 5.38) {\tiny $u_5^{2}$};
\node[draw, circle, fill=orange!10, inner sep=0pt] (e5_2) at (7.81, 5.38) {\tiny $u_5^{3}$};
\node[draw, circle, fill=orange!30, opacity=0.7, inner sep=2pt] (5_main) at (5.0, 5.58) {$\nou_5$};
\node[draw, dotted, very thick, circle, fill=yellow!10, opacity=0.4, inner sep=10pt] (5_top) at (7.5, 5.58) {};
\node[draw, circle, fill=orange!10, inner sep=0pt] (e6_0) at (7.5, 4.18) {\tiny $u_6^{1}$};
\node[draw, circle, fill=orange!10, inner sep=0pt] (e6_1) at (7.19, 3.63) {\tiny $u_6^{2}$};
\node[draw, circle, fill=orange!10, inner sep=0pt] (e6_2) at (7.81, 3.63) {\tiny $u_6^{3}$};
\node[draw, circle, fill=orange!30, opacity=0.7, inner sep=2pt] (6_main) at (5.0, 3.83) {$\nou_6$};
\node[draw, dotted, very thick, circle, fill=yellow!10, opacity=0.4, inner sep=10pt] (6_top) at (7.5, 3.83) {};
\begin{scope}[on background layer]
\draw[] (r) -- (s0);
\draw[] (r) -- (s1);
\draw[] (r) -- (s2);
\draw[] (r) -- (s3);
\draw[] (s0) -- (1_main);
\draw[very thick, gray!70] (1_main) -- (1_top);
\draw[] (s0) -- (2_main);
\draw[very thick, gray!70] (2_main) -- (2_top);
\draw[] (s0) -- (4_main);
\draw[very thick, gray!70] (4_main) -- (4_top);
\draw[] (s1) -- (3_main);
\draw[very thick, gray!70] (3_main) -- (3_top);
\draw[] (s1) -- (4_main);
\draw[very thick, gray!70] (4_main) -- (4_top);
\draw[] (s1) -- (6_main);
\draw[very thick, gray!70] (6_main) -- (6_top);
\draw[] (s2) -- (1_main);
\draw[very thick, gray!70] (1_main) -- (1_top);
\draw[] (s2) -- (2_main);
\draw[very thick, gray!70] (2_main) -- (2_top);
\draw[] (s2) -- (5_main);
\draw[very thick, gray!70] (5_main) -- (5_top);
\draw[] (s3) -- (2_main);
\draw[very thick, gray!70] (2_main) -- (2_top);
\draw[] (s3) -- (5_main);
\draw[very thick, gray!70] (5_main) -- (5_top);
\draw[] (s3) -- (6_main);
\draw[very thick, gray!70] (6_main) -- (6_top);
\end{scope}
\end{tikzpicture}

%% file: sections/conclusions.tex
We have shown a variety of $\#\mathsf{P}$-hardness results for counting orderings in graphs, posets, and geometric configurations, by leveraging gadget amplification to the limit.

Beyond the bipartite case, our hardness result for $\#\mathsf{SVO}$ also extends to split graphs. 
Indeed, in the construction for rooted SVOs, one may turn the root vertex $r$ together with all set vertices $\ns_i$ into a clique, while keeping the copied element vertices as an independent set. Since, after the root is placed, all set vertices are immediately available, the whole argument applies without change. The reduction from rooted to unrooted SVOs also preserves this structure, as it only adds leaves adjacent to the root.

There is also flexibility in the choice of the source problem for our reductions. In particular, all our reductions from $\#\mathsf{X3C}$ could be adapted to start from $\#\mathsf{X2C}$. The problem $\#\mathsf{X2C}$ is equivalent to counting perfect matchings of a graph, and it is a classical result of Valiant \cite{V79} that the latter problem is $\#\mathsf{P}$-complete. We chose to reduce from $\#\mathsf{X3C}$ because the proof is not significantly more complicated, and $\#\mathsf{X3C}$ is a more robust problem in the sense that it is $\#\mathsf{P}$-complete under parsimonious reductions and remains hard even in restricted settings like planar instances~\cite{X3C}.

The hardness of planar $\#\mathsf{X3C}$ suggests a natural direction: adapting our reductions to preserve planarity. Although planar $\#\mathsf{X3C}$ is $\#\mathsf{P}$-hard, our reductions do not preserve planarity. The first obstruction is the presence of a root vertex adjacent to all set vertices, which appears in several of the constructions and can be incompatible with some planar embedding of the incidence graph. For the reduction to linear extensions of posets of height $2$, where this root vertex is absent, there is a different obstruction: the amplification step creates many copies of the elements $u_j$ with the same neighborhood, and this can create subgraphs isomorphic to $K_{3, 3}$. Thus, the complexity of the planar variants of the counting problems studied here remains open, to the best of our knowledge.

%% file: sections/acknowledgments.tex
We thank Igor Pak for bringing the question of counting peeling
sequences to our attention and David Eppstein for clarifying its
background. The initial ideas for this paper were developed during the
2026 Dagstuhl seminar \emph{Knowledge Compilation in Artificial
Intelligence, Databases, and Formal Methods}. We thank the seminar
organizers and participants for valuable discussions on the counting
problems studied here.

This work was partially supported by the ANID Millennium Science
Initiative Program (code ICN17002) and the DARPA expMath program under
DARPA CMO contract HR0011262E028. The second author was also supported
by ANID (Doctorado Nacional, 2025, folio 21251617).

\section*{AI methodology}

We used ChatGPT 5.5, ChatGPT 5.6 Sol, and Gemini 3.1 Pro to assist with the development of some of the proofs in the paper. For successive vertex orderings, we developed the initial gadget construction and its connection to exact covers on our own. With AI assistance, we obtained an initial continuous argument based on a lemma establishing that certain multiple integrals are rational functions of the amplification parameter. We subsequently reviewed this argument, identified its connection to the hook-length formula, and reformulated it as a discrete combinatorial proof. We then developed the ideas underlying the reductions for height-two posets and shellings on our own. For peeling sequences, we used ChatGPT 5.6 Sol to obtain the initial reduction idea and proof outline. We also used AI assistance to adapt the argument to N-free posets and to obtain a preliminary proof outline for perfect elimination orderings. We subsequently developed and adapted these arguments to the interpolation framework presented in this paper. We reviewed and revised all AI-assisted arguments and take full responsibility for the correctness of the paper.

%% file: sections/appendix.tex
\section{Proofs of Intermediate Results}\label{app:proofs}

\subsection{Proof of Lemma \ref{lemma:tree-le}}
\label{subsec:hook_length}

By induction on $|T|$. For $|T| = 1$ the formula yields $1$, which is correct. For the inductive case, let $r$ be the root of $T$, and let $v_1, \dots, v_d$ be its out-neighbors. Due to the tree structure, for $1 \leq i \neq j \leq d$, all elements of $T_{v_i}$ are incomparable with the elements of $T_{v_j}$. Thus, all linear extensions of $T$ result from placing $r$ first, then taking a linear extension $\tau_i$ for each $T_{v_i}$, and now arranging all elements except for $r$ into an ordering $\pi$ such that projecting $\pi$ to the elements of $T_{v_i}$ leaves them in the order given by $\tau_i$. Therefore,
    \begin{align*}
         \#\mathsf{LE}(T) &= \binom{|T| - 1}{|T_{v_1}|, |T_{v_2}|, \dots, |T_{v_d}|} \prod_{i=1}^d \#\mathsf{LE}(T_{{v_i}})\\
         &=  \frac{(|T|-1)!}{ |T_{v_1}|! \cdot |T_{v_2}|! \cdots |T_{v_d}|!}  \prod_{i=1}^d \frac{|T_{v_i}|!}{\prod_{v \in V(T_{v_i})} |T_v|} \tag{Inductive Hypothesis} \\
         &= \frac{(|T| - 1)!}{\prod_{i=1}^d \prod_{v \in V(T_{v_i})} |T_v|}\\
         &= \frac{(|T| - 1)!}{\prod_{v \in V \setminus \{r \}} |T_v|} \\ &= \frac{|T|!}{\prod_{v \in V} |T_v|}.
    \end{align*}

\subsection{Proof of Lemma \ref{lemma:le-formula}}

 We define a new poset $T_{\sigma}$ with the same elements as $P_{I, q}$, but a new relation $\preceq$, given by the following conditions:
\begin{enumerate}
    \item $\ns_{\sigma(j)} \preceq \ns_{\sigma(j+1)}$ for every $1 \leq j < n$.
    \item If $S_{\sigma(j)}$ is the last set covering element $u$, then $\ns_{\sigma(j)} \preceq u^{(i)}$ for every $i \in \{1, \ldots, q\}$.
\end{enumerate}
Note that $T_{\sigma}$ is a tree poset, and its linear extensions are exactly the valid linear extensions of $P_{I, q}$ that respect $\sigma$. We have that the subtree rooted at $\ns_{\sigma(j)}$ contains exactly $n-j+1$ elements of the form $\ns_{\sigma(k)}$ for $k \geq j$, and $\sum_{k=j}^n q \cdot c_{k}(\sigma)$ elements of the form $u^{(i)}$.
Thus, by~\Cref{lemma:tree-le}, we have that
\begin{align*}
    \#\mathsf{LE}(P_{I, q})_{\sigma} &= \frac{(n+3hq)!}{\prod_{j=1}^n \left((n-j+1) + \sum_{k=j}^n q \cdot c_{k}(\sigma)\right)}\\
    &=  \frac{(n+3hq)!}{\prod_{i=1}^n \left(i + \sum_{j=1}^{i} q\cdot c_{n-j+1}(\sigma)\right)} \tag{change variables}.
\end{align*}

\subsection{Proof of Lemma \ref{lem:general_position}}
Suppose for the sake of contradiction that there are $3$ collinear points in $\points_q$. We need to consider five cases. The fact that all anchors are hull points will follow from the analysis of the first two cases.
\begin{enumerate}
    \item \textbf{At least two of the points are northern anchors.} Notice that all northern anchors belong to the half parabola $y = -x^2+2^{(n+1)^2}+1$ with $x \geq 0$, so the third point cannot be another northern anchor. Observe also that the line $L$ through the two northern anchors with $x$-coordinates $\frac{i}{q}$ and $\frac{j}{q}$ is given by
    \[y = -\frac{i+j}{q}x + \frac{ij}{q^2} + 2^{(n+1)^2}+1.\]
    Now, for any poset point $p_a = (2^{\rho_1(a)},\, 2^{(n+1)\rho_2(a)})$, we have that 
    \begin{align*}
        -\frac{i+j}{q}\cdot{2^{\rho_1(a)}} + \frac{ij}{q^2} + 2^{(n+1)^2}+1 &\geq -\frac{q+q}{q}\cdot 2^{n} + 0 + 2^{(n+1)^2}+1 \\
        &> 2^{(n+1)^2} - 2^{n+1} \\
        &\geq 2^{(n+1)n} \\
        &\geq 2^{(n+1)\rho_2(a)},
    \end{align*}
    so $p_a \not \in L$. On the other hand, any eastern anchor $\eaan_\ell^{(q)} = (2^{n+1}+1-\frac{\ell^2}{2^{(n+1)^2}q^2},\, \frac{\ell}{q})$ has $x$-coordinate $x_0 \leq 2^{n+1}+1$, and thus
        \begin{align*}
        -\frac{i+j}{q}\cdot x_0 + \frac{ij}{q^2} + 2^{(n+1)^2}+1 &> -\frac{q+q}{q}\cdot (2^{n+1}+1) + 0 + 2^{(n+1)^2}+1 \\
        &= 2^{(n+1)^2} - 2^{n+2} - 1 \\
        &> 1 \\
        &\geq \frac{\ell}{q},
    \end{align*}
    so $\eaan_\ell^{(q)} \not \in L$. Therefore, we conclude that this case is impossible. In fact, this also shows that all northern anchors are hull points. Indeed, for each $j\in \{0,\ldots,q-2\}$, every point of $\points_q\setminus\{\noan_j^{(q)},\noan_{j+1}^{(q)}\}$ lies strictly below the line through $\noan_j^{(q)}$ and $\noan_{j+1}^{(q)}$. Hence, each segment between $\noan_j^{(q)}$ and $\noan_{j+1}^{(q)}$ is an edge of the convex hull. Notice also that, if $q=1$, then $\noan_0^{(1)}$ is the highest point of $\points_1$, so in this case it is also a hull point.
    
    \item \textbf{At least two of the points are eastern anchors.} Notice that all eastern anchors belong to the half parabola $x = -\frac{1}{2^{(n+1)^2}}y^2+2^{n+1}+1$ with $y \geq 0$, so the third point cannot be another eastern anchor. Observe also that the line $L$ through the two eastern anchors with $y$-coordinates $\frac{i}{q}$ and $\frac{j}{q}$ is given by
    \[x = -\frac{i+j}{2^{(n+1)^2}q}y + \frac{ij}{2^{(n+1)^2}q^2} + 2^{n+1}+1.\]
    Now, for any poset point $p_a = (2^{\rho_1(a)},\, 2^{(n+1)\rho_2(a)})$, we have that 
    \begin{align*}
        -\frac{i+j}{2^{(n+1)^2}q} \cdot 2^{(n+1)\rho_2(a)} + \frac{ij}{2^{(n+1)^2}q^2} + 2^{n+1}+1
        &\geq -\frac{q+q}{2^{(n+1)^2}q}\cdot 2^{(n+1)n} + 0 + 2^{n+1}+1 \\
        &\geq 2^{n+1} + 1 - \frac{1}{2^n}\\
        &\geq 2^{n+1} \\
        &> 2^{\rho_1(a)},
    \end{align*}
    so $p_a \not \in L$. On the other hand, any northern anchor $\noan_\ell^{(q)} = (\frac{\ell}{q},\, 2^{(n+1)^2}+1-\frac{\ell^2}{q^2})$ has $y$-coordinate $y_0 \leq 2^{(n+1)^2}+1$, and thus
        \begin{align*}
        -\frac{i+j}{2^{(n+1)^2}q} \cdot y_0 + \frac{ij}{2^{(n+1)^2}q^2} + 2^{n+1}+1 &> -\frac{q+q}{2^{(n+1)^2}q}\cdot (2^{(n+1)^2}+1) + 0 + 2^{n+1}+1 \\
        &= 2^{n+1} -1 - \frac{1}{2^{(n+1)^2-1}} \\
        &> 1 \\
        &\geq \frac{\ell}{q},
    \end{align*}
    so $\noan_\ell^{(q)} \not \in L$. Therefore, we conclude that this case is impossible. In fact, this also shows that all eastern anchors are hull points. Indeed, for each $j\in \{0,\ldots,q-2\}$, every point of $\points_q\setminus\{\eaan_j^{(q)},\eaan_{j+1}^{(q)}\}$ lies strictly to the left of the line through $\eaan_j^{(q)}$ and $\eaan_{j+1}^{(q)}$. Hence, each segment between $\eaan_j^{(q)}$ and $\eaan_{j+1}^{(q)}$ is an edge of the convex hull. Notice also that, if $q=1$, then $\eaan_0^{(1)}$ is the rightmost point of $\points_1$, so in this case it is also a hull point.
    
    \item \textbf{Exactly two of the points are poset points.} Let $p_a$ and $p_b$ be those poset points, and assume without loss of generality that $\rho_1(a) < \rho_1(b)$. Observe that the line $L$ through $p_a$ and $p_b$ is given by
    \begin{equation} \label{eq:slope}
    y = \frac{2^{(n+1)\rho_2(b)}-2^{(n+1)\rho_2(a)}}{2^{\rho_1(b)}-2^{\rho_1(a)}}\left(x-2^{\rho_1(a)}\right) + 2^{(n+1)\rho_2(a)}.
    \end{equation}
    Let us write $m(a, b)$ to denote the slope in equation~\eqref{eq:slope}. Notice that 
    \[\frac{2^{2(n+1)}-2^{n+1}}{2^n-2} \leq |m(a, b)| \leq \frac{2^{n(n+1)}-2^{n+1}}{4-2}.\]
    Equation~\eqref{eq:slope} can be rewritten as
    \begin{equation*}
    x = \frac{1}{m(a,b)}\left(y-2^{(n+1)\rho_2(a)}\right) + 2^{\rho_1(a)}.
    \end{equation*}
    Now, any northern anchor $\noan_j^{(q)} = (\frac{j}{q},\, 2^{(n+1)^2}+1-\frac{j^2}{q^2})$ has $y$-coordinate $y_0 \geq 2^{(n+1)^2}$, but
    \begin{align*}
    m(a, b) \left(\frac{j}{q}-2^{\rho_1(a)}\right) + 2^{(n+1)\rho_2(a)} &\leq \frac{2^{n(n+1)}-2^{n+1}}{4-2} \cdot 2^{n} + 2^{(n+1)n} \\
    &< 2^{n^2+2n-1} + 2^{(n+1)n} \\
    &< 2^{(n+1)^2} \\
    &\leq y_0,
    \end{align*}
    so $\noan_j^{(q)} \not \in L$. Now let $\eaan_j^{(q)} = (2^{n+1}+1-\frac{j^2}{2^{(n+1)^2}q^2},\, \frac{j}{q})$ be an eastern anchor. Recall that $\rho_1(a) < \rho_1(b)$, and so $\rho_1(a) \leq n-1$. Notice that $\eaan_j^{(q)}$ has $x$-coordinate $x_0 \geq 2^{n+1}$, but
    \begin{align*}
    \frac{1}{m(a,b)}\left(\frac{j}{q}-2^{(n+1)\rho_2(a)}\right) + 2^{\rho_1(a)} &\leq \frac{(2^{\rho_1(b)}-2^{\rho_1(a)}) \cdot 2^{(n+1)\rho_2(a)}}{|2^{(n+1)\rho_2(b)}-2^{(n+1)\rho_2(a)}|} + 2^{n-1} \\
    &= \frac{2^{\rho_1(b)}-2^{\rho_1(a)}}{|2^{(n+1)(\rho_2(b)-\rho_2(a))}-1|} + 2^{n-1} \\
    &\leq \frac{2^{\rho_1(b)}-2^{\rho_1(a)}}{1-2^{-(n+1)}} + 2^{n-1} \\
    &< \frac{2^n}{1-2^{-(n+1)}} + 2^{n-1} \\
    &\leq \frac{4}{3} \cdot 2^n + 2^{n-1} \\
    &< 2^{n+1} \\
    &\leq x_0,
    \end{align*}
    and so $\eaan_j^{(q)} \not \in L$. Therefore, we conclude that this case is impossible.
    
    \item \textbf{The three points are poset points.} Let $p_a$, $p_b$, and $p_c$ be those poset points, and assume without loss of generality that $\rho_1(a) < \rho_1(b) < \rho_1(c)$. If the points are collinear, then we have the following equality between slopes:
    \begin{equation} \label{eq:slopes}
    \frac{2^{(n+1)\rho_2(b)}-2^{(n+1)\rho_2(a)}}{2^{\rho_1(b)}-2^{\rho_1(a)}} = \frac{2^{(n+1)\rho_2(c)}-2^{(n+1)\rho_2(b)}}{2^{\rho_1(c)}-2^{\rho_1(b)}}.
    \end{equation}
    We only need to consider the cases when $\rho_2(a) < \rho_2(b) < \rho_2(c)$ or when $\rho_2(a) > \rho_2(b) > \rho_2(c)$, as in any other case equation~\eqref{eq:slopes} will fail because of a sign mismatch. Suppose first that $\rho_2(a) < \rho_2(b) < \rho_2(c)$. Then,
    \[\frac{2^{(n+1)\rho_2(b)}-2^{(n+1)\rho_2(a)}}{2^{\rho_1(b)}-2^{\rho_1(a)}} < \frac{2^{(n+1)\rho_2(b)}}{2},\]
    whereas
    \[\frac{2^{(n+1)\rho_2(c)}-2^{(n+1)\rho_2(b)}}{2^{\rho_1(c)}-2^{\rho_1(b)}} \geq \frac{2^{(n+1)\rho_2(b)}\left(2^{(n+1)(\rho_2(c)-\rho_2(b))} - 1\right)}{2^n} \geq 2^{(n+1)\rho_2(b)},\]
    which leads to a contradiction. 
    Now suppose that $\rho_2(a) > \rho_2(b) > \rho_2(c)$. Taking absolute values of the slopes, we have that
    \[\frac{2^{(n+1)\rho_2(a)}-2^{(n+1)\rho_2(b)}}{2^{\rho_1(b)}-2^{\rho_1(a)}} \geq \frac{2^{(n+1)\rho_2(b)}\left(2^{(n+1)(\rho_2(a)-\rho_2(b))} - 1\right)}{2^n} \geq 2^{(n+1)\rho_2(b)},\]
    whereas
    \[\frac{2^{(n+1)\rho_2(b)}-2^{(n+1)\rho_2(c)}}{2^{\rho_1(c)}-2^{\rho_1(b)}} < \frac{2^{(n+1)\rho_2(b)}}{2}.\]
    Therefore, we conclude that this case is impossible.

    \item \textbf{There is a point of each type.} Let those points be $\noan_i^{(q)}$, $p_a$ and $\eaan_j^{(q)}$. Consider the line $L$ given by $2^{n(n+1)}x+y = 2^{(n+1)^2}$. Notice that $p_a$ lies below $L$, while $\noan_i^{(q)}$ and $\eaan_j^{(q)}$ lie above $L$. However, $p_a$ is to the right of the northern anchor $\noan_i^{(q)}$ and to the left of the eastern anchor $\eaan_j^{(q)}$, so it is impossible for these three points to be collinear.
\end{enumerate}

\subsection{Proof of Lemma \ref{lemma:minimal_elements}}
First, suppose that there exists another poset point $p_b \in S$ such that $b \preceq a$. Let us write $p_a = (x_a,\, y_a)$ and $p_b = (x_b,\, y_b)$. We have that $x_b < x_a$ and that $y_b < y_a$, so $p_a$ lies strictly northeast of $p_b$. Let $\noan_i^{(q)} \in S$ be a northern anchor, and let $\eaan_j^{(q)} \in S$ be an eastern anchor.
Consider the line $L$ through $p_b$ and $p_a$. Notice that $\noan_i^{(q)}$ lies strictly to one side of $L$, since it is northwest of $p_b$, while $\eaan_j^{(q)}$ lies strictly to the other side, since it is southeast of $p_b$. Hence, the segment connecting both anchors intersects $L$ at some point $z$.
Moreover, $z$ lies on the ray starting at $p_b$ and passing through $p_a$, since the segment connecting both anchors lies strictly above the line given by $2^{n(n+1)}x+y = 2^{(n+1)^2}$, whereas $p_b$ lies below it, and $2^{n(n+1)}x+y$ increases along $L$ in the direction from $p_b$ towards $p_a$, which is northeast. Furthermore, $p_a$ also lies below that line, while $z$ lies above it, so $p_a$ lies strictly between $p_b$ and $z$. Therefore, 
$p_a$ is in the convex hull of $\left\{\noan_i^{(q)},\, \eaan_j^{(q)},\, p_b\right\}$,
and so $p_a$ is not a hull point. 

Conversely, suppose that there is no other $p_b\in S$ such that $b\preceq a$. Consider the line $L$ given by
\[
\frac{x}{2^{\rho_1(a)}} + \frac{y}{2^{(n+1)\rho_2(a)}} =2.
\]
Notice that the point $p_a$ lies on $L$.
Let $p_b\in S$ be any other poset point. Since it is not true that $b\preceq a$, either $\rho_1(b) > \rho_1(a)$ or $\rho_2(b) > \rho_2(a)$.
In the first case, we have that
\[
\frac{2^{\rho_1(b)}}{2^{\rho_1(a)}} \geq 2,
\]
while the $y$-coordinate of $p_b$ is strictly positive. In the second case, we have that
\[
\frac{2^{(n+1)\rho_2(b)}}
{2^{(n+1)\rho_2(a)}}
\geq 2^{n+1} > 2.
\]
Thus, in either case, $p_b$ lies above $L$.
Now observe that the $y$-coordinate of any northern anchor is greater than $2^{(n+1)^2}$, and
\[
\frac{2^{(n+1)^2}}{2^{(n+1)\rho_2(a)}} \geq 2^{n+1} > 2.
\]
Hence, all northern anchors lie above $L$.
Finally, the $x$-coordinate of any eastern anchor is strictly greater than $2^{n+1}$, and 
\[
\frac{2^{n+1}}{2^{\rho_1(a)}} \geq 2.
\]
Consequently, all eastern anchors also lie above $L$. This shows that $p_a$ is indeed a hull point, thus completing the proof.

\subsection{Proof of Lemma \ref{lemma:anchors_indist}}
In the third case of the proof of \Cref{lem:general_position}, we proved that all northern anchors lie above any line that passes through two distinct poset points. Hence, all northern anchors lie in the same cell of the arrangement formed by those lines. By Carathéodory's theorem, a poset point $p\in S$ fails to be a hull point of $S\cup T$, for a nonempty set $T$ of northern anchors, exactly when $p$ lies in a triangle spanned by other points of $S\cup T$; such a triangle has at most one anchor, since the northern anchors lie above all poset points and on the same side of the line through $p$ and any other poset point. Consequently, adding any nonempty subset of northern anchors has the same effect on which poset points are vertices of the convex hull. The argument for the eastern anchors is the same, as we also showed that all those anchors lie on the right side of any line through two distinct poset points.

\subsection{Proof of Lemma \ref{lemma:peels-formula}}
First, suppose that $m \neq k$, so that $a < b$. Let us call the anchors that are exhausted first \emph{type $A$} anchors, and the anchors that are exhausted last \emph{type $B$} anchors.
Notice that for a peeling sequence to satisfy the conditions given by the triple $(\sigma, m, k)$, the ordering of the poset points, as well as the relative position (with respect to those poset points) of the last northern anchor and the last eastern anchor, are prescribed. Conversely, by~\Cref{lem:general_position,lemma:minimal_elements,lemma:anchors_indist}, whether each point is a hull point when it is removed depends only on these prescriptions; since $(\sigma, m, k) \in \ompe$, every ordering that satisfies them is therefore a peeling sequence. Now, there are $\binom{a+q-1}{q-1}$ ways to insert the other type $A$ anchors, and after that, there are $\binom{b+2q-1}{q-1}$ ways to insert the other type $B$ anchors. We also need to consider the possible orderings of the anchors of the same type, and so we need to multiply by an additional $(q!)^2$ factor.
Hence, in this case, we have 
\[
\peel(\points_q)_{\sigma, m, k} = (q!)^2\binom{a+q-1}{q-1} \binom{b+2q-1}{q-1}.
\]
If $a=b$, we need to add an additional factor of $2$ because it does not matter which of the two types of anchors is exhausted first. Therefore, in general, we have that
\[
\peel(\points_q)_{\sigma, m, k} = \left(1+\mathds{1}_{m=k} \right) (q!)^2 \binom{a+q-1}{q-1} \binom{b+2q-1}{q-1}.
\]
Finally, we have that
\begin{align*}
\frac{\peel(\points_q)_{\sigma, m, k}}{(n+2q)!} &= \frac{\left(1+\mathds{1}_{m=k} \right) (q!)^2}{(2q+n)!} \cdot \frac{(q+a-1)!}{a! (q-1)!} \cdot \frac{(2q+b-1)!}{(q+b)!(q-1)!} \\
&= \frac{\left(1+\mathds{1}_{m=k} \right) q^2 }{a!} \cdot
\frac{(q+a-1)!}{(q+b)!} \cdot \frac{(2q+b-1)!}{(2q+n)!} \\
&= \frac{\left(1+\mathds{1}_{m=k} \right)q^2}{a! \left(\prod_{j = a}^b{(q+j)} \right) \left( \prod_{j = b}^{n}{(2q+j)}\right)}.
\end{align*}

\subsection{Proof of Proposition \ref{prop:peels_interpolation}}
Observe that the polynomial
\[
    D(q) := n! \left(\prod_{j = 0}^{n}{(q+j)} \right) \left(\prod_{j = 0}^{n}{(2q+j)} \right)
\]
is a multiple of the denominator of $\chi_{\sigma, m, k}(q)$ for every $(\sigma, m, k) \in \ompe$. Moreover, $D(q)$ has degree $2n+2$. Therefore, 
\[
    \sum_{(\sigma, m, k) \in \ompe} D(q) \cdot \chi_{\sigma, m, k}(q) = D(q) \cdot \sum_{(\sigma, m, k) \in \ompe} \chi_{\sigma, m, k}(q)
\] is also a polynomial, $R(q)$, of degree at most $2n+2$, since each summand has degree \[(2n+2) + 2 - (n-\min(m, k)+2),\] which is at most $2n+2$. Thus, we can obtain its coefficients by first evaluating $\frac{\peel(\points_q)}{(n+2q)!}$ at the points $1, 2, \ldots, 2n+3$ (via the oracle) and $D$ at the same points (directly, since we know $D$ explicitly), then applying Lagrange interpolation, and finally putting the result into canonical form. All these steps can be done in polynomial time in terms of $n$.
Notice that $D$ can also be put into canonical form in polynomial time.

\subsection{Proof of Proposition \ref{prop:peel_limit}}
As a polynomial in $q$, the denominator of $\chi_{\sigma,m,k}(q)$ has degree
    \[
        (\max(m,k)-\min(m,k)+1)+(n-\max(m,k)+1)=n-\min(m,k)+2,
    \]
whereas its numerator has degree $2$.
Thus, for every $(\sigma,m,k)\in\ompe$ with $(m,k)\neq(n,n)$,
we have
\[
    \lim_{q\to\infty}\chi_{\sigma,m,k}(q)=0.
\]

We claim that $(\sigma,n,n)\in\ompe$ if and only if $\sigma$ is a linear extension of $P$.
First, suppose that $(\sigma,n,n)\in\ompe$. At least one anchor of each type remains present throughout the peeling of the poset points. By \Cref{lemma:minimal_elements},
each poset point peeled must correspond to a minimal element of the remaining poset. Thus, $\sigma$ is a linear extension of $P$.
Conversely, suppose that $\sigma$ is a linear extension of $P$. We can peel the poset points in the order prescribed by $\sigma$, keeping all anchors until the end.
At each step, the next poset element is minimal among those remaining, so its point is a hull point by \Cref{lemma:minimal_elements}. The remaining anchors can then be peeled in any order, since they are hull points by \Cref{lem:general_position} and remain so under deletions. Hence, $(\sigma,n,n)\in\ompe$.

Since $\ompe$ is finite and does not depend on $q$, we obtain
\begin{align*}
\lim_{q \to \infty} \sum_{(\sigma, m, k) \in \ompe} \chi_{\sigma, m, k}(q) &= \sum_{(\sigma, m, k) \in \ompe} \lim_{q \to \infty}{\chi_{\sigma, m, k}(q)} \\
&= \sum_{\sigma \in \LE(P)} \lim_{q \to \infty}{\chi_{\sigma, n, n}(q)} \\
&= \sum_{\sigma \in \LE(P)} \lim_{q \to \infty}{\frac{2q^2}{n! (q+n) (2q+n)}} \\
&= \sum_{\sigma \in \LE(P)} {\frac{1}{n!} \cdot \left(\lim_{q \to \infty} \frac{q}{q+n}\right) \cdot \left(\lim_{q \to \infty} \frac{2q}{2q+n}\right)} \\
&= \sum_{\sigma \in \LE(P)} {\frac{1}{n!}} \\
&= \frac{\#\mathsf{LE}(P)}{n!}.
\end{align*}

\subsection{Proof of Theorem \ref{thm:peeling}}
For membership in $\#\mathsf{P}$, we need to be able to verify in polynomial time whether an ordering of a configuration of rational points in the plane is a peeling sequence. It is known that the convex hull of a set of points in the plane can be computed in polynomial time \cite{de2000computational}. As the integer arithmetic operations on numerators and denominators involved in that computation can be calculated exactly in polynomial time in terms of their bit lengths, the computation of the convex hull runs in polynomial time even when the coordinates of the points are binary-encoded rational numbers of polynomial bit length. We can thus compute the convex hull after each vertex deletion and verify that the next point is a hull point.

For $\#\mathsf{P}$-hardness, we now have a Turing reduction from the problem of counting linear extensions of posets of dimension $2$, which is known to be $\#\mathsf{P}$-hard \cite[Theorem~1.3]{dittmer2018counting}. Given such a poset $P$, by \Cref{prop:peels_interpolation} and \Cref{prop:peel_limit}, we can compute in polynomial time two polynomials $R$ and $D$ such that
    \[
        \#\mathsf{LE}(P) = n! \cdot \lim_{q \to \infty} \frac{R(q)}{D(q)}.
    \] 
Since $\frac{R(q)}{D(q)}$ is a ratio of polynomials, of size polynomial in $|P|$, the limit can be computed symbolically in polynomial time, as it corresponds simply to the ratio of the leading coefficients.

\subsection{Proof of Lemma \ref{lem:le_cover_characterization}}
First, suppose that $\sigma \in \omxc(X)$ for some $X$. It is clear that $C_{i, \sigma} = i$ for $i \in \{1, \dots, 4h\}$. Now notice that the only way for $C_{4h+1, \sigma}$ to be greater than $C_{4h, \sigma}$ is if the edge at position $4n-4h$ according to $\sigma$ is of the form $(r, \ns_k)$, where $S_k \notin X$. However, that is not possible since the $3$ edges of $E'$ whose starting vertex is $\ns_k$ do not belong to $E_X$, but they must occur after the edge $(r, \ns_k)$ because $\sigma \in \omle$. This shows that $C_{4h+1, \sigma} = 4h$. Also, since $C_{i+1, \sigma} - C_{i, \sigma} \in \{0, 1\}$ for every $i \in [4n-1]$, it follows that $C_{i, \sigma} < i$ for every $i \in \{4h+1, \dots, 4n\}$.

For the second part, suppose that $\sigma \in \omle$ does not induce an exact cover. For the sake of contradiction, suppose that $C_{i, \sigma} = i$ for every $i \in [4h]$. 
Let $X \subseteq \mathcal{S}$ be the set of elements $S_i$ for which the edge $(r, s_i)$ is among the last $4h$ edges according to $\sigma$. Note that, since $C_{i, \sigma} = i$ for every $i \in [4h]$ and $|U| = 3h$, necessarily $|X| \geq h$. Also, there are at least $3h$ edges of the form $(\ns_i, \nou_j)$ for $S_i \in X$, $u_j \in S_i$, and, since $\sigma \in \omle$, all of them must also be at the last $4h$ positions according to $\sigma$. This forces $|X| = h$. Now, because we are assuming that $C_{i, \sigma} = i$ for every $i \in [4h]$, each of those $3h$ edges must be the last incoming edge of the vertex $\nou_j$ that it enters. This means that $X$ is an exact cover and that $\sigma$ induces it, a contradiction. Therefore, it must be the case that $|B_\sigma| > 0$.

\subsection{Proof of Lemma \ref{lem:nfree_claim}}
For ease of notation, let $\mathcal{S} \setminus X = \{S_1, \dots, S_{n-h}\}$.
Note that $\sigma \in \omxc(X)$ if and only if $\sigma$ consists of a linear extension of the edge poset of $E' \setminus E_X$ (in the first $4n-4h$ positions) followed by a linear extension of the edge poset of $E_X$ (in the last $4h$ positions). The edge poset of $E_X$ decomposes into $h$ disjoint posets of $4$ elements, each one of them consisting of edges of the form $(r, s_i)$, $(s_i, \nou_{j_1})$, 
$(s_i, \nou_{j_2})$,
$(s_i, \nou_{j_3})$, where $S_i = \{u_{j_1}, u_{j_2}, u_{j_3}\}$. To have a valid linear extension of such a poset with 4 elements, the only condition that must be satisfied is for $(r, s_i)$ to be its first element. Each of those $h$ posets has $3!$ linear extensions. Therefore, the number of linear extensions of the edge poset of $E_X$ is
$$\frac{(4h)!}{(4!)^h} \cdot (3!)^h = \frac{(4h)!}{4^h}.$$

Let us write $\Gamma_X$ to denote the set of linear extensions of the edge poset of $E' \setminus E_X$. Note that, given $\sigma \in \omxc(X)$, the terms $C_{i, \sigma}$ for $i \in \{4h+1, \dots, 4n\}$ only depend on the first $4n-4h$ edges of $\sigma$. Therefore, we will use the notation $C_{i, \tau}$, where $\tau \in \Gamma_X$. We thus have
    \[
    \sum_{\sigma \in\omxc(X)} \prod_{i=4h+1}^{4n} \frac{1}{i-C_{i, \sigma}} = \frac{(4h)!}{4^h} \sum_{\tau \in\Gamma_X} \prod_{i=4h+1}^{4n} \frac{1}{i-C_{i, \tau}}.
    \]
Therefore, our claim reduces to proving that
      \begin{equation}\label{eq:claim_core_le}
          \sum_{\tau \in \Gamma_X} \prod_{i=4h+1}^{4n} \frac{1}{i-C_{i, \tau}} = \frac{1}{3^{n-h}}.
      \end{equation}
We prove~equation~\eqref{eq:claim_core_le} through a probabilistic interpretation.

For every $i \in \{1, \dots, n-h\}$, let $r_i$ be the edge $(r, \ns_i)$. For each set $S_i \in \mathcal{S} \setminus X$, let $u_i^1, u_i^2, u_i^3$ be an enumeration of its $3$ elements. For $1 \leq k \leq n-h$ and $1 \leq j \leq 3$ let $g_{k}^j$ denote the edge $(\ns_k, \nou_k^j)$, where $\nou_k^j$ is the vertex corresponding to the element $u_k^j$.
Observe that $\tau$ is simply an ordering of the $4n-4h$ edges $r_1, \dots, r_{n-h}, g_1^1, \dots, g_{n-h}^3$ where each $g_{k}^j$ must appear after $r_k$. Consider the following random process $\mathcal{D}$ to generate an ordering of $E' \setminus E_X$:
    \begin{enumerate}
    
        \item Initialize an empty sequence $\pi_0$, a counter $Q_0 \gets 3n-3h$, $A_0 \gets \varnothing$, and ${R_0 \gets \{r_1, \ldots, r_{n-h}\}}$.
        \item Repeat the following for $t \in \{1, \dots, 4n-4h\}$: 
        \begin{enumerate}
            \item Sample an element $e_t$ from $R_{t-1} \cup A_{t-1}$ according to the following distribution:
            \[
            \Pr[e_t = x] = \begin{cases}
                3/Q_{t-1} & \text{if } x \in R_{t-1}\\
                1/Q_{t-1} & \text{if } x \in A_{t-1}.
            \end{cases}
            \]
            \item Obtain $\pi_t$ by appending $e_t$ to the end of $\pi_{t-1}$.
            \item If $e_t = r_k \in R_{t-1}$ for some $k$, then set\\ $R_t \gets R_{t-1} \setminus \{e_{t}\}$, $Q_t \gets Q_{t-1}$, and $A_t \gets A_{t-1} \cup \{ g_k^1, g_k^2, g_k^3\}$. 
            \item If $e_t \in A_{t-1}$, then set $R_t \gets R_{t-1}$,  $Q_{t} \gets Q_{t-1} - 1$, and $A_t \gets A_{t-1} \setminus \{e_t\}$.
        \end{enumerate}
        \item Output $\pi_{4n-4h}$.
    \end{enumerate}
    Step (2.a) is well defined due to the following invariant: for every $t$, we have that $3|R_t| + |A_t| = Q_t$, which holds directly for $t = 0$, and then inductively for $t \geq 1$. If $e_t \in R_{t-1}$, we have 
    \[
    3|R_t| + |A_t| = 3(|R_{t-1}|- 1) + |A_{t-1}| + 3 \overset{\mathrm{I.H.}}{=} Q_{t-1} = Q_t.
    \]
    On the other hand, if $e_t \in A_{t-1}$, we have
    \[
    3|R_t| + |A_t| = 3|R_{t-1}| + |A_{t-1}| - 1 \overset{\mathrm{I.H.}}{=} Q_{t-1} - 1 = Q_t.
    \]
    As the invariant holds, we conclude that
    \[
    \sum_{x \in R_{t-1} \cup A_{t-1}} \Pr[e_t = x] = |R_{t-1}| \cdot \frac{3}{Q_{t-1}} + |A_{t-1}| \cdot \frac{1}{Q_{t-1}} = \frac{3|R_{t-1}| + |A_{t-1}|}{Q_{t-1}} = 1.
    \]
    The process $\mathcal{D}$ guarantees that each $g_k^j$ appears after $r_k$ in $\pi_{4n-4h}$, since $g_k^j$ is only added to $A_{t}$ once $r_k$ has been added to $\pi_t$. Thus, the support of $\mathcal{D}$ is exactly $\Gamma_X$ and therefore
    \begin{equation}\label{eq:distr_le}
    \sum_{\tau \in \Gamma_X} \Pr_{\mathcal{D}}[\pi_{4n-4h} = \tau] = 1.
      \end{equation}
    Using equation~\eqref{eq:distr_le}, we immediately  get~equation~\eqref{eq:claim_core_le} if we  prove  that for every $\tau \in \Gamma_X$,
    \begin{equation}\label{eq:prob_le}
        \Pr_{\mathcal{D}}[\pi_{4n-4h} = \tau] = 3^{n-h} \prod_{i=4h+1}^{4n} \frac{1}{i-C_{i, \tau}}.
    \end{equation}
Thus, our proof reduces to showing~equation~\eqref{eq:prob_le}.
    First, we have
    \[
     \Pr_{\mathcal{D}}[\pi_{4n-4h} = \tau] =  \prod_{t=1}^{4n-4h} \Pr_{\mathcal{D}}\left[\pi_{4n-4h}(t) = \tau(t) \, \Big|\, \bigwedge_{j=1}^{t-1 }\pi_{4n-4h}(j) = \tau(j)\right].
    \]
We now claim that, if the first $t \geq 0$ steps of $\mathcal{D}$ coincide with $\tau \in \Gamma_X$, then it must be the case that
\begin{equation} \label{eq_rec_qt}
Q_{t} = 4n-t - C_{4n-t, \tau}.
\end{equation}
Let $I_\tau \subseteq \{1, \dots, 4n-4h\}$ be the set of positions $t$ for which $\tau(t)$ is an edge of the form $r_k$. Note that, since the final edge set $E_X$ corresponds to an exact cover, we have that $C_{i+1, \tau} - C_{i, \tau} = 1$ for $i \in \{4h+1, \dots, 4n-1\}$ exactly when $4n-i \in I_\tau$.
Observe that initially \[Q_0 = 3n-3h = 4n - 0 - (n+3h) = 4n - 0 - C_{4n-0, \tau}.\] 
Now there are two cases to consider. The first case is when $t \in I_\tau$, and then we have that
$$Q_t = Q_{t-1} \overset{\mathrm{I.H.}}{=} 4n - (t-1) - C_{4n+1-t, \tau} = 4n - (t-1) - (C_{4n-t, \tau}+1) = 4n - t - C_{4n-t, \tau}.$$
The second case is when $t \notin I_\tau$, and then we have that
$$Q_t = Q_{t-1} -1 \overset{\mathrm{I.H.}}{=} 4n - (t-1) - C_{4n+1-t, \tau} - 1= 4n - (t-1) - C_{4n-t, \tau} - 1 = 4n - t - C_{4n-t, \tau}.$$
Therefore, equation~\eqref{eq_rec_qt} follows by induction. 

Now observe that if $t \in I_\tau$, then we have
     \[
     \Pr_{\mathcal{D}}\left[\pi_{4n-4h}(t) = \tau(t) \, \Big|\, \bigwedge_{j=1}^{t-1 }\pi_{4n-4h}(j) = \tau(j)\right] = \frac{3}{Q_{t-1}} = \frac{3}{4n+1-t-C_{4n+1-t, \tau}},
   \]
   and if $t \not\in I_\tau$, then 
     \[
     \Pr_{\mathcal{D}}\left[\pi_{4n-4h}(t) = \tau(t) \, \Big|\, \bigwedge_{j=1}^{t-1 }\pi_{4n-4h}(j) = \tau(j)\right] = \frac{1}{Q_{t-1}} = \frac{1}{4n+1-t-C_{4n+1-t, \tau}}.
   \]
   Because $\left| I_\tau \right|= n-h$, we have that
    \begin{align*}
     \Pr_{\mathcal{D}}[\pi_{4n-4h} = \tau] &=  \prod_{t=1}^{4n-4h} \Pr_{\mathcal{D}}\left[\pi_{4n-4h}(t) = \tau(t) \, \Big|\, \bigwedge_{j=1}^{t-1 }\pi_{4n-4h}(j) = \tau(j)\right] \\ &= 3^{n-h} \prod_{t=1}^{4n-4h} \frac{1}{4n+1-t-C_{4n+1-t, \tau}} \\
     &= 3^{n-h} \prod_{i=4h+1}^{4n} \frac{1}{i-C_{i, \tau}},
    \end{align*}
which is exactly~\eqref{eq:prob_le}. We thus conclude the proof.

\subsection{Proof of Theorem \ref{thm:nfree-height-three}}
Membership in $\#\mathsf{P}$ is straightforward, since an ordering of the elements of a poset can be checked to be a linear extension in polynomial time. For $\#\mathsf{P}$-hardness, let $I$ be an arbitrary instance of $\#\mathsf{X3C}$. Using~\Cref{lem:nfree_claim}, we obtain the following:

\begin{align*}
   \lim_{q \to -1} (q+1)^{4h} \sum_{\sigma \in \omle} \chi_{\sigma}(q) &= \lim_{q \to -1} \sum_{X \in \mathsf{X3C}(I)} \sum_{\sigma \in \omxc(X)} (q+1)^{4h} \chi_{\sigma}(q) \\
    &= \frac{1}{(4h)!} \sum_{X \in \mathsf{X3C}(I)} \sum_{\sigma \in \omxc(X)} \prod_{i=4h+1}^{4n} \frac{1}{i - C_{i, \sigma}} \\
    &= \frac{1}{(4h)!} \sum_{X \in \mathsf{X3C}(I)} \frac{(4h)!}{4^h3^{n-h}} \tag{\Cref{lem:nfree_claim}} \\
    &= 4^{-h} \cdot 3^{h-n} \cdot \#\mathsf{X3C}(I).
\end{align*}
By the same argument as in~\Cref{prop:le_compute-polynomials}, now with denominator factors $Cq+\ell$ for $0\leq C\leq 3h+n$ and $1\leq \ell\leq 4n$, we have that
\[\sum_{\sigma \in \omle} \chi_{\sigma}(q)\]
is equal to a rational function \(\frac{R(q)}{D(q)}\), where \(R(q)\) and \(D(q)\) are polynomials in canonical form that can be computed in polynomial time by evaluating 
$\frac{\#\mathsf{LE}(P_{I, q})}{(4n+3hq + nq)!}$ at polynomially many positive integer values of $q$. Hence, the limit
\[
\lim_{q\to -1} (q+1)^{4h}
\sum_{\sigma \in \omle} \chi_{\sigma}(q) \ = \ 
\lim_{q\to -1} (q+1)^{4h}\frac{R(q)}{D(q)}
\]
can be computed in polynomial time via an oracle for $\#\mathsf{LE}$, which is queried only on the N-free posets $P_{I,q}$ of height~$3$. Notice that this limit can be computed by simply evaluating the numerator and denominator at $q=-1$, except when this evaluation yields the indeterminate form $\frac{0}{0}$, in which case we repeatedly apply L'Hôpital's rule until the resulting limit is no longer of this form. It follows that \(\#\mathsf{X3C}(I)\) can also be computed in polynomial time by multiplying this limit by \(4^h \cdot 3^{n-h}\). This concludes the proof of the theorem.


\subsection{Proof of Lemma \ref{lemma:svo-formula}}
\label{app:proof-svo-formula}
 We define a poset $T_{\sigma} := (V(G_{I, q}) \setminus \{r\}, \preceq)$ whose Hasse diagram is a rooted tree by the following relationships:
\begin{enumerate}
    \item $s_{\sigma(j)} \preceq s_{\sigma(j+1)}$ for every $1 \leq j < n$.
    \item If $S_{\sigma(j)}$ is the earliest set covering element $u$, then $s_{\sigma(j)} \preceq u^{(i)}$ for every $i \in\{1, \ldots, q\}$.
\end{enumerate}
An illustration continuing~\Cref{ex:x3c} is presented in~\Cref{fig:poset-tree}.
Note that $\ns_{\sigma(1)}$ is the root, and that the linear extensions of this poset are exactly the valid $\mathsf{SVO}$s of $G_{I, q}$ (after $r$) that respect $\sigma$. We have that the subtree rooted at $\ns_{\sigma(j)}$ contains exactly $n-j+1$ vertices of the form $\ns_{\sigma(k)}$ for $k \geq j$, and $\sum_{k=j}^n q \cdot c_{k}(\sigma)$ vertices of the form $u^{(i)}$.
Thus, by~\Cref{lemma:tree-le}, we have that
\begin{align*}
    \#\mathsf{SVO}(G_{I, q}; r)_{\sigma} &= \frac{(n+3hq)!}{\prod_{j=1}^n\left((n-j+1) + \sum_{k=j}^n q \cdot c_{k}(\sigma)\right)}\\
    &=  \frac{(n+3hq)!}{\prod_{i=1}^n \left(i + \sum_{j=1}^{i} q\cdot c_{n-j+1}(\sigma)\right)} \tag{changing variables}.
\end{align*}

\subsection{Proof of Lemma \ref{lem:compute-polynomials}}

Note that all terms $i + \sum_{j=1}^i q \cdot c_{n-j+1}(\sigma)$ in the denominator of $\chi_{\sigma}(q)$ are of the form $Cq + \ell$ for integers $0 \leq C \leq |U|$ and $1 \leq \ell \leq n$. Thus, 
the polynomial 
\[
    D(q) := \prod_{C = 0}^{|U|} \prod_{\ell=1}^{n} (Cq + \ell)
\]
is necessarily divisible by the polynomial $\frac{1}{\chi_\sigma(q)}$ for every $\sigma$, making $D(q) \cdot \chi_\sigma(q)$ a polynomial of degree at most $\deg D \leq |U|n$. Therefore,
\[
    \sum_{\sigma \in [n]!} D(q)\cdot \chi_\sigma(q) = D(q) \cdot \sum_{\sigma \in [n]!} \chi_\sigma(q)
\] is also a polynomial, $R(q)$, of degree at most $|U| \cdot n$.  Thus, we can obtain its coefficients by first evaluating $\frac{\#\mathsf{SVO}(G_{I, q}; r)}{(n+3hq)!}$ at the points $1, 2, \ldots, |U|n + 1$ (via the oracle) and $D$ at the same points (directly, since we know $D$ explicitly), then applying Lagrange interpolation, and finally putting the result into canonical form. Notice that $D$ can also be put into canonical form in polynomial time.

\subsection{Proof of Lemma \ref{lem:reduction}}

Define $C_{i, \sigma} := \sum_{j=1}^i c_{n-j+1}(\sigma)$, which corresponds to the number of elements of $U$ whose earliest covering set appears within the last $i$ sets according to $\sigma$. Thus,
\[
    \chi_{\sigma}(q) = \prod_{i=1}^n \frac{1}{i + q \cdot C_{i, \sigma}}.
\]

We first consider the case where the first $h$ sets according to $\sigma$ form an exact cover, so that $C_{i, \sigma} = 0$ when $i \leq n-h$, and $C_{i, \sigma} = 3(i+h-n)$ when $n-h < i \leq n$. Now observe that
\begin{align*}
\lim_{q \to \infty} q^h \chi_{\sigma}(q) &= \lim_{q \to \infty} q^h \prod_{i=1}^n \frac{1}{i + q C_{i, \sigma}} \\ &= \frac{1}{(n-h)!} \cdot \prod_{i=n-h+1}^n\lim_{q \to \infty} \frac{q}{i+3q(i+h-n)} \\
&= \frac{1}{(n-h)!} \cdot \prod_{i=n-h+1}^n\frac{1}{3(i+h-n)} \\
&= \frac{3^{-h}}{(n-h)!h!}.
\end{align*}
We now consider the case where the first $h$ sets according to $\sigma$ do not form an exact cover, so that $C_{i, \sigma} > 0$ for all $i \geq n-h$. Then we have that
\begin{align*}
0 \leq \lim_{q \to \infty} q^h \chi_{\sigma}(q) &= \lim_{q \to \infty} q^h \prod_{i=1}^n \frac{1}{i + q C_{i, \sigma}}\\
&=
\lim_{q \to \infty}\left(\prod_{i=1}^{n-h-1} \frac{1}{i+qC_{i, \sigma}}\right) \cdot \frac{1}{n-h + qC_{n-h, \sigma}}\cdot \left(\prod_{i=n-h+1}^n \frac{q}{i + q C_{i, \sigma}}\right)
\\ &\leq \frac{1}{(n-h-1)!} \cdot \left(\lim_{q \to \infty} \frac{1}{n-h+qC_{n-h, \sigma}} \right) \cdot \prod_{i=n-h+1}^n\lim_{q \to \infty} \frac{q}{i+qC_{i, \sigma}} \\
&= \frac{1}{(n-h-1)!} \cdot 0 \cdot \prod_{i=n-h+1}^n\frac{1}{C_{i, \sigma}} \\
&= 0.
\end{align*}
Let $\omxc \subseteq [n]!$ be the set of orderings $\sigma$ such that the sets $S_{\sigma(1)}, \dots, S_{\sigma(h)}$ form an exact cover. Notice that each exact cover corresponds to exactly $(n-h)!h!$ orderings in $\omxc$. Therefore,
\[
\lim_{q \to \infty} q^h \cdot \sum_{\sigma \in [n]!} \chi_{\sigma}(q) = \sum_{\sigma \in [n]!} \lim_{q \to \infty} q^h \chi_{\sigma}(q) = \sum_{\sigma \in \omxc}{\frac{3^{-h}}{(n-h)!h!}} = 3^{-h} \cdot \#\mathsf{X3C}(I).
\]    

\subsection{Proof of Theorem \ref{thm:main2}}
    Let $I$ be an arbitrary instance of $\#\mathsf{X3C}$.
    All oracle queries are made on the graphs $G_{I,q}$, which are bipartite for every positive integer $q$.
    By~\Cref{lem:compute-polynomials} and~\Cref{lem:reduction}, we can compute in polynomial time polynomials $R$ and $D$ such that
    \[
        \#\mathsf{X3C}(I) = 3^h \cdot \lim_{q \to \infty} \frac{q^h R(q)}{D(q)}.
    \]
    Since $\frac{q^hR(q)}{D(q)}$ is a ratio of polynomials, of size polynomial in the size of $I$, the limit can be computed symbolically in polynomial time, as it is zero when the numerator has smaller degree than the denominator, and otherwise equals the ratio of the leading coefficients.

\subsection{Proof of Proposition \ref{prop:unrooting_SVO}}
    

    Let $(G; r)$ be an instance of $\#\mathsf{SVO}(G; r)$, where $G$ is bipartite, and let $N = |V(G)|$. We construct another bipartite graph $H_\ell$ by adding $\ell$ additional vertices to $G$, all of which are exclusively connected to $r$. We call these vertices the \emph{leaves} of $H_\ell$.
    Now, observe that there are only two kinds of SVOs for $H_\ell$: those that start with a leaf, and those that do not.
    For the first kind, we claim that their number is
    \[
    \ell \cdot \# \mathsf{SVO}(G; r) \cdot \binom{\ell+N-2}{\ell-1} (\ell-1)!.
    \]
    Indeed, there are $\ell$ options for the initial leaf, after which the only possibility is to place $r$. The remaining ordering, when projected down to $V(G)$, must be an SVO of $G$. The $\ell-1$ remaining leaves can be placed in any of the remaining $\ell+N-2$ positions, in any of their $(\ell-1)!$ possible orderings.

    For the second kind, let $\mathsf{SVO}(G)$ be the set of SVOs of $G$, and, for $\pi \in \mathsf{SVO}(G)$, let $\rho_\pi \coloneqq \pi^{-1}(r)$ be the rank of $r$ in $\pi$. Then, $\pi$ can be extended to an SVO of $H_\ell$ by adding the $\ell$ leaves in any of the $\ell+N - \rho_\pi$ positions after $r$, in any of their $\ell!$ orderings. Therefore, the second kind is counted by
    \[
    \sum_{\pi \in \mathsf{SVO}(G)} \binom{\ell+N - \rho_\pi}{\ell} \cdot \ell!.
    \]
    We can separate the case $\rho_\pi = 1$, and obtain
    \[
    \sum_{\pi \in \mathsf{SVO}(G)} \binom{\ell+N - \rho_\pi}{\ell} \cdot \ell! = \# \mathsf{SVO}(G; r) \cdot \binom{\ell+N-1}{\ell} \cdot \ell! + \sum_{\substack{\pi \in \mathsf{SVO}(G)\\ \rho_\pi > 1}} \binom{\ell+N - \rho_\pi}{\ell} \cdot \ell!.
    \]
    Therefore,
    \[
    \frac{\#\mathsf{SVO}(H_\ell)}{\ell!} = \# \mathsf{SVO}(G; r)\left[\binom{\ell+N-2}{\ell-1} +  \binom{\ell+N-1}{\ell} \right] + \sum_{\substack{\pi \in \mathsf{SVO}(G)\\ \rho_\pi > 1}} \binom{\ell+N - \rho_\pi}{\ell}.
    \]
    Now, observe that the right side of the previous equation is a polynomial in $\ell$, whose degree is at most $N-1$. Therefore, by evaluating 
    \[
    \frac{\#\mathsf{SVO}(H_1)}{1!}, \frac{\#\mathsf{SVO}(H_2)}{2!}, \dots, \frac{\#\mathsf{SVO}(H_{N})}{N!}
    \]
    via an oracle for $\#\mathsf{SVO}(\cdot)$, and using Lagrange interpolation, we recover the coefficients of the polynomial in canonical form
    \[
    P(\ell) := \# \mathsf{SVO}(G; r)\left[\binom{\ell+N-2}{\ell-1} +  \binom{\ell+N-1}{\ell} \right] + \sum_{\substack{\pi \in \mathsf{SVO}(G)\\ \rho_\pi > 1}} \binom{\ell+N - \rho_\pi}{\ell}.
    \]
    But the coefficient of $\ell^{N-1}$ in this polynomial is precisely $\frac{2}{(N-1)!}\#\mathsf{SVO}(G; r)$. Therefore, we have shown that polynomially many queries to an oracle for $\#\mathsf{SVO}(\cdot)$ allow us to compute $\#\mathsf{SVO}(G; r)$.

\subsection{Proof of Lemma \ref{lem:shellings-formula}}
Fix $\sigma \in \omsh$, and recall the edges $e_1^\sigma, \dots, e_{4n}^\sigma$ and the values $\rho_\sigma(u_j)$ defined before~\Cref{lem:shellings-formula}. We define a poset $T_{\sigma} := (E_q, \preceq)$, where the order $\preceq$ is given by:
\begin{enumerate}
    \item $e_1^\sigma \preceq e_2^\sigma \preceq \dots  \preceq e_{4n-1}^\sigma \preceq e_{4n}^\sigma$.
    \item For each edge $e := \{\nou_j, u_j^{(k)}\} \in E_q \setminus E'$, we have $e_{\rho_\sigma(u_j)}^\sigma \preceq e$.
\end{enumerate}
Because $\sigma \in \omsh$, observe that the linear extensions of $T_\sigma$ are exactly the shellings of $H_{I, q}$ whose first edge is incident to $r$ and whose projection down to $E'$ coincides with $\sigma$. Thus, by~\Cref{lemma:tree-le}, we have
\[
\#\mathsf{Shellings}(H_{I, q}; r)_{\sigma} = (4n+3hq)! \prod_{i=1}^{4n}\frac{1}{i + q \cdot C_{i, \sigma}},
\]
where $C_{i,\sigma}$ is as in the statement of~\Cref{lem:shellings-formula}.
Notice that $C_{i, \sigma} \leq C_{j, \sigma}$ whenever $i < j$.

\subsection{Proof of Lemma \ref{lem:shellings-compute-polynomials}}
By the same argument as in~\Cref{lem:compute-polynomials}, we have that
\(
\sum_{\sigma \in \omsh} \chi_{\sigma}(q)
\)
is a rational function \(\frac{R(q)}{D(q)}\), where \(R(q)\) and \(D(q)\) are polynomials in canonical form that can be computed in polynomial time by evaluating 
$\frac{\#\mathsf{Shellings}(H_{I,q};r)}{(4n+3hq)!}$ at polynomially many positive integer values of $q$.

\subsection{Proof of Lemma \ref{lem:shellings-reduction}}
\label{subsec:claim_shellings}
Let us say that an ordering $\sigma \in \omsh$ \emph{induces} an exact cover $X \subseteq \mathcal{S}$ if the first $4h$ edges according to $\sigma$ correspond in some order to:
all edges $\{r, \ns_i\}$ for $S_i \in X$, and all edges $\{\ns_i, \nou_j\}$ for $S_i \in X$, $u_j \in S_i$. As in the proof of~\Cref{lem:reduction}, let $\omxc \subseteq \omsh$ be the subset of orderings that induce an exact cover. We will need the following claim regarding such orderings.

\begin{claim}\label{claim:cover_characterization}
Given any $\sigma \in \omsh$, the following statements hold:
\begin{enumerate}
\item If $\sigma \in \omxc$, then $C_{4n-4h, \sigma} = 0$ and $C_{4n-4h+1, \sigma} > 0$.
\item If $\sigma \notin \omxc$, then $C_{4n-4h, \sigma} > 0$.
\end{enumerate}
\end{claim}
\begin{proof}[Proof of~\Cref{claim:cover_characterization}]
In this proof, we abbreviate $e_i \coloneqq e_i^\sigma$ for every $i \in [4n]$.

First, suppose that $\sigma$ induces an exact cover $X \subseteq \mathcal{S}$. Then, by definition, the subgraph of $H_{I, q}$ induced by the first $4h$ edges of $\sigma$ already contains all vertices $\nou_j$, and therefore $C_{4n-4h, \sigma} = 0$.
Since $X$ is an exact cover, the subgraph of $H_{I, q}$ induced by the vertices
    \(
    \{ r \} \cup \{ \ns_i\ \mid S_i \in X\} \cup \{\nou_j \mid u_j \in U \}
    \)
is a tree. Therefore, since $\sigma$ is a shelling of the subgraph induced by the edges in $E'$, the edge $e_{4h}$ must be incident to a leaf of that tree. Since $e_1$ is incident to the root $r$, we necessarily have that $e_{4h}$ is of the form $\{\ns_i, \nou_j\}$ for $S_i \in X$ and $u_j \in S_i$. This shows that $C_{4n-4h+1, \sigma} > 0$.

For the second part, we prove the contrapositive. Suppose that $\sigma \in \omsh$ is such that $C_{4n-4h, \sigma} = 0$. Then, the subgraph $H_{\sigma}$ of $H_{I, q}$ induced by the edges $e_1, \dots, e_{4h}$ must contain all vertices $\nou_j$ for $u_j \in U$. Now let $Y$ be the set of vertices of the form $\ns_i$ belonging to $H_\sigma$. Because $H_\sigma$ is bipartite and $Y$ is one of the parts, we have that
\begin{equation}\label{eq:degrees}
      4h = \sum_{v \in Y}{\deg_{H_{\sigma}}(v)} \leq \sum_{v \in Y}{\deg_{H_{I, q}}(v)} = \sum_{v \in Y}{4} = 4|Y|,
      \end{equation}
and hence $|Y| \geq h$. On the other hand, since $\sigma$ is a shelling of the subgraph induced by the edges in $E'$ and its first edge is incident to $r$, we have that $H_\sigma$ is connected and contains $r$, and therefore $1+|Y|+3h$, its number of vertices, must be at most $4h+1$, which implies that $|Y| = h$. Equality in~\eqref{eq:degrees} implies that
$\deg_{H_\sigma}(v)=4$ for every $v\in Y$. Thus, $H_\sigma$ contains all four edges incident to each vertex of $Y$. The corresponding $h$ sets cover all $3h$
elements of $U$, so they form an exact cover $X$. Since $H_\sigma$ has exactly $4h=4|Y|$ edges, these are precisely the edges incident to $Y$; hence $\sigma$ induces $X$, that is, $\sigma\in\omxc$.
\end{proof} 

Now, the first part of \Cref{claim:cover_characterization} implies that, for any $\sigma \in \omxc $, we have $C_{i, \sigma} = 0$ if and only if $i \leq 4n - 4h$. Therefore, for any $\sigma \in \omxc$, we have
\begin{align*}
\lim_{q \to \infty} q^{4h} \chi_{\sigma}(q) &= \lim_{q \to \infty} q^{4h} \prod_{i=1}^{4n} \frac{1}{i + q C_{i, \sigma}} \\ &= \frac{1}{(4n-4h)!} \cdot \prod_{i=4n-4h+1}^{4n}\lim_{q \to \infty} \frac{q}{i+qC_{i, \sigma}}\\
&= \frac{1}{(4n-4h)!} \cdot \prod_{i=4n-4h+1}^{4n}\frac{1}{C_{i,\sigma}}.
\end{align*}

On the other hand, by the second part of \Cref{claim:cover_characterization}, for any $\sigma \in \omsh \setminus \omxc$, we have
\begin{align*}
    0 \leq \lim_{q \to \infty} q^{4h} \chi_{\sigma}(q) &= \lim_{q \to \infty} q^{4h} \left(\prod_{i = 1}^{4n-4h-1} \frac{1}{i + q C_{i, \sigma}}\right) \cdot \left(\prod_{i = 4n-4h}^{4n} \frac{1}{i + q C_{i, \sigma}}\right)\\
    &\leq \lim_{q \to \infty} q^{4h} \left(\prod_{i = 1}^{4n-4h-1} \frac{1}{i}\right) \cdot \left(\prod_{i = 4n-4h}^{4n} \frac{1}{i + q C_{i, \sigma}}\right) \\
    &= \frac{1}{(4n-4h-1)!} \lim_{q \to \infty} q^{-1}  \left(\prod_{i = 4n-4h}^{4n} \frac{q}{i + q C_{i, \sigma}}\right) \\
    &= \frac{1}{(4n-4h-1)!} \cdot \left(\prod_{i = 4n-4h}^{4n} \frac{1}{C_{i, \sigma}}\right) \cdot \lim_{q \to \infty} q^{-1} \\
    &= 0.
\end{align*}

To finish the proof we will rely on the following claim, whose proof uses a probabilistic argument.
\begin{claim}\label{claim:omega-xc}
Fix an exact cover $X \subseteq \mathcal{S}$, and let $\omxc(X) \subseteq \omxc$ be the set of orderings that induce the exact cover $X$. Then,
    \[
\sum_{\sigma \in\omxc(X)} \prod_{t=1}^{4h} \frac{1}{C_{4n-4h+t, \sigma}} = (4n-4h)! \cdot 3^{-h}.
\]
\end{claim}

\begin{proof}[Proof of~\Cref{claim:omega-xc}]
\input{sections/appendix-shellings}

\end{proof}
Using \Cref{claim:omega-xc}, we obtain the following:
\begin{align*}
   \lim_{q \to \infty} q^{4h} \sum_{\sigma \in \omsh} \chi_{\sigma}(q)
    &=  \lim_{q \to \infty}\sum_{\sigma \in \omxc} q^{4h} \chi_{\sigma}(q) + \lim_{q \to \infty}\sum_{\sigma \in \omsh \setminus \omxc} q^{4h} \chi_{\sigma}(q)\\
    &= \lim_{q \to \infty}\sum_{\sigma \in \omxc} q^{4h} \chi_{\sigma}(q)\\
    &= \lim_{q \to \infty} \sum_{X \in \mathsf{X3C}(I)} \sum_{\sigma \in \omxc(X)} q^{4h} \chi_{\sigma}(q)\\
    &= \sum_{X \in \mathsf{X3C}(I)} \sum_{\sigma \in \omxc(X)} \frac{1}{(4n-4h)!} \prod_{i=4n-4h+1}^{4n}\frac{1}{C_{i,\sigma}}\\
    &= \frac{1}{(4n-4h)!}\sum_{X \in \mathsf{X3C}(I)} \sum_{\sigma \in \omxc(X)} \prod_{t=1}^{4h}\frac{1}{C_{4n-4h+t,\sigma}}\\
    &= \frac{1}{(4n-4h)!}\sum_{X \in \mathsf{X3C}(I)} (4n-4h)! \cdot 3^{-h} \tag{\Cref{claim:omega-xc}}\\
    &= 3^{-h} \cdot \#\mathsf{X3C}(I).
\end{align*}

\subsection{Proof of Theorem \ref{thm:shellings2}}
Let $I$ be an arbitrary instance of $\#\mathsf{X3C}$, and let $R$ and $D$ be the polynomials given by~\Cref{lem:shellings-compute-polynomials}. Then the limit
\[
\lim_{q\to\infty} q^{4h}\sum_{\sigma \in \omsh} \chi_{\sigma}(q) \ = \
\lim_{q\to\infty} q^{4h}\frac{R(q)}{D(q)}
\]
can be computed in polynomial time via an oracle for counting rooted shellings: we compute $R$ and $D$ and then evaluate the limit symbolically, as in the proof of~\Cref{thm:main2}. All oracle queries are made on the graphs $H_{I,q}$, which are bipartite for every positive integer $q$. It follows from~\Cref{lem:shellings-reduction} that \(\#\mathsf{X3C}(I)\) can also be computed in polynomial time by multiplying this limit by \(3^h\). This concludes the proof of the theorem.

\subsection{Proof of Proposition \ref{prop:unrooting_shellings}}
The reduction is essentially the same as the one presented in the proof of~\Cref{prop:unrooting_SVO}. Let $M$ be the number of edges of $G$. If $r$ is isolated, then $\#\mathsf{Shellings}(G; r)=0$,
so we may assume that $r$ is incident to an edge. We construct another bipartite graph $H_\ell$ by adding $\ell$ additional vertices to $G$, called the \emph{leaves} of $H_\ell$, all of which are exclusively connected to $r$. We refer to the $\ell$ new edges created by this gadget as the \emph{pendant edges} of $H_\ell$. There are only two kinds of shellings for $H_\ell$: those that start with a pendant edge and those that do not. For the first kind, we claim that their number is
    \[
    \ell \cdot \# \mathsf{Shellings}(G; r) \cdot \binom{\ell+M-1}{\ell-1} (\ell-1)!.
    \]
Indeed, there are $\ell$ options for the initial pendant edge. The remaining ordering, when projected down to $E(G)$, must be a shelling of $G$ whose first edge is incident to $r$, and the $\ell-1$ remaining pendant edges can be inserted in any of the remaining $\ell+M-1$ positions, in any of their $(\ell-1)!$ possible orderings.

For the second kind, let $\mathsf{Shellings}(G)$ be the set of shellings of $G$, and, for $\pi \in \mathsf{Shellings}(G)$, let $\rho_\pi$ be the first position $i \in [M]$ at which $r$ is an endpoint of the $i$-th edge of the shelling $\pi$. Then, $\pi$ can be extended to a shelling of $H_\ell$ by adding the $\ell$ pendant edges in any of the $\ell+M - \rho_\pi$ positions after $\rho_\pi$, in any order. Therefore, the second kind is counted by
    \[
    \sum_{\pi \in \mathsf{Shellings}(G)} \binom{\ell+M - \rho_\pi}{\ell} \ell!.
    \]
Separating the case $\rho_\pi = 1$, we obtain
    \[
    \sum_{\pi \in \mathsf{Shellings}(G)} \binom{\ell+M - \rho_\pi}{\ell} \ell! = \# \mathsf{Shellings}(G; r) \binom{\ell+M-1}{\ell} \ell!+ \sum_{\substack{\pi \in \mathsf{Shellings}(G)\\ \rho_\pi > 1}} \binom{\ell+M - \rho_\pi}{\ell} \ell!.
    \]
Now we recall the identity
$$\binom{\ell+M-1}{\ell-1} +  \binom{\ell+M-1}{\ell} = \binom{\ell+M}{\ell},$$
and we use it to obtain that
    \[
    \frac{\#\mathsf{Shellings}(H_\ell)}{\ell!} = \# \mathsf{Shellings}(G; r) \binom{\ell+M}{\ell} + \sum_{\substack{\pi \in \mathsf{Shellings}(G)\\ \rho_\pi > 1}} \binom{\ell+M - \rho_\pi}{\ell}.
    \]

Now, observe that the right side of the previous equation is a polynomial in $\ell$, whose degree is at most $M$. Therefore, by evaluating 
    \[
    \frac{\#\mathsf{Shellings}(H_1)}{1!}, \frac{\#\mathsf{Shellings}(H_2)}{2!}, \dots, \frac{\#\mathsf{Shellings}(H_{M+1})}{(M+1)!}
    \]
    via an oracle for $\#\mathsf{Shellings}(\cdot)$, and using interpolation, we recover the coefficients of said polynomial. But notice that the coefficient of $\ell^M$ is precisely $\frac{1}{M!}\#\mathsf{Shellings}(G; r)$. Therefore, we have shown that polynomially many queries to an oracle for $\#\mathsf{Shellings}(\cdot)$ allow us to compute $\#\mathsf{Shellings}(G; r)$.

%% file: sections/appendix-shellings.tex
For ease of notation, let $X = \{S_1, \dots, S_h\}$. Then, let $E_X$ be the set of edges involving $X$:
    \[
    E_X := \{ \{r, \ns_i\} \mid S_i \in X\} \cup \{ \{\ns_i, \nou_j\} \mid S_i \in X, u_j \in S_i\}.
    \]
    Note that $|E_X| = 4h$. Let us write $\Gamma_X$ to denote the set of shellings rooted at $r$ of the graph induced by the edges $E_X$. Observe that each $\sigma \in \omxc(X)$ consists of an element of $\Gamma_X$ followed by some ordering of the edges $E' \setminus E_X$. Conversely, given any shelling $\tau \in \Gamma_X$, we can follow it by any of the $(4n-4h)!$ possible orderings of the edges $E' \setminus E_X$ and obtain a valid $\sigma \in \omxc(X)$. Note also that, given $\sigma \in \omxc(X)$, the terms $C_{4n-4h+t, \sigma}$ for $t \in [4h]$ only depend on the first $4h$ edges of $\sigma$. We will thus use the notation $C_{4n-4h+t, \tau}$, where $\tau \in \Gamma_X$. Then,
    \[
    \sum_{\sigma \in\omxc(X)} \prod_{t=1}^{4h} \frac{1}{C_{4n-4h+t, \sigma}} = (4n-4h)! \sum_{\tau \in\Gamma_X} \prod_{t=1}^{4h} \frac{1}{C_{4n-4h+t, \tau}}.
    \]
    Therefore, our claim reduces to proving that
      \begin{equation}\label{eq:claim-core}
          \sum_{\tau \in \Gamma_X} \prod_{t=1}^{4h} \frac{1}{C_{4n-4h+t, \tau}} = 3^{-h}.
      \end{equation}
      We prove~equation~\eqref{eq:claim-core} through a probabilistic interpretation.

      For every $i \in \{1, \dots, h\}$, let $r_i$ be the edge $\{r, \ns_i\}$. For each set $S_i \in X$, let $u_i^1, u_i^2, u_i^3$ be an enumeration of its $3$ elements. Then, for $1 \leq i \leq h$ and $1 \leq j \leq 3$ let $g_{i}^j$ denote the edge $\{\ns_i, \nou_i^j\}$, where $\nou_i^j$ is the vertex corresponding to the element $u_i^j$.
      Observe that $\tau$ is simply an ordering of the $4h$ edges $r_1, \dots, r_h, g_1^1, \dots, g_h^3$ where each $g_{i}^j$ must appear after $r_i$. 
    Now, consider the following random process to generate an ordering of $E_X$:
    \begin{enumerate}
    
        \item Initialize an empty sequence $\pi_0$, a counter $Q_0 \gets 3h$, a set $A_0 \gets \varnothing$, and $R_0 \gets \{r_1, \ldots, r_h\}$.
        \item Repeat the following for $t \in \{1, \dots, 4h\}$: 
        \begin{enumerate}
            \item Sample an element $e_t$ from $R_{t-1} \cup A_{t-1}$ according to the following distribution:
            \[
            \Pr[e_t = x] = \begin{cases}
                3/Q_{t-1} & \text{if } x \in R_{t-1}\\
                1/Q_{t-1} & \text{if } x \in A_{t-1}.
            \end{cases}
            \]
            \item Obtain $\pi_t$ by appending $e_t$ to the end of $\pi_{t-1}$.
            \item If $e_t = r_i \in R_{t-1}$ for some $i$, then set\\ $R_t \gets R_{t-1} \setminus \{e_{t}\}$,  $Q_t \gets Q_{t-1}$, and $A_t \gets A_{t-1} \cup \{ g_i^1, g_i^2, g_i^3\}$. 
            \item If $e_t \in A_{t-1}$, then set  $R_t \gets R_{t-1}$,  $Q_{t} \gets Q_{t-1} - 1$, and $A_t \gets A_{t-1} \setminus \{e_t\}$.
        \end{enumerate}
        \item Output $\pi_{4h}$.
    \end{enumerate}
    Let us first note that step (2.a) is well defined due to the following invariant: for every $t$, we have that $3|R_t| + |A_t| = Q_t$, which holds directly for $t = 0$, and then inductively for $t \geq 1$. If $e_t \in R_{t-1}$, we have 
    \[
    3|R_t| + |A_t| = 3(|R_{t-1}|- 1) + |A_{t-1}| + 3 \overset{\mathrm{I.H.}}{=} Q_{t-1} = Q_t.
    \]
    On the other hand, if $e_t \in A_{t-1}$, we have
    \[
    3|R_t| + |A_t| = 3|R_{t-1}| + |A_{t-1}| - 1 \overset{\mathrm{I.H.}}{=} Q_{t-1} - 1 = Q_t.
    \]
    As the invariant holds, we conclude indeed 
    \[
    \sum_{x \in R_{t-1} \cup A_{t-1}} \Pr[e_t = x] = |R_{t-1}| \cdot \frac{3}{Q_{t-1}} + |A_{t-1}| \cdot \frac{1}{Q_{t-1}} = \frac{3|R_{t-1}| + |A_{t-1}|}{Q_{t-1}} = 1.
    \]
    Note as well that the previous process guarantees that each $g_i^j$ appears after $r_i$ in $\pi_{4h}$, since $g_i^j$ is only added to $A_{t}$ once $r_i$ has been added to $\pi_t$. Thus,  if we let $\mathcal{D}$ denote the previous random process, then the support of $\mathcal{D}$ is exactly $\Gamma_X$ and therefore
    \begin{equation}\label{eq:distr}
    \sum_{\tau \in \Gamma_X} \Pr_{\mathcal{D}}[\pi_{4h} = \tau] = 1.
      \end{equation}

    Using equation~\eqref{eq:distr}, we immediately  get~equation~\eqref{eq:claim-core} if we  prove  that for every $\tau \in \Gamma_X$,
    \begin{equation}\label{eq:prob}
        \Pr_{\mathcal{D}}[\pi_{4h} = \tau] = 3^h \prod_{t=1}^{4h} \frac{1}{C_{4n-4h+t, \tau}}.
    \end{equation}

    Thus, our proof reduces to showing~equation~\eqref{eq:prob}.
    First, we have
    \[
     \Pr_{\mathcal{D}}[\pi_{4h} = \tau] =  \prod_{t=1}^{4h} \Pr_{\mathcal{D}}\left[\pi_{4h}(t) = \tau(t) \, \Big|\, \bigwedge_{j=1}^{t-1 }\pi_{4h}(j) = \tau(j)\right].
    \]
Observe that, for $\tau \in \Gamma_X$ and $t \in \{1, \dots, 4h\}$, the value $C_{4n+1-t, \tau}$ only depends on the first $t-1$ edges of $\tau$. We now claim that, if the first $t \geq 0$ steps of the random process coincide with $\tau \in \Gamma_X$, then it must be the case that $Q_{t} = C_{4n-t, \tau}$. Let $I_\tau \subseteq \{1, \dots, 4h\}$ be the set of positions $t$ for which $\tau(t)$ is an edge of the form $\{\ns_i, \nou_i^j\}$, with $S_i \in X$ and $u_i^j \in S_i$. Note that, since $X$ is an exact cover, for each $u_j \in U$ there is exactly one edge of $E_X$ containing $\nou_j$, and hence
$C_{4n-t, \tau} = |I_\tau \cap \{t+1, \dots, 4h\}|$. Now observe that initially $Q_0 = 3h = |I_\tau|$, and we set $Q_{t} \gets Q_{t-1} - 1$ exactly when $t \in I_\tau$. Therefore, our claim follows directly by induction.

Now observe that if $t \in I_\tau$, then we have
     \[
     \Pr_{\mathcal{D}}\left[\pi_{4h}(t) = \tau(t) \, \Big|\, \bigwedge_{j=1}^{t-1 }\pi_{4h}(j) = \tau(j)\right] = \frac{1}{Q_{t-1}} = \frac{1}{C_{4n+1-t, \tau}},
   \]
   and if $t \not\in I_\tau$, then 
     \[
     \Pr_{\mathcal{D}}\left[\pi_{4h}(t) = \tau(t) \, \Big|\, \bigwedge_{j=1}^{t-1 }\pi_{4h}(j) = \tau(j)\right] = \frac{3}{Q_{t-1}} = \frac{3}{C_{4n+1-t, \tau}}.
   \]
   Because $\left|\{1, \dots, 4h\} \setminus I_\tau \right| = h$, we have that
    \[
     \Pr_{\mathcal{D}}[\pi_{4h} = \tau] =  \prod_{t=1}^{4h} \Pr_{\mathcal{D}}\left[\pi_{4h}(t) = \tau(t) \, \Big|\, \bigwedge_{j=1}^{t-1 }\pi_{4h}(j) = \tau(j)\right] = 3^{h} \prod_{t=1}^{4h} \frac{1}{C_{4n-4h+t, \tau}},
    \]
which is exactly~\eqref{eq:prob}. We thus conclude the proof.